\documentclass[11pt,letterpaper]{article}
\usepackage{standalone}

\usepackage[letterpaper,margin=1in]{geometry}
\usepackage{amsthm}
\usepackage{amsmath}
\usepackage{amssymb}
\usepackage{thmtools}
\usepackage{thm-restate}
\usepackage{times}
\usepackage{mathtools}

\usepackage[pdfpagelabels,bookmarks=false,hyperfootnotes=false]{hyperref}
\hypersetup{colorlinks, linkcolor=darkblue, citecolor=darkgreen, urlcolor=darkblue}

\DeclareFontFamily{U}{BOONDOX-calo}{\skewchar\font=45 }
\DeclareFontShape{U}{BOONDOX-calo}{m}{n}{
  <-> s*[1.05] BOONDOX-r-calo}{}
\DeclareFontShape{U}{BOONDOX-calo}{b}{n}{
  <-> s*[1.05] BOONDOX-b-calo}{}
\DeclareMathAlphabet{\link}{U}{BOONDOX-calo}{m}{n}
\DeclareMathAlphabet{\blink}{U}{BOONDOX-calo}{b}{n}

\usepackage[utf8]{inputenc}

\usepackage[
backend=biber,
style=alphabetic,
citestyle=alphabetic,
maxalphanames=4,
maxcitenames=99,
mincitenames=98,
maxbibnames=99,
giveninits=true,
]{biblatex}

\usepackage{xpatch}

\makeatletter
\patchcmd\blx@bblinput{\blx@blxinit}
                      {\blx@blxinit
                      }{}{\fail}
\makeatother

\usepackage[capitalize,nameinlink]{cleveref}

\usepackage{lmodern}

\usepackage{bbm}
\usepackage{thm-restate}

\newtheoremstyle{light} {\topsep}                    {\topsep}                    {\itshape}                   {}                           {\scshape}                   {.}                          {.5em}                       {}  

\newtheorem{theorem}{Theorem}[section]
\newtheorem{lemma}[theorem]{Lemma}

\newtheorem{definition}[theorem]{Definition}

\newtheorem{corollary}[theorem]{Corollary}

\theoremstyle{light}
\newtheorem{claiminproof}[theorem]{Claim}
\makeatletter
\if@cref@capitalise
\crefname{claiminproof}{Claim}{Claims}
\else
\crefname{claiminproof}{claim}{claims}
\fi

\if@cref@capitalise
\crefname{algocf}{Algorithm}{Algorithms}
\else
\crefname{claiminproof}{algorithm}{algorithm}
\fi
\makeatother

\usepackage{url}
\usepackage[table]{xcolor}
\usepackage{array}

\usepackage{setspace}

\usepackage{wrapfig}

\makeatletter
\newcommand{\labeltarget}[1]{\Hy@raisedlink{\hypertarget{#1}{}}}
\makeatother

\usepackage{upgreek}

\usepackage{complexity}

\usepackage[inline]{enumitem}
\usepackage{moreenum}

\usepackage[super]{nth}

\usepackage{bm}

\usepackage{xspace}

\usepackage{ifthen}

\usepackage[immediate]{silence}
\usepackage{sectsty}
\DeactivateWarningFilters[tmp]
\allsectionsfont{\boldmath}

\setlist[enumerate]{nosep,topsep=0.1em}
\setlist[enumerate,1]{label=(\roman*), leftmargin=2.2em}
\setlist[itemize]{nosep,topsep=0.3em}

\usepackage[textsize=footnotesize, color=blue!30!white]{todonotes}
\usepackage{xfrac}

\usepackage{tikz}
\usetikzlibrary{calc}
\usetikzlibrary{math}
\usetikzlibrary{shapes.geometric}
\usetikzlibrary{arrows}
\usetikzlibrary{decorations.pathreplacing, decorations.pathmorphing}

\usepackage{tcolorbox}
\tcbuselibrary{skins}

\usepackage{float}
\usepackage{graphicx}
\graphicspath{{graphics/}}
\makeatletter
\newcommand\appendtographicspath[1]{\g@addto@macro\Ginput@path{#1}}
\makeatother

\usepackage[margin=10pt,font=small,labelfont=bf,skip=-5pt]{caption}
\usepackage{subcaption}

\usepackage[vlined,ruled,algo2e]{algorithm2e}
\SetAlgoSkip{bigskip}
\SetAlgoInsideSkip{smallskip}

\usepackage{mdframed}

\let\truehypersetup\hypersetup
\renewcommand\hypersetup[1]{}
\usepackage{bigfoot}
\let\hypersetup\truehypersetup
\definecolor{darkblue}{rgb}{0,0,0.38}
\definecolor{darkred}{rgb}{0.8,0,0}
\definecolor{darkgreen}{rgb}{0.1,0.35,0}

\DeclareMathOperator{\argmin}{argmin}
\DeclareMathOperator{\slack}{slack}
\DeclareMathOperator{\lca}{lca}
\DeclareMathOperator{\leftp}{left}
\DeclareMathOperator{\rightp}{right}

\newcommand\OPT{\ensuremath{\mathrm{OPT}}}
\newcommand\Drop{\ensuremath{\mathrm{Drop}}}
\newcommand\shadows{\ensuremath{\mathrm{shadows}}}
\renewcommand{\epsilon}{\varepsilon}

\def\cupp{\stackrel{.}{\cup}}
\def\bigcupp{\stackrel{.}{\bigcup}}

\makeatletter
\let\@@pmod\pmod
\DeclareRobustCommand{\pmod}{\@ifstar\@pmods\@@pmod}
\def\@pmods#1{\mkern8mu({\operator@font mod}\mkern 6mu#1)}
\makeatother

\makeatletter
\let\@@mod\mod
\DeclareRobustCommand{\mod}{\@ifstar\@mods\@@mod}
\def\@mods#1{\mkern8mu{\operator@font mod}\mkern 6mu#1}
\makeatother

\def\Cscr{\mathcal{C}}

\def\Kscr{\mathcal{K}}
\def\Lscr{\mathcal{L}}

\def\Rscr{\mathcal{R}}

\def\Xscr{\mathcal{X}}
 
\makeatletter
\def\@fnsymbol#1{\ensuremath{\ifcase#1\or *\or \ddagger\or
    \mathsection\or \mathparagraph\or \|\or **\or \dagger\dagger
    \or \ddagger\ddagger \else\@ctrerr\fi}}
\makeatother

\title{A $(1.5+\epsilon)$-Approximation Algorithm for Weighted Connectivity Augmentation\thanks{This project received funding from Swiss National Science Foundation grant 200021\_184622 and the European Research Council (ERC) under the European Union's Horizon 2020 research and innovation programme (grant agreement No 817750).}
}

\author{
Vera Traub\thanks{
Research Institute for Discrete Mathematics and Hausdorff Center for Mathematics, University of Bonn.
Part of this work was done at ETH Zurich.
Email: \href{mailto:traub@dm.uni-bonn.de}{traub@dm.uni-bonn.de}.
}
\and
Rico Zenklusen\thanks{
Department of Mathematics, ETH Zurich, Zurich, Switzerland.
Email: \href{mailto:ricoz@ethz.ch}{ricoz@ethz.ch}.}
}

\date{}

\begin{document}

\maketitle
\thispagestyle{empty}
\addtocounter{page}{-1}
\enlargethispage{-1cm}

\begin{abstract}
Connectivity augmentation problems are among the most elementary questions in Network Design.
Many of these problems admit natural $2$-approximation algorithms, often through various classic techniques, whereas it remains open whether approximation factors below $2$ can be achieved.
One of the most basic examples thereof is the Weighted Connectivity Augmentation Problem (WCAP).
In WCAP, one is given an undirected graph together with a set of additional weighted candidate edges, and the task is to find a cheapest set of candidate edges whose addition to the graph increases its edge-connectivity.
We present a $(1.5+\varepsilon)$-approximation algorithm for WCAP, showing for the first time that factors below $2$ are achievable.

On a high level, we design a well-chosen local search algorithm, inspired by recent advances for Weighted Tree Augmentation.
To measure progress, we consider a directed weakening of WCAP and show that it has highly structured planar solutions.
Interpreting a solution of the original problem as one of this directed weakening allows us to describe local exchange steps in a clean and algorithmically amenable way.
Leveraging these insights, we show that we can efficiently search for good exchange steps within a component class for link sets that is closely related to bounded treewidth subgraphs of circle graphs.
Moreover, we prove that an optimum solution can be decomposed into smaller components, at least one of which leads to a good local search step as long as we did not yet achieve the claimed approximation guarantee.
\end{abstract}

\begin{tikzpicture}[overlay, remember picture, shift = {(current page.south east)}]
\begin{scope}[shift={(-1.1,2.5)}]
\def\hd{2.5}
\node at (-2.15*\hd,0) {\includegraphics[height=0.7cm]{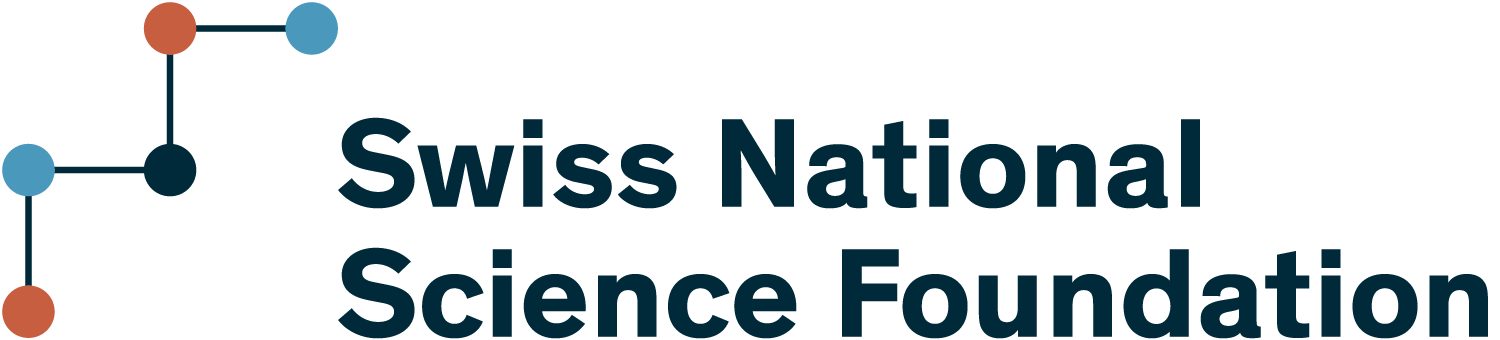}};
\node at (-\hd,0) {\includegraphics[height=1.0cm]{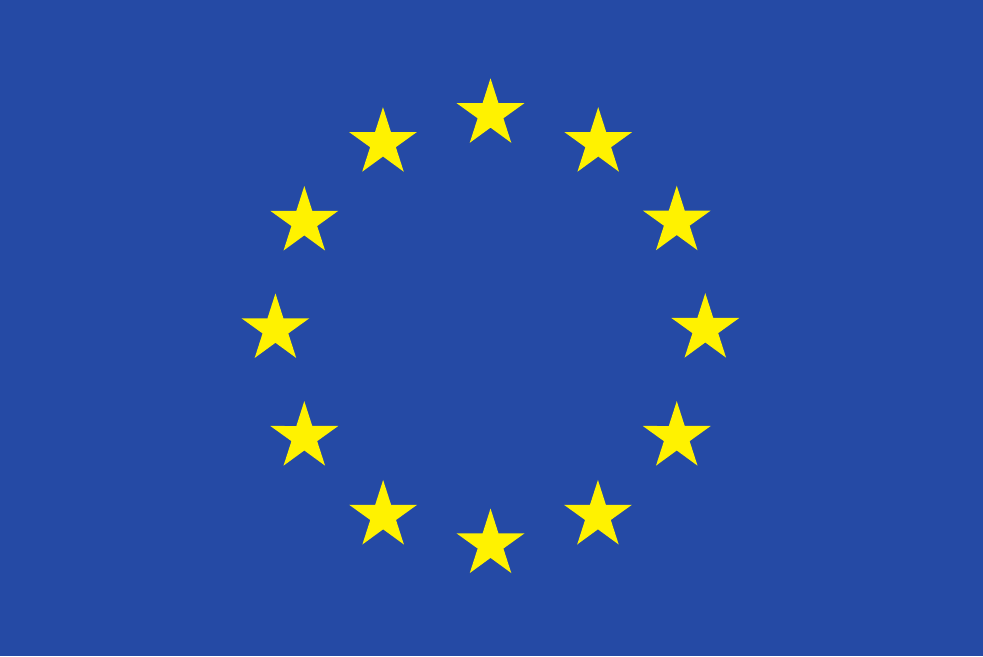}};
\node at (-0.2*\hd,0) {\includegraphics[height=1.2cm]{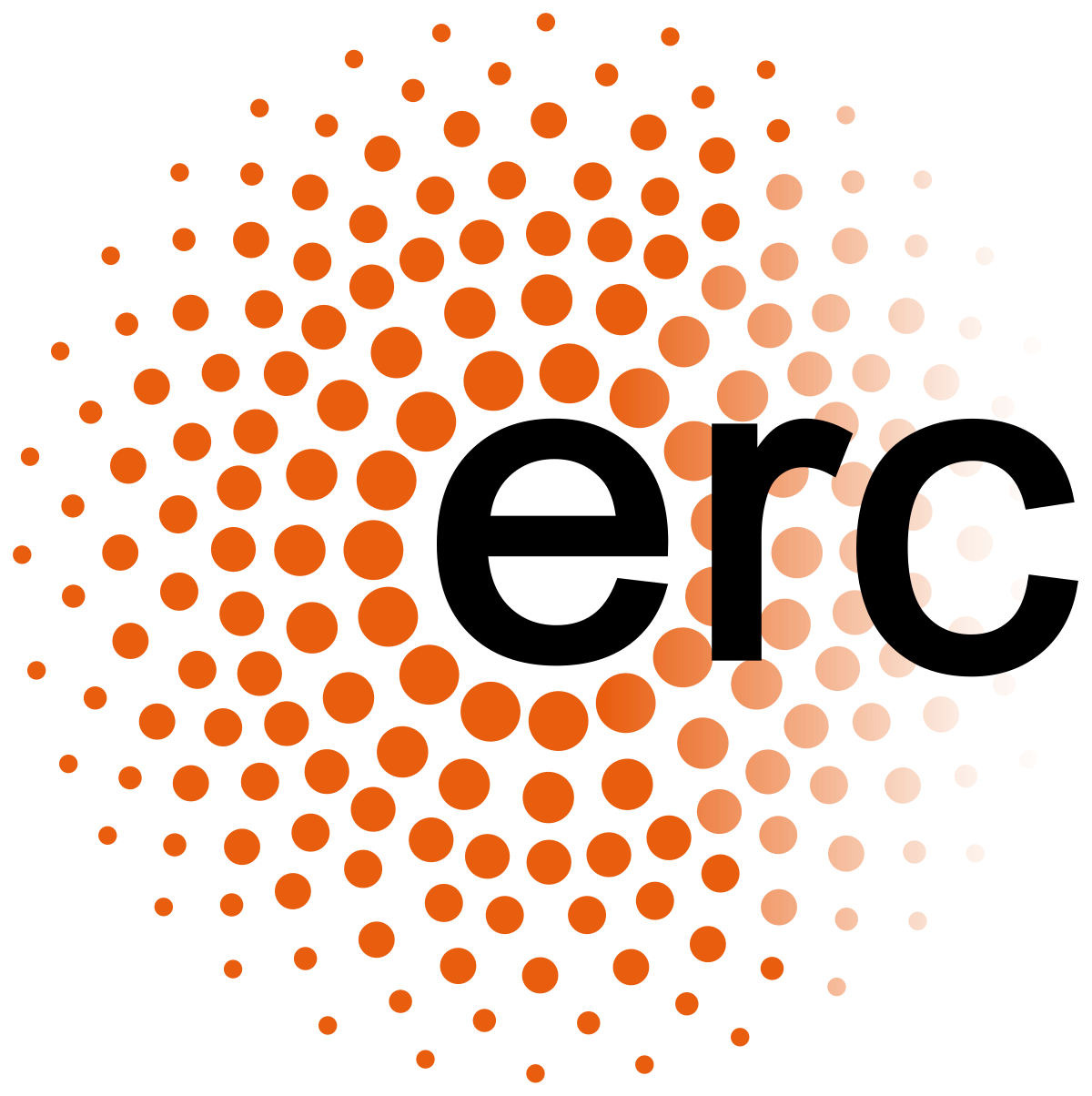}};
\end{scope}
\end{tikzpicture}

\clearpage

\section{Introduction}\label{sec:intro}

The edge-connectivity of a graph is a key network parameter and a large number of classical and modern network design problems revolve around it.
Prominent examples are the minimum spanning tree problem, the Steiner tree problem, creating cheap graphs with high edge-connectivity, or augmenting the edge-connectivity of a graph in the most economical way, which leads to the class of augmentation problems.
Most of the above network design problems are \APX-hard and have become prominent problems in the field of approximation algorithms.
Interestingly, despite extensive research, there remain very elementary problems in this class that have proved surprisingly hard to make any progress beyond some easy-to-obtain canonical approximation factors.
This holds in particular for one of the arguably most natural and central augmentation problems, namely the Weighted Connectivity Augmentation Problem (WCAP), which is the topic of this paper.
In WCAP, one is given an undirected graph $G=(V,E)$ together with a set $L \subseteq \begin{psmallmatrix} V\\ 2 \end{psmallmatrix}$ of candidate edges, also called links, which have non-negative costs $c\colon L \to \mathbb{R}_{\geq 0}$.
The task is to find a cheapest link set $S\subseteq L$ whose addition to $G$ will increase its edge-connectivity.
Formally, if we denote by $k$ the edge-connectivity of $G$, then the task is to find an (approximate) minimizer of
\begin{equation}
\min\left\{c(S) \colon S\subseteq T \text{ s.t.~}(V,E\cup S) \text{ is $(k+1)$-edge-connected}\right\}.\tag{WCAP}
\end{equation}
The problem is well-known to be \APX-hard.
This follows for example from \APX-hardness of the (unweighted) Tree Augmentation Problem (TAP), which is a heavily studied special case of WCAP where the underlying graph $G$ is a spanning tree and all links have unit costs.
Prior to this work, nothing better than $2$-approximations were known for WCAP.
Moreover, a $2$-approximation is arguably easy to obtain and many approaches are known to achieve this factor, including a simple reduction to a directed problem (see~\cite{khuller_1993_approximation,khuller_1994_biconnectivity} for a description of the reduction, where the directed problem can then be solved by a result of~\cite{frank_1989_application}) and by now well-established primal-dual and iterative rounding approaches~\cite{goemans_1994_improved,jain_2001_factor}.

So far, progress beyond the factor $2$ has only been achieved for special cases of WCAP.
This includes in particular the aforementioned heavily studied (Unweighted) Tree Augmentation Problem, where a long sequence of work~\cite{frederickson_1981_approximation,khuller_1993_approximation,nagamochi_2003_approximation,cheriyan_2008_integrality,even_2009_approximation,cohen_2013_approximation,kortsarz_2016_simplified,nutov_2017_tree,cheriyan_2018_approximating_a,cheriyan_2018_approximating_b,adjiashvili_2018_beating,fiorini_2018_approximating,grandoni_2018_improved,kortsarz_2018_lp-relaxations,}
led to an approximation factor of $1.393$~\cite{cecchetto_2021_bridging}.
This same factor can also be achieved for the unweighted Connectivity Augmentation Problem~\cite{cecchetto_2021_bridging}, which is a problem for which first better-than-$2$ approximation have only been obtained very recently~\cite{byrka_2020_breaching} (see also~\cite{nutov_2021_approximation}).
Also very recently, first approximation factors below $2$ have been obtained for the Weighted Tree Augmentation Problem~\cite{traub_2021_better,traub_2022_local} through the identification and analysis of a class of local improvement steps.

In this work, we aim at making the first progress on the approximability of WCAP below factor~$2$.

\subsection{Our results}
Our main result is the following.
\begin{theorem}\label{thm:main}
For every $\epsilon >0$, there is a polynomial-time $(1.5+\epsilon)$-approximation algorithm for the Weighted Connectivity Augmentation Problem.
\end{theorem}
Our result is inspired by recent progress for WTAP~\cite{traub_2021_better,traub_2022_local}, and, on a high level, follows a similar approach.
However, there are substantial new hurdles faced in connectivity augmentation that do not exist in tree augmentation. (The higher complexity of connectivity augmentation versus tree augmentation is also reflected in the unweighted case, where first progress on unweighted connectivity augmentation happened two decades after first progress on unweighted TAP.)
We expand on our technical contributions and differences to prior approaches in \cref{sec:overview}.

\subsection{Organization of paper}
We finish the introduction by highlighting some further related work in \cref{sec:furtherPrior}.
\cref{sec:overview} provides an overview of our approach, highlighting the main steps and providing a full high-level picture, with links to later sections for further details on specific steps.
The overview and \cref{sec:wrap,sec:directed-solutions,sec:drop,sec:thin-components,sec:decomposition-thm,sec:dp} provide all details to get a better-than-$2$ approximation for WCAP.
However, for a better exposition, in this part we show a weaker result than our main result, \cref{thm:main}, namely a $(1+\ln(2)+\epsilon)$-approximation.
In \cref{sec:local-search}, we then show how to obtain the claimed $(1.5+\epsilon)$-approximation through a further strengthening.

\subsection{Further related work}\label{sec:furtherPrior}

The body on related literature is vast.
A related problem class is to design cheap $k$-edge-connected subgraphs under different assumptions.
This includes the $k$-edge-connected spanning subgraph problem ($k$-ECSS).
In its unweighted version, the task is to find a $k$-edge-connected subgraph with a smallest number of edges that spans all vertices.
The currently strongest approximation algorithms for unweighted $2$-ECSS achieve an approximation factor of $\sfrac{4}{3}$~\cite{sebo_2014_shorter,hunkenschroder_2019_approximation}.
However, for the weighted counterpart of $2$-ECSS (and more generally weighted $k$-ECSS for $k \ge 2$), an approximation factor of $2$ \cite{khuller_1994_biconnectivity} remains the best known approximation guarantee.

Moreover, for unweighted $k$-ECSS, approximation factors approaching $1$ as $k$ grows can be achieved, even in directed settings (\cite{cheriyan_2000_approximating,gabow_2004_special,gabow_2009_approximating}).
The special case of $2$-ECSS with costs restricted to $0$ and $1$ leads to the Forest Augmentation Problem (FAP), for which a first better-than-$2$ approximation has been obtained very recently~\cite{grandoni_2022_breaching}, where a $1.9973$-approximation was presented.
Better approximations are known for the Matching Augmentation Problem (MAP), which is the special case of FAP where the $0$-weight edges form a matching~\cite{cheriyan_2020_improved,cheriyan_2020_matching,bamas_2022_simple}.

 \section{Overview of Techniques}\label{sec:overview}

The key of our approach is to define a particular type of very structured WCAP solutions, and then derive results on how such solutions can be improved under well-defined circumstances.
These structured WCAP solutions are actually solutions to an auxiliary \emph{directed} version of WCAP, based on directed links.
Our goal is to replace a (potentially large) subset of directed links with a (potentially large) subset of original, undirected links of smaller cost.
This leads to \emph{mixed solutions}, which contain both directed and undirected links.
We can then either continue our improvement steps by considering only the leftover directed links of the mixed solution, which leads to a relative greedy algorithm, or by reinterpreting the mixed solution as a directed one and then continue with improvement steps, which leads to a stronger local search procedure.

A central contribution of this paper is to present a way how this high-level idea can be realized. 
To this end, we introduce a sufficient and structure-rich condition of when a directed link can be dropped after adding a set of undirected links.
This allows us to use a carefully defined component class for undirected links within which we can, in particular, find a link set that maximizes the cost of directed links we can drop due to our sufficient condition minus the cost of the undirected links to be added.
Moreover, using the same component class, we can show through a decomposition theorem that improving replacement steps are possible under well-defined conditions.

Superficially, these last steps of identifying a proper component class, for which both an optimization algorithm and a structural decomposition theorem can be found, have already been followed in other local search and relative greedy contexts, including for the Weighted Tree Augmentation Problem.
Indeed, they are a well-known recipe.
The challenge lies in finding a way to realize them, and, on this level, we differ substantially from prior work in multiple ways, which we highlight in the remaining part of this overview when expanding on the various steps of our approach.

This overview is organized as follows.
We start in \cref{sec:ovReducRing} with the simple and well-known observation that the underlying graph of a WCAP instance can be assumed to be a ring/cycle.
In \cref{sec:ovDirSols}, we discuss the auxiliary directed version of WCAP that underlies our approach.
\cref{sec:ovReplByUndir} discusses under what circumstances we drop directed links after adding a set of undirected ones.
We achieve this by assigning to each directed link a family of $2$-cuts for which it is responsible.
This notion of responsibility, despite being a simplification, leads to various structural properties that are crucially exploited in later steps.
Such properties, which also give clean descriptions of when links can be dropped, are discussed in \cref{sec:ovWhenToDropLinks}.
In \cref{sec:ovThinComps}, we introduce the component class used in our procedures, and discuss its key properties, namely that we can efficiently solve relevant optimization problems over that class and that there exists a component of the class that allows for making progress (under well-defined conditions), which follows from a decomposition theorem.
Finally, \cref{sec:commentsOptThmDecThm} provides a brief discussion on some key aspects of our proofs of the optimization and decomposition result.

\subsection{Reduction to ring instances}\label{sec:ovReducRing}

We start with the well-known fact that WRAP can be reduced to the case where the underlying graph is a ring/cycle, which has been previously observed in~\cite{galvez_2019_cycle}.
This leads to the Weighted Ring Augmentation Problem (WRAP).
\begin{definition}[Weighted Ring Augmentation Problem (WRAP)]
A WCAP instance $(G,L,c)$ is an instance of the Weighted Ring Augmentation Problem if $G$ is a cycle.
\end{definition}
The following \lcnamecref{lem:reduction-to-WRAP} summarizes the reduction statement.
\begin{restatable}[]{lemma}{redToWRAP}
\label{lem:reduction-to-WRAP}
Let $\alpha > 1$.
If there is an $\alpha$-approximation algorithm for WRAP, then there is an $\alpha$-approximation algorithm for WCAP.
\end{restatable}

The idea behind the reduction is that one can first reduce to the graph being a cactus, i.e., an undirected graph where every edge lies in precisely one cycle, using the fact that minimum cuts of a graph can be represented by a cactus~\cite{dinitz_1976_structure}.
In a second step, one can ``unfold'' the cactus into a ring by adding appropriate links of weight zero.

\subsection{Structured solutions through directed WRAP}\label{sec:ovDirSols}

We now discuss the directed auxiliary problem that we use to obtain structured solutions.
Passing through directed settings to approximate undirected augmentation problems is a classical idea that has been used in related contexts (see, e.g., \cite{frederickson_1981_approximation,khuller_1993_approximation,khuller_1994_biconnectivity}).
In particular, very recently, the same directed problem we employ here has been used in the context of unweighted WCAP; however, for a very different purpose that relied on properties of the dual of a linear programming formulation of the used directed problem setting~\cite{cecchetto_2021_bridging}.

The directed version we use is defined with respect to an arbitrary root.
We therefore first introduce \emph{rooted WRAP} instances, which have a fixed root $r$ and, moreover, also a fixed edge $e_r$ among the two edges of the ring incident to $r$.
Fixing the edge $e_r$ will be convenient later, e.g., for introducing in an unambiguous way a left/right notion along the cycle.
\begin{definition}[Rooted Weighted Ring Augmentation Problem (rooted WRAP)]
An instance $(G=(V,E),L,c,r,e_r)$ of a \emph{rooted} Weighted Ring Augmentation Problem (rooted WRAP) consists of a WRAP instance $(G,L,c)$ together with a root $r\in V$ and an edge $e_r\in \delta_E(r)$.
\end{definition}

Here and in the following we denote by $\delta_F(U)$, for an edge/link set $F$ and a vertex set $U$, the set of edges/links in $F$ that have exactly one endpoint in $U$.
Given a rooted WRAP instance $(G,L,c,r,e_r)$, let
\begin{equation*}
\mathcal{C}_G\coloneqq \left\{C\subseteq V\setminus \{r\} \colon |\delta_E(C)| = 2\right\}
\end{equation*}
the set of all $2$-cuts of $G$, where we choose the $2$-cuts such that they do not contain the root.
Note that a set $S\subseteq L$ of links is a WRAP solution if and only if every cut $C\in \Cscr_G$ is covered by a link from $S$, i.e., we have $\delta_S(C) \neq \emptyset$ for all $C\in \Cscr_G$.

A directed link $\vec{\link{l}}$ in a rooted WRAP instance is a tuple $\vec{\link{l}}=(u,v)\in V\times V$ with $u \neq v$.
We think of $\vec{\link{l}}=(u,v)$ as covering all cuts of $\mathcal{C}_G$ that $\vec{\link{l}}$ is entering.
This is strictly weaker than an undirected link between $u$ and $v$, which also covers the cuts in $\mathcal{C}_G$ that $\vec{\link{l}}$ is leaving.
This notion of which $2$-cuts are covered by directed links readily leads to the definition below of a directed WRAP solution.
(See \cref{fig:dirWRAPSol}.)
\begin{definition}[Directed WRAP solution]
Let $\mathcal{I}=(G,L,c,r,e_r)$ be a rooted WRAP instance and let $\vec{F}\subseteq V\times V$ be a family of directed links.
We say that $\vec{F}$ is a solution to $\mathcal{I}$ if, for every $C\in \mathcal{C}_G$, there is a directed link $\vec{\link{l}}\in \vec{F}$ that enters $C$.
\end{definition}
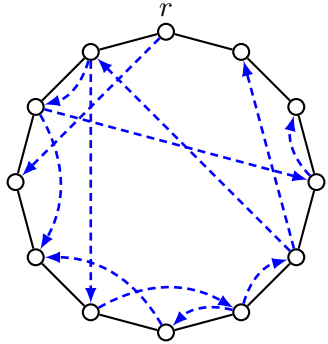
\begin{figure}[!ht]
\begin{center}
\begin{tikzpicture}[scale=1,
ns/.style={thick,draw=black,fill=white,circle,minimum size=6,inner sep=2pt},
es/.style={thick},
lks/.style={line width=1pt, blue, densely dashed},
dlks/.style={lks, -latex},
ts/.style={every node/.append style={font=\scriptsize}}
]

\def\rad{2}
\def\num{12}

\begin{scope}[every node/.style={ns}]
\foreach \i in {1,...,\num} {
  \pgfmathsetmacro\r{90+(\i-1)*360/\num}
  \node (\i) at (\r:\rad) {};
}
\end{scope}

\begin{scope}
\node at (1)[above=2pt] {$r$};
\end{scope}

\begin{scope}[es]
\foreach \i in {1,...,\num} {
\pgfmathtruncatemacro\j{1+mod(\i,\num)}
\draw (\i) -- (\j);
}
\end{scope}

\begin{scope}[dlks]
\draw (1) -- (4);
\draw (2) -- (6);
\draw (2) to[bend left] (3);
\draw (3) to[bend left] (5);
\draw (3) -- (10);
\draw (6) to[bend left] (8);
\draw (7) to[bend right] (5);
\draw (8) to[bend left] (9);
\draw (8) to[bend right] (7);
\draw (9) to (2);
\draw (9) -- (12);
\draw (10) to[bend left] (11);
\end{scope}

\end{tikzpicture}
 \end{center}
\caption{Example of a directed WRAP solution depicted as dashed blue arrows.}
\label{fig:dirWRAPSol}
\end{figure}

It turns out that directed WRAP instances can be solved efficiently.
An easy way to transform a WRAP instance into a directed one is by bidirecting each link $\link{l}\in L$ and assigning to the two directed links corresponding to $\link{l}$ the cost $c(\link{l})$.
By solving the resulting bidirected instance, we thus get a directed WRAP solution of cost at most $2c(\OPT)$, where $\OPT\subseteq L$ is an optimal undirected WRAP solution.
(This follows by observing that bidirecting $\OPT$ leads to a directed WRAP solution.)
To obtain highly structured directed WRAP solutions, we allow for shortening directed links $\vec{\link{l}}$, which means replacing $\vec{\link{l}}$ by an equal-cost directed link that covers a subset of the $2$-cuts in $\mathcal{C}_G$ covered by $\vec{\link{l}}$.
(See \cref{fig:ov_shortenings} for an example and \cref{sec:directed-solutions} for more details.)
\begin{figure}[!ht]
\begin{center}
\begin{tikzpicture}[scale=1,
ns/.style={thick,draw=black,fill=white,circle,minimum size=6,inner sep=2pt},
es/.style={thick},
lks/.style={line width=1pt, blue, densely dashed},
dlks/.style={lks, -latex},
ts/.style={every node/.append style={font=\scriptsize}}
]

\def\rad{2}
\def\num{12}

\begin{scope}[every node/.style={ns}]
\foreach \i in {1,...,\num} {
  \pgfmathsetmacro\r{90+(\i-1)*360/\num}
  \node (\i) at (\r:\rad) {};
}
\end{scope}

\begin{scope}
\node at (1)[above=2pt] {$r$};
\path (\num) to node[above=-2pt] {$e_r$} (1);
\end{scope}

\begin{scope}[es]
\foreach \i in {1,...,\num} {
\pgfmathtruncatemacro\j{1+mod(\i,\num)}
\draw (\i) -- (\j);
}
\end{scope}

\begin{scope}[dlks]

\colorlet{col_l1}{cyan!70!black}
\colorlet{col_l2}{blue!50!black}

\draw[col_l1] (1)  -- node[right=3pt,pos=0.3] {$\vec{\link{l}}_1$} (4);
\draw[col_l2] (10) -- node[above, pos=0.15] {$\vec{\link{f}}_1$} (6);

\begin{scope}[densely dotted]
\draw[col_l1] (2)  -- node[above left=0pt,pos=0.2] {$\vec{\link{l}}_2$} (4);
\draw[col_l2] (8)  -- node[above,pos=0.1] {$\vec{\link{f}}_2$} (6);
\end{scope}
\end{scope}

\end{tikzpicture}
 \end{center}
\caption{Examples of shortenings (dotted arrows) of directed links (dashed arrows).
The directed links $\vec{\link{l}_2}$ and $\vec{\link{f}_2}$ are shortenings of the directed links $\vec{\link{l}_1}$ and $\vec{\link{f}_1}$, respectively.}
\label{fig:ov_shortenings}
\end{figure}
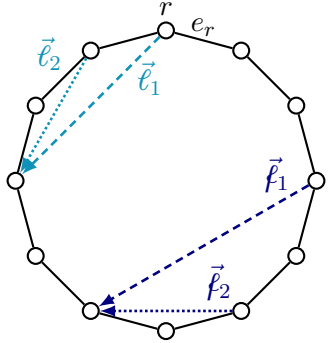
We call a directed and possibly shortened version of a link $\link{l}\in L$ a \emph{shadow} of $\link{l}$, and denote by $\shadows(L)\subseteq V\times V$ the set of all shadows of links in $L$.

Our algorithm starts by computing a WRAP solution of cost at most $2c(\OPT)$, by determining an optimal directed WRAP solution.
We then shorten its links until we obtain a \emph{non-shortenable} solution, i.e., a directed WRAP solution such that replacing any of its links by a distinct shorter version of it would lead to a link set that is not a directed WRAP solution anymore.
\begin{lemma}\label{lem:goodNonShortenableSol}
Given a rooted WRAP instance $(G,L,c,r,e_r)$, one can efficiently find a non-shorten\-able directed WRAP solution $\vec{F}\subseteq \shadows(L)$ with $c(\vec{F})\leq 2c(\OPT)$.
\end{lemma}

Most importantly, non-shortenable WRAP solutions turn out to be highly structured, which we heavily exploit to later replace directed link sets by cheaper undirected ones.
The structure of non-shortenable WRAP solutions is described in \cref{thm:ov_StructureNonShortenable} below, where, for a ring graph $G$, we say that a directed link set $\vec{F}$ is \emph{$G$-planar} if a planar graph is obtained, when drawing $G$ in a planar way on a circle and adding the links in $\vec{F}$ as straight lines.
Moreover, an \emph{$r$-arborescence} is an arborescence rooted at $r$.
See \cref{fig:ov_nonShortenableExample} for an example of a non-shortenable directed WRAP solution.
(Moreover, see \cref{sec:directed-solutions} for an equivalent, non-geometric definition of $G$-planarity.)
\begin{restatable}[]{theorem}{propNonShortenable}
\label{thm:ov_StructureNonShortenable}
Let $(G,L,c,r,e_r)$ be a rooted WRAP instance and $\vec{F} \subseteq V\times V$ be a non-shortenable directed solution thereof.
Then $\vec{F}$ fulfills
\begin{enumerate}[label=(\roman*)]
\item $(V,\vec{F})$ is an $r$-arborescence;
\item\label{item:Gplanar} $\vec{F}$ is $G$-planar;
\item\label{item:noTwoSameDir} for any $v\in V$, no two directed links of $\delta^+_{\vec{F}}(v)$ go in the same direction (along the ring).
\end{enumerate}
\end{restatable}

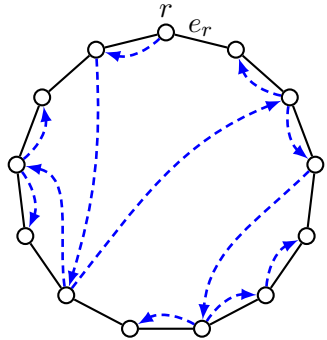
\begin{figure}[!ht]
\begin{center}
\begin{tikzpicture}[scale=1,
ns/.style={thick,draw=black,fill=white,circle,minimum size=6,inner sep=2pt},
es/.style={thick},
lks/.style={line width=1pt, blue, densely dashed},
dlks/.style={lks, -latex},
ts/.style={every node/.append style={font=\scriptsize}}
]

\def\rad{2}
\def\num{13}

\begin{scope}[every node/.style={ns}]
\foreach \i in {1,...,\num} {
  \pgfmathsetmacro\r{90+(\i-1)*360/\num}
  \node (\i) at (\r:\rad) {};
}
\end{scope}

\begin{scope}
\node at (1)[above=2pt] {$r$};
\path (\num) to node[above=-2pt] {$e_r$} (1);
\node at (4) [left=0.5pt] {};
\node at (5) [left=0.5pt] {};
\node at (6) [below=2pt] {};
\node at (8) [below=2pt] {};
\node at (11) [right=0.5pt] {};
\end{scope}

\begin{scope}[es]
\foreach \i in {1,...,\num} {
\pgfmathtruncatemacro\j{1+mod(\i,\num)}
\draw (\i) -- (\j);
}
\end{scope}

\begin{scope}[dlks]
\draw (1) to[bend left] (2);
\draw (2) to[bend left=10] (6);
\draw (6) to[bend right, in=-120, out=-10] (4);
\draw (4) to[bend right] (3);
\draw (4) to[bend left] (5);
\draw (6) to[bend left, in=160, out=10] (12);
\draw (12) to[bend left] (13);
\draw (12) to[bend right] (11);
\draw (11) to[bend right, out=-10] (8);
\draw (8) to[bend right] (7);
\draw (8) to[bend left] (9);
\draw (9) to[bend left] (10);
\end{scope}

\end{tikzpicture}
 \end{center}
\caption{Example of a non-shortenable WRAP solution depicted as dashed blue arrows.}
\label{fig:ov_nonShortenableExample}
\end{figure}
We highlight that the tree structure of the arborescence does not provide any clear path to connect the problem to Weighted Tree Augmentation.
However, leveraging the rich structure of non-shortenable directed WRAP solutions, we now expand on how we replace directed links by undirected link sets.

\subsection{Replacing directed links by undirected components}\label{sec:ovReplByUndir}

As mentioned, we think of directed links as only covering the $2$-cuts that they enter.
Hence, this leads to the following natural notion of \emph{mixed WRAP solution}.
\begin{definition}[mixed WRAP solution]
Let $\mathcal{I}=(G,L,c,r,e_r)$ be a rooted WRAP instance, $F\subseteq L$, and $\vec{F}\subseteq \shadows(L)$.
For $C\in \mathcal{C}_G$, we say that
\begin{itemize}
\item a directed link $\vec{\link{l}} \in \vec{F}$ \emph{covers} $C$ if $\vec{\link{l}} \in \delta_{\vec{F}}^-(C)$, and 
\item an undirected link $\link{l} \in F$ \emph{covers} $C$ if $\link{l} \in \delta_F(C)$.
\end{itemize}
We call $\vec{F}\cupp F$ a \emph{(mixed) solution} to $\mathcal{I}$ if, for every $C\in \mathcal{C}_G$, there is a link in $\vec{F}\cupp F$ covering $C$.
\end{definition}
The notion of mixed WRAP solution leads to a canonical way to define directed link sets that can be dropped form a directed WRAP solution after adding a set of undirected links.
Namely, when the resulting mixed link set is a mixed WRAP solution.
However, we use a more restrictive (and thus sufficient) condition of when a directed link can be dropped, by assigning to each directed $\vec{\link{l}}$ a sub-family $\Rscr_{\vec{F}}(\vec{\link{l}})$ of all $2$-cuts for which it is responsible, such that each $2$-cut is assigned to precisely one directed link.\footnote{See \cref{def:responsible} for a formal definition of responsibility and \cref{lem:responsibility_unique} for a proof that, for each $2$-cut, there is precisely one directed link responsible for it.}

Hence, whenever we add  an undirected link set $K\subseteq L$ to a non-shortenable directed WRAP solution $\vec{F}\subseteq \shadows(L)$, we can remove from $\vec{F}\cup K$ every link $\vec{\link{l}}\in \vec{F}$ for which all cuts in $\mathcal{R}_{\vec{F}}(\vec{\link{l}})$ are covered by $K$, leading to another mixed solution.
We denote this droppable link set of $\vec{F}$ when adding $K$ by
\begin{equation*}
   \Drop_{\vec{F}}(K) \coloneqq \Bigl\{ \vec{\link{l}}\in \vec{F}: \text{ for all }C\in \Rscr_{\vec{F}}(\vec{\link{l}})\text{ we have } \delta_K(C) \ne \emptyset \Bigr\}.
\end{equation*}

To improve a directed WRAP solution $\vec{F}$ through the addition of a set of undirected links $K\subseteq L$ and subsequent removal of $\Drop_{\vec{F}}(K)$, we aim at finding a cheap link set $K$ for which $\Drop_{\vec{F}}(K)$ is expensive.
In what follows, we identify a good class $\mathfrak{K}\subseteq 2^L$ of components/link sets, over which we can find such a cheap link set $K$ with expensive drop if it exists.

Note that the drop is always with respect to a directed non-shortenable solution $\vec{F}$.
This corresponds to a first step in improving a fully directed WRAP solution.
However, as we see later, we can also use it to improve the directed part of mixed solutions, even if the directed part is not a full directed solution by itself.
Nevertheless, in what follows it is helpful to think about improving a directed WRAP solution instead of a mixed one.
We recall that \cref{lem:goodNonShortenableSol} guarantees that we can start with a $2$-approximate non-shortenable directed WRAP solution.
Hence, a small constant-factor improvement of such a solution already leads to a better-than-$2$ algorithm for WCAP.

\subsection{Understanding when links can be dropped}\label{sec:ovWhenToDropLinks}

To obtain a suitable link class for components, we rely on a better understanding of what properties an undirected link set $K\subseteq L$ needs to fulfill for a link $\vec{\link{l}}\in \vec{F}$ to be part of $\Drop_K(\vec{F})$.
The key structure in this context is the link intersection graph.
It is defined as follows, where two links $\link{l},\link{f}\in L$ are called \emph{intersecting} if either their endpoints interleave on the ring---i.e., when drawing both links as straight lines on a planar circular embedding of the ring, they are crossing---or they share an endpoint.
\begin{definition}\label{def:linkIntersectionGraph}
Let $(G,L,c)$ be a WRAP instance.
The \emph{link intersection graph} $H$ is the graph with vertex set $L$ and edge set
\[
\Bigl\{ \{\link{l},\link{f}\}\in
\begin{psmallmatrix}
L\\
2
\end{psmallmatrix}
 : \link{l} \text{ and } \link{f}\text{ intersect}\Bigr\}.
\]
For a link set $K\subseteq L$, the \emph{link intersection graph of $K$} is the subgraph $H[K]$ of the link intersection graph $H$ induced by $K$.
\end{definition}
See \cref{fig:ovlinkIntersectionGraph} for an example.
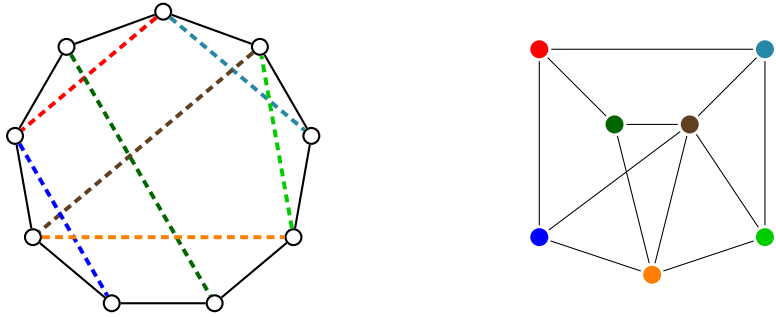
\begin{figure}[!ht]
\begin{center}
\begin{tikzpicture}[scale=1,
ns/.style={thick,draw=black,fill=white,circle,minimum size=6,inner sep=2pt},
LIGns/.style={thick,draw,fill,circle,minimum size=6,inner sep=2pt, outer sep=1.5pt},
es/.style={thick},
lks/.style={line width=1.5pt, densely dashed},
dlks/.style={lks, -latex},
ts/.style={every node/.append style={font=\scriptsize}}
]

\def\rad{2}
\def\num{9}

\begin{scope}[every node/.style={ns}]
\foreach \i in {1,...,\num} {
  \pgfmathsetmacro\r{90+(\i-1)*360/\num}
  \node (\i) at (\r:\rad) {};
}
\end{scope}

\begin{scope}[es]
\foreach \i in {1,...,\num} {
\pgfmathtruncatemacro\j{1+mod(\i,\num)}
\draw (\i) -- (\j);
}
\end{scope}

\begin{scope}[lks]
\draw[red] (1) -- (3);
\draw[cyan!60!black] (1) -- (8);
\draw[brown!50!black] (9) -- (4);
\draw[green!40!black] (2) -- (6);
\draw[green!80!black](7) -- (9);
\draw[blue] (3) -- (5);
\draw[orange] (4) -- (7);
\end{scope}

\begin{scope}[shift={(5,0)}]
\begin{scope}[every node/.style={LIGns}]
\node[red] (v) at (0,1.5) {};
\node[cyan!60!black] (c) at (3,1.5) {};
\node[brown!50!black] (br) at (2,0.5) {};
\node[green!40!black] (g) at (1,0.5) {};
\node[green!80!black] (r) at (3,-1) {};
\node[blue] (bl) at (0,-1) {};
\node[orange] (o) at (1.5,-1.5) {};
\end{scope}

\begin{scope}

\end{scope}[es]
\draw (v) -- (g);
\draw (v) -- (bl);
\draw (v) -- (c);
\draw (g) -- (br);
\draw (g) -- (o);
\draw (bl) -- (br);
\draw (bl) -- (o);
\draw (c) -- (br);
\draw (c) -- (r);
\draw (r) -- (o);
\draw (r) -- (br);
\draw (o) -- (br);
\end{scope}

\end{tikzpicture}
 \end{center}
\caption{Example of a WRAP instance (left) and its link intersection graph (right).
Every link in the left picture corresponds to a vertex of the link intersection graph, which is shown in the same color in the right picture.}
\label{fig:ovlinkIntersectionGraph}
\end{figure}
Graphs arising as intersection graphs of segments of a circle, like link intersection graphs, are known as \emph{circle graphs} in Graph Theory, and have been heavily studied under various aspects. (See \cref{sec:drop} for some further information and references.)

Whether a link $(u,v)\in \vec{F}$ is contained in a set $\Drop_{\vec{F}}(K)$ for $K\subseteq L$ can be characterized by how the intersection graph connects vertices of the ring, according to the following notion.
\begin{definition}\label{def:ConnectedInLinkIntersectionGraph}
Let $(G=(V,E),L,c)$ be a WRAP instance and $K\subseteq L$. A vertex $u\in V$ is \emph{connected to} $v\in V$ \emph{in the link intersection graph $H[K]$ of $K$} if there is a path in $H[K]$ from a link incident to $u$ to a link incident to $v$.
\end{definition}

We will show that a directed link $(u,v)\in \vec{F}$ is contained in $\Drop_{\vec{F}}(K)$ if and only if $v$ is connected, in the link intersection graph of $K$, to a vertex that is not a descendant of $v$ in the arborescence $(V,\vec{F})$.
We call such vertices \emph{$v$-good vertices}, and all other vertices of $V$, i.e., $v$'s descendants in $(V,\vec{F})$, including $v$ itself, are called \emph{$v$-bad}. (See \cref{fig:characterizeDrop}, which also illustrates \cref{lem:ov_characterize_drop} stated below.)
\begin{figure}[!ht]
\begin{center}
\begin{tikzpicture}[scale=1.0,
ns/.style={thick,draw=black,fill=white,circle,minimum size=6,inner sep=2pt},
es/.style={thick},
lks/.style={line width=1pt, blue, densely dashed},
dlks/.style={lks, -latex},
ts/.style={every node/.append style={font=\scriptsize}}
]

\def\rad{2}
\def\num{13}

\coordinate (c1) at (0,0);

\begin{scope}[every node/.style={ns}]
\foreach \i in {1,...,\num} {
  \pgfmathsetmacro\r{90+(\i-1)*360/\num}
  \ifthenelse{\i < 5 \OR \i > 10}
  { 
      \node[fill=blue!70!green, fill opacity=0.6] (\i) at (\r:\rad) {};
  }{
      \node[fill=orange, fill opacity=0.6] (\i) at (\r:\rad) {};
  }
}
\end{scope}

\begin{scope}
\node at (1)[above=2pt] {$r$};
\path (\num) to node[above=-2pt] {$e_r$} (1);
\node at (8) [below=2pt] {$v$};
\end{scope}

\begin{scope}[es]
\foreach \i in {1,...,\num} {
\pgfmathtruncatemacro\j{1+mod(\i,\num)}
\draw (\i) -- (\j);
}
\end{scope}

\begin{scope}[dlks]
\draw (1) to[bend left] (2);
\draw (2) to[bend left=45] (4);
\draw (4) to[bend right] (3);
\draw (4) to[bend right] (12);
\draw (12) to[bend left] (13);
\draw (12) to[bend right] (11);
\draw[red] (11) to[bend right, out=-10] node[above=2pt] () {$\vec{\link{l}}$} (8);
\draw (8) to[bend right=40] (6);
\draw (6) to[bend right] (5);
\draw (6) to[bend left] (7);
\draw (8) to[bend left] (9);
\draw (9) to[bend left] (10);
\end{scope}

\begin{scope}[shift={(7,0)}]

\coordinate (c2) at (0,0);

\def\rad{2}
\def\num{13}

\begin{scope}[every node/.style={ns}]
\foreach \i in {1,...,\num} {
  \pgfmathsetmacro\r{90+(\i-1)*360/\num}
  \ifthenelse{\i < 5 \OR \i > 10}
  { 
      \node[fill=blue!70!green, fill opacity=0.6] (\i) at (\r:\rad) {};
  }{
      \node[fill=orange, fill opacity=0.6] (\i) at (\r:\rad) {};
  }
}
\end{scope}

\begin{scope}
\node at (1)[above=2pt] {$r$};
\path (\num) to node[above=-2pt] {$e_r$} (1);
\node at (8) [below=2pt] {$v$};
\node[blue!70!green] at ($(c1)!0.5!(c2) + (0,0.2)$) {$v$-good};
\node[orange] at ($(c1)!0.5!(c2) + (0,-0.6)$) {$v$-bad};
\end{scope}

\begin{scope}[es]
\foreach \i in {1,...,\num} {
\pgfmathtruncatemacro\j{1+mod(\i,\num)}
\draw (\i) -- (\j);
}
\end{scope}

\begin{scope}[lks]
\draw (8) to node[above=2pt] () {$\link{l}_0$} (6);
\draw (7) to node[above=2pt] () {$\link{l}_1$} (10);
\draw (10) to node[above=2pt] () {$\link{l}_2$} (3);
\end{scope}

\end{scope}

\end{tikzpicture}
 \end{center}
\caption{
The left picture shows a non-shortenable directed WRAP solution $\vec{F}$ with a link $\vec{\link{l}}$ highlighted in red. 
For the head $v$ of $\vec{\link{l}}$, the $v$-good vertices are colored blue, while the $v$-bad vertices are colored orange.
The right picture shows a link set $\{\link{l}_0, \link{l}_1, \link{l}_2\}$ that induces a path in the link intersection graph. 
This path connects $v$ to a $v$-good vertex.
Thus, by \cref{lem:ov_characterize_drop}, we have $\vec{\link{l}} \in \Drop_{\vec{F}}(\{\link{l}_0, \link{l}_1, \link{l}_2\})$.
}
\label{fig:characterizeDrop}
\end{figure}
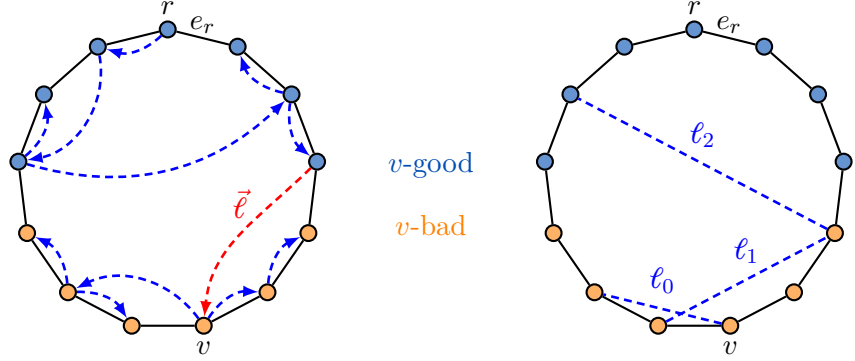
\begin{restatable}[]{lemma}{ovCharDrop}\label{lem:characterize_drop}
\label{lem:ov_characterize_drop}
Let $(G,L,c,r, e_r)$ be a rooted WRAP instance, $\vec{F} \subseteq \shadows(L)$ be a non-shortenable directed solution thereof, and  $K\subseteq L$.
Then $(u,v)\in \vec{F}$ is contained in $\Drop_{\vec{F}}(K)$ if and only if $v$ is connected to a $v$-good vertex in the link intersection graph of $K$.
\end{restatable}

The structure of non-shortenable directed solutions, as guaranteed by \cref{thm:ov_StructureNonShortenable}, allows for showing that sets of $v$-bad vertices are consecutive on the ring, i.e., they form an interval.

From \cref{lem:ov_characterize_drop}, we are able to derive the following central characterization of which directed links can be dropped after adding an undirected link set $S$ that is connected in the link intersection graph, i.e., where the subgraph $H[S]$ of $H$ induced by $S$ is connected.
The characterization uses the notion of \emph{least common ancestor $\lca(U)$} of a set $U\subseteq V$ with respect to a non-shortenable directed WRAP solution $\vec{F}$, which is the vertex farthest away from the root $r$ in the arborescence $(V, \vec{F})$ that is a common ancestor of all vertices in $U$.
Moreover, for a link set $S$, we denote by $V(S)$ the set of vertices that are incident to a link from $S$.
\begin{restatable}{lemma}{AllExceptOneDroppable}\label{lem:ov_dropOfConnectedS}
Let $(G,L,c,r,e_r)$ be a rooted WRAP instance with non-shortenable solution $\vec{F} \subseteq \shadows(L)$, and let $S\subseteq L$ be a link set such that $H[S]$ is connected, where $H$ is the link intersection graph.
Then
\begin{equation*}
\Drop_{\vec{F}}(S) = \left(\bigcup_{v\in V(S)} \delta^-_{\vec{F}}(v)\right)\setminus \delta^-_{\vec{F}}(\lca(V(S))).
\end{equation*}
\end{restatable}
\cref{lem:ov_dropOfConnectedS} has important consequences, especially when it comes to finding a good component $K$ to add.
Most importantly, it implies that, after adding a link set $S\subseteq L$ for which $H[S]$ is connected, we can drop each directed link in $\vec{F}$ except for possibly the single link $\vec{\link{l}}$ in $\delta_{\vec{F}}^-(\lca(V(S)))$, if the head of $\vec{\link{l}}$ is an endpoint of a link in $S$.
(Note that, indeed, $\delta^-_{\vec{F}}(\lca(V(S)))$ contains at most one link because $\vec{F}$ is an arborescence due to \cref{thm:ov_StructureNonShortenable}.)
We crucially exploit this property in our dynamic program to find good components, because it allows us to limit the bookkeeping we have to do to combine smaller solutions into larger ones.

Note that for an arbitrary link set $K\subseteq L$ to add to a non-shortenable directed solution, we can apply the statement independently to each connected component of $H[K]$.
Hence, \cref{lem:ov_dropOfConnectedS} reveals another important link between what can be dropped and the link intersection graph, namely that link sets $S\subseteq L$ that are connected in the link intersection graph lead to an easy-to-describe set of droppable links.

\subsection{Thin components and their key properties}\label{sec:ovThinComps}

We now introduce our component class over which we optimize.
The link sets $F\subseteq L$ in this component class are characterized by admitting an inclusion-wise maximal laminar family
\footnote{
A family of sets is called \emph{laminar} if for every two sets $A$, $B$ in the family, we have $A\subseteq B$, $B\subseteq A$, or $A\cap B = \emptyset$.
}$\mathcal{L}\supseteq \mathcal{C}_G$ of $2$-cuts, such that each set in $\mathcal{L}$ is crossed by only few links of $F$. 
We call such link sets \emph{$\alpha$-thin}, as they play an analogous role to $\alpha$-thin link sets in the context of Weighted Tree Augmentation.
\begin{restatable}[$\alpha$-thin]{definition}{alphaThin}
\label{def:alpha-thin}
Let $(G,L,c,r, e_r)$ be a rooted WRAP instance and $\alpha \in \mathbb{Z}_{>0}$.
A set $K\subseteq L$ is \emph{$\alpha$-thin} if there exists an inclusion-wise maximal laminar family $\Lscr \subseteq \Cscr_G$ such that $|\delta_K(C)| \le \alpha$ for all $C\in \Lscr$.
\end{restatable}
See \cref{fig:ov_definition-width} for an example of a $3$-thin link set.
Interestingly, there is a close connection between link sets with bounded treewidth in the link intersection graph and our thinness notion, which is arguably more natural for our purposes.
In particular, any constant-thin link set has bounded treewidth in the link intersection graph and vice versa. (See \cref{sec:triangulations} for more details.)
\begin{figure}[!ht]
\begin{center}
\begin{tikzpicture}[scale=1,
ns/.style={thick,draw=black,fill=white,circle,minimum size=6,inner sep=2pt},
es/.style={thick},
lks/.style={line width=1pt, blue, densely dashed},
dlks/.style={lks, -latex},
ts/.style={every node/.append style={font=\scriptsize}}
]

\def\num{11}
\def\hd{1.2}

\begin{scope}[shift={(0,0)}]

\begin{scope}[every node/.style={ns}]
\foreach \i in {1,...,\num} {
  \node (a\i) at (\i*\hd,0) {};
}
\end{scope}

\begin{scope}
\node at ($(a1)+(-0.3,0)$)[below=3pt] {$r=v_1$};
\pgfmathtruncatemacro\n{\num-1}
\foreach \i in {2,...,\n} {
\node at (a\i)[below=3pt] {$v_{\i}$};
}
\node at ($(a\num)+(0.1,0)$)[below=3pt] {$v_{\num}$};
\end{scope}

\begin{scope}[es]
\foreach \i in {2,...,\num} {
\pgfmathtruncatemacro\j{\i-1}
\draw (a\j) -- (a\i);
}
\draw (a1) to[out=-50,in=230,looseness=0.35] node[below,pos=0.8] {$e_r$} (a\num);
\end{scope}

\def\ih{0.15}
\def\vsep{0.25}
\begin{scope}[every path/.style={draw, fill, fill opacity=0.1,very thick}]
\newcommand\Interval[3]{
  \pgfmathsetmacro{\xa}{(#1)*\hd - \hd / 3}
  \pgfmathsetmacro{\xb}{(#2)*\hd + \hd / 3}
  \pgfmathsetmacro{\ya}{#3 - \ih / 2 }
  \pgfmathsetmacro{\yb}{#3}
  \pgfmathsetmacro{\yc}{#3 + \ih / 2 }
  
  \draw (\xa,\yb) -- (\xb,\yb);
  \draw (\xa,\ya) -- (\xa, \yc);
  \draw (\xb,\ya) -- (\xb, \yc);
}
\begin{scope}[gray!90!black]
\Interval{2}{2}{7*\vsep}
\Interval{3}{3}{5*\vsep}
\Interval{4}{4}{3*\vsep}
\Interval{5}{5}{2*\vsep}
\Interval{6}{6}{\vsep}
\Interval{7}{7}{\vsep}
\Interval{8}{8}{4*\vsep}
\Interval{9}{9}{5*\vsep}
\Interval{10}{10}{4*\vsep}
\Interval{11}{11}{4*\vsep}
\Interval{6}{7}{2*\vsep}
\Interval{10}{11}{5*\vsep}
\Interval{5}{7}{3*\vsep}
\Interval{4}{8}{5*\vsep}
\Interval{3}{8}{6*\vsep}
\Interval{9}{11}{6*\vsep}
\Interval{3}{11}{7*\vsep}
\Interval{2}{11}{8*\vsep}
\end{scope}
\begin{scope}[blue]
\Interval{4}{7}{4*\vsep}
\end{scope}
\end{scope}

\end{scope}
\begin{scope}[shift={(0,-3)}]

\begin{scope}[every node/.style={ns}]
\foreach \i in {1,...,\num} {
  \node (a\i) at (\i*\hd,0) {};
}
\end{scope}

\begin{scope}
\node at ($(a1)+(-0.3,0)$)[below=3pt] {$r=v_1$};
\pgfmathtruncatemacro\n{\num-1}
\foreach \i in {2,...,\n} {
\node at (a\i)[below=3pt] {$v_{\i}$};
}
\node at ($(a\num)+(0.1,0)$)[below=3pt] {$v_{\num}$};
\end{scope}

\def\cw{0.55}
\begin{scope}[every path/.style={draw, fill, fill opacity=0.1,thick,blue}]
\newcommand\Interval[2]{
  \pgfmathsetmacro{\xa}{(#1)*\hd}
  \pgfmathsetmacro{\xb}{(#2)*\hd}
\draw (\xa,-\cw) to[bend left, in=90, out=90, looseness=1.7] (\xa,\cw)
  -- (\xb,\cw)
 to[bend left, in=90, out=90, looseness=1.7] (\xb,-\cw)
 -- (\xa,-\cw)
;
}
\Interval{4}{7}
\end{scope}

\begin{scope}[es]
\foreach \i in {2,...,\num} {
\pgfmathtruncatemacro\j{\i-1}
\draw (a\j) -- (a\i);
}
\draw (a1) to[out=-50,in=230,looseness=0.35] node[below,pos=0.8] {$e_r$} (a\num);
\end{scope}

\begin{scope}[lks, darkred]
\draw (a1) to[bend left=40] (a5);
\draw (a4) to[bend left=40] (a8);
\draw (a6) to[bend  left=40] (a10);
\draw (a9) to[bend  left=40] (a11);
\draw (a2) to[bend  left=30] (a9);
\end{scope}

\end{scope}

\end{tikzpicture}
 \end{center}
\caption{
The upper part of the figure shows a cycle $G$ with root $r$ and a maximal laminar subfamily $\Lscr$ of $\Cscr_G$.
The lower part of the figure shows the same cycle $G$ together with a $3$-thin link set $K$ (shown in red).
The fact that $K$ is $3$-thin is certified by the laminar family $\Lscr$ from the upper part of the figure because we have $|\delta_K(C)|\le 3$ for all  $C\in \Lscr$.
As an example, a cut $C\in \Lscr$ with $|\delta_K(C)|= 3$ is highlighted in blue in both the upper and the lower part of the figure.
}
\label{fig:ov_definition-width}
\end{figure}
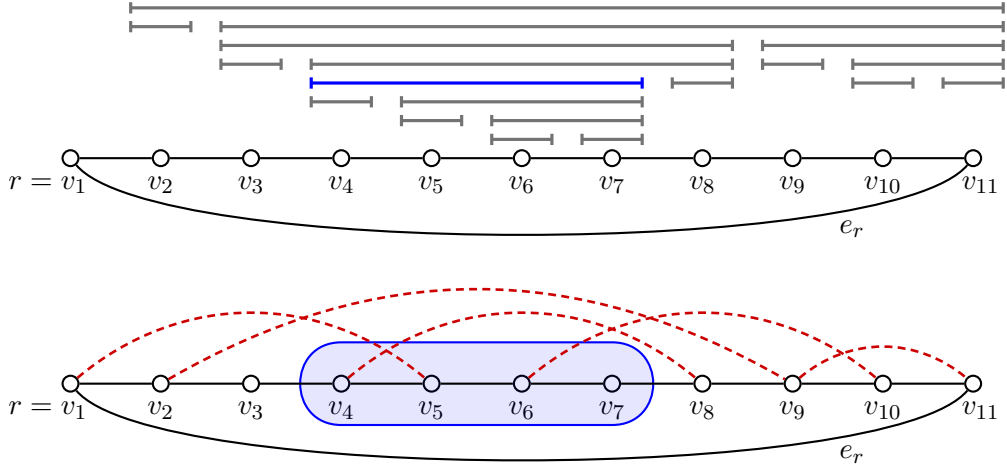

This notion of thinness, together with the rich properties previously derived, allow us to obtain two key theorems, one about finding a good $\alpha$-thin component (an optimization theorem) and one about decomposing any WRAP solutions into $\alpha$-thin parts that can be used to show existence of a good $\alpha$-thin component (a decomposition theorem).
These statements are among our key technical contributions.
We start with the optimization theorem, where we use the convention $\frac{0}{0}\coloneqq 1$ and $\frac{x}{0}\coloneqq \infty$ for $x > 0$.
\begin{restatable}[optimization theorem]{theorem}{ovFindCompForRelGreedy}\label{thm:ov_find-component-for-relative-greedy}
Given a rooted WRAP instance with a non-shortenable directed solution $\vec{F}_0 \subseteq \shadows(L)$, a set $\vec{F} \subseteq \vec{F}_0$, and a constant $\alpha$, we can efficiently compute among all $\alpha$-thin components a component $K$ minimizing
\[
\frac{c(K)}{c(\Drop_{\vec{F}_0}(K) \cap \vec{F})}.
\]
Moreover, if $\vec{F}\neq \emptyset$, then the returned minimizer $K$ satisfies $\Drop_{\vec{F}_0}(K)\cap \vec{F}\neq \emptyset$.
\end{restatable}
One way to leverage this theorem and illustrate its usefulness is by realizing the following natural relative greedy algorithm.
Given a WRAP instance $(G,L,c,r,e_r)$, we start with a non-shortenable directed WRAP solution $\vec{F}_0 \subseteq \shadows(L)$ with $c(\vec{F}_0)\leq 2 c(\OPT)$, as guaranteed by \cref{lem:goodNonShortenableSol}.
(We recall that $\OPT$ is an optimal solution to the WRAP instance.)
The algorithm maintains a mixed solution $\vec{F}\cup T$, which is $\vec{F}_0$ at the start, and iteratively adds an $\alpha$-thin component $K\subseteq L$ to it and removes $\Drop_{\vec{F}_0}(K)\cap \vec{F}$.
The optimization theorem then states that we can efficiently find an $\alpha$-thin component $K$ (for a constant $\alpha$) that minimizes the ratio between the cost of $K$ and what can be dropped when adding $K$.
Hence, we can efficiently maximize what we gain in terms of cost versus what we pay.

To make sure that such an algorithm will lead to a solution of approximation factor better than~$2$, we rely on the following decomposition theorem.
\begin{restatable}[decomposition theorem]{theorem}{ovDecompositionTheorem}\label{thm:ov_decomposition-theorem}
Let $(G,L,c, r, e_r)$ be a rooted WRAP instance, $S\subseteq L$ be a WRAP solution, and let $\vec{F}_0\subseteq \shadows(L)$ be a non-shortenable directed WRAP solution.
Then, for any $\epsilon >0$, there exists a partition $\mathcal{K}$ of $S$ into $4\lceil\sfrac{1}{\epsilon}\rceil$-thin sets and a set $\vec{R}\subseteq \vec{F}_0$ such that
\begin{enumerate}
\item\label{item:ov_allIsDroppableExceptForR} for every $\vec{\link{l}}\in \vec{F}_0\setminus \vec{R}$, there exists some $K \in \mathcal{K}$ such that $\vec{\link{l}}\in \Drop_{\vec{F}_0}(K)$, and  \label{item:ov_links_outside_R_covered}
\item\label{item:ov_RIsCheap} $c(\vec{R}) \le \epsilon \cdot c(\vec{F}_0)$.
\end{enumerate}
\end{restatable}
To see how the decomposition theorem allows for showing the existence of good $\alpha$-thin components, consider a mixed WRAP solution $\vec{F}\cup T$ with $\vec{F}\subseteq \vec{F}_0$ and apply the decomposition theorem to an optimal WRAP solution $S=\OPT\subseteq L$.
The partition $\mathcal{K}$ of $\OPT$, guaranteed by \cref{thm:ov_decomposition-theorem}, leads to the following relation:
\begin{equation}\label{eq:ov_existence_good_comp}
\begin{aligned}
\sum_{K\in \mathcal{K}} c(\Drop_{\vec{F}_0}(K)\cap \vec{F}) &\geq c(\vec{F}\setminus \vec{R}) \geq c(\vec{F}) - \varepsilon\cdot c(\vec{F}_0) \geq c(\vec{F}) - 2\varepsilon \cdot c(\OPT)\\
&= \left(\frac{c(\vec{F})}{c(\OPT)} - 2 \epsilon\right)\cdot c(\OPT)
= \left(\frac{c(\vec{F})}{c(\OPT)} - 2 \epsilon\right)\cdot \sum_{K\in \mathcal{K}} c(K),
\end{aligned}
\end{equation}
where the first inequality follows from \cref{thm:ov_decomposition-theorem}~\ref{item:ov_allIsDroppableExceptForR}, the second one from point~\ref{item:ov_RIsCheap}, the third one because we start with a $2$-approximation, i.e., $c(\vec{F}_0)\leq 2c(\OPT)$, and the final equality is a consequence of $\mathcal{K}$ being a partition of $\OPT$.

Hence, if $c(\vec{F}) > (1+2\varepsilon) c(\OPT)$, then~\eqref{eq:ov_existence_good_comp} implies that there is a component $K\in \mathcal{K}$ with $c(K)<c(\Drop_{\vec{F}_0}\cap \vec{F})$, i.e., a $4 \lceil \sfrac{1}{\varepsilon}\rceil$-thin component $K$ that costs less than what we can drop when adding it to the mixed solution.
Finally, \cref{thm:ov_find-component-for-relative-greedy} guarantees that we can find such a component efficiently.
More generally, if $c(\vec{F})$ is significantly more expensive than $c(\OPT)$, then~\eqref{eq:ov_existence_good_comp} implies by an averaging argument that there is a thin component of much lower cost than what can be dropped.
We formalize this reasoning in \cref{sec:thin-components} to conclude a first better-than-$2$ approximation for WCAP with approximation factor $1+\ln 2 + \epsilon < 1.7$, which we later improve in \cref{sec:local-search} through a local search approach to a $(1.5+\epsilon)$-approximation.

\subsection{Brief comments on proving the optimization and decomposition theorem}\label{sec:commentsOptThmDecThm}

We finish the overview section with some brief discussion on proving the optimization and decomposition theorem.
As mention, we only provide brief comments here to highlight some important aspects and refer for further information to \cref{sec:dp} and \cref{sec:decomposition-thm}, respectively.

\subsubsection*{Comments on optimization theorem}

Our algorithm that implies the optimization theorem is a dynamic program (DP).
Despite the fact that constant-thin link sets have bounded treewidth, as mention earlier, this does not imply that we can employ a standard DP approach common for bounded treewidth graphs.
The reason is that we do not optimize a quantity over a fixed constant-thin link set, but need to find a best constant-thin link set in a bigger graph (the link intersection graph), which is not constant-thin in general.
This is a fundamentally different task.

Our DP constructs an $\alpha$-thin component $K$ minimizing $\sfrac{c(K)}{c(\Drop_{\vec{F}_0}(K)\cap \vec{F})}$ bottom-up with respect to the unknown maximal laminar family $\mathcal{L}$ over $V\setminus \{r\}$ that certifies $\alpha$-thinness of $K$.
In particular, for a $2$-cut $C\in \mathcal{C}_G$, which is a candidate set for being in $\mathcal{L}$, we only consider link sets $S\subseteq \bigcup_{v\in C}\delta_L(v)$ with $|S\cap \delta_L(C)|\leq \alpha$, because we are interested in $\alpha$-thin sets.
Of course, for a dynamic program to be efficient, we cannot simply save all such sets $S$ but need to save only a small fingerprint of such solutions.
Having $|S\cap \delta_L(C)|\leq \alpha$ implies that only constantly many connected components of $H[S]$ have at least one link in $\delta_L(C)$.
Hence, only these components have any relevance for extending the solution beyond $C$.
This is where \cref{lem:ov_dropOfConnectedS} comes into play.
It implies that, for each connected component $T$ of $H[S]$, there is at most one directed link whose head is an endpoint of a link in $T$ and that is not in $\Drop_{\vec{F}_0}(T)$.
This is the only directed link whose head is an endpoint of a link in $T$ and which may become droppable later, depending on how we extend the solution $S$.
As we show in \cref{sec:dp}, we can formalize and exploit this idea to keep the bookkeeping low and design an efficient dynamic program to prove the optimization theorem.

\subsubsection*{Comments on decomposition theorem}

The proof of our decomposition theorem is inspired by a similar decomposition theorem used in the context of WTAP \cite{traub_2021_better}, but differs in some important ways.
In particular, in WTAP, for each link $\vec{\link{l}}\in \vec{F}_0$ of the simpler starting solution , a minimal subset $S_{\vec{\link{l}}}\subseteq S$ of links of the given solution $S\subseteq L$ with $\vec{\link{l}}\in \Drop_{\vec{F}_0}(S_{\vec{\link{l}}})$ was fixed.\footnote{These subsets $S_{\vec{\link{l}}}$ where chosen with some additional property that is not further relevant for this discussion.}
(The links in the starting solution were called \emph{up-links} and are undirected in WTAP; however, their role in our context is taken by the directed links, which is why we use the notation $\vec{\link{l}}$.)
Then, by appropriately ordering the links in $S_{\vec{\link{l}}}$, which can be thought of as creating a path with vertices $S_{\vec{\link{l}}}$, and taking the union of all those paths for all links $\vec{\link{l}}\in \vec{F}_0$, a branching was obtained with vertices given by $\vec{F}_0$.
The structure of the branching, together with further properties, then allowed for identifying a set $\vec{R}\subseteq \vec{F}_0$ and a partition $\mathcal{K}$ of $S$ into thin components fulfilling properties analogous to the ones claimed in our decomposition theorem, \cref{thm:ov_decomposition-theorem}.

It turns out that the same approach fails in the more general WRAP setting, as it does not lead to a branching, which was key in the WTAP approach.
Even though this may seem unnatural at first sight, we address this issue by not choosing minimal sets $S_{\vec{\link{l}}}$ within $S$ with $\vec{\link{l}}\in \Drop_{\vec{F}_0}(S_{\vec{\link{l}}})$ but by choosing a minimal set of larger building blocks for which $\vec{\link{l}}$ is droppable.
Due to their structure, we call these carefully selected building blocks \emph{festoons}.
More precisely, we first create a partition $\Xscr$ of the given WRAP solution $S$ into such festoons, and then show that a similar approach as for WTAP can be extended to WRAP if we exchange the role of single links by festoons.
Hence, each festoon is a link set that remains together and will be a subset of one part of $\mathcal{K}$.
(See \cref{sec:decomposition-thm} for more details.)
 \section{Reducing to the Weighted Ring Augmentation Problem}\label{sec:wrap}

In this section we prove \cref{lem:reduction-to-WRAP}, i.e., we describe a reduction showing that to approximate WCAP it suffices to approximate the Weighted Ring Augmentation Problem (WRAP), which corresponds to WCAP with the underlying graph $G$ being a ring/cycle. 
This is a well-known fact that has been previously observed in \cite{galvez_2019_cycle}.

In order to prove that WRAP and WCAP are equivalent, we rely on the cactus representation of the minimum cuts of a graph from \cite{dinitz_1976_structure}.
A \emph{cactus} is a connected undirected graph $G=(V,E)$ where every edge $e\in E$ is contained in exactly one cycle in $G$.\footnote
{
In the literature, a cactus is sometimes defined as a graph where every edge is contained in at most one cycle. 
We will only consider $2$-edge-connected cacti, in which case the two definitions coincide.
}
\cite{dinitz_1976_structure} showed that the minimum cuts of an undirected graph $G$ can be represented by a cactus as follows.

\begin{theorem}[\cite{dinitz_1976_structure}]\label{thm:cactus-representation}
For any undirected graph $G=(V,E)$, one can efficiently construct a cactus $\overline{G}=(\overline{V},\overline{E})$ together with a map $\phi: V \to \overline{V}$ such that
\begin{enumerate}
\item If $C\subseteq V$ is a minimum cut in $G$, then there is a 2-cut $\overline{C} \subseteq \overline{V}$ in the cactus $\overline{G}$ with
\[
   C = \{ v\in V: \phi(v) \in \overline{C} \}.
\]
\item  If $\overline{C} \subseteq \overline{V}$ is a 2-cut in the cactus $\overline{G}$, then
\[
   C = \{ v\in V: \phi(v) \in \overline{C} \}.
\]
is a minimum cut in $G$.
\end{enumerate}
\end{theorem}

For CAP, i.e., the special case of WCAP with $w\equiv 1$, \cref{thm:cactus-representation} has previously been used to reduce the problem to the special case where the graph $G$ is a cactus; see, e.g., \cite{basavaraju_2014_parameterized, byrka_2020_breaching, nutov_2021_approximation, cecchetto_2021_bridging}.
For general WCAP, this reduction can be strengthened to reduce to the special case where the graph $G$ is a cycle/ring \cite{galvez_2019_cycle}.
We provide a proof for completeness.

\redToWRAP*
\begin{proof}
We first compute a cactus $\overline{G}=(\overline{V}, \overline{E})$ together with a map $\phi$ as in \cref{thm:cactus-representation}.
For a link $\link{l}=\{u,v\} \in L$, we define the corresponding link $\overline{\link{l}} \coloneqq \{\phi(u), \phi(v)\}$ and we define $ \overline{L} \coloneqq \{ \overline{\link{l}} : \link{l} \in L \}$.
Moreover, we set $\overline{c}(\overline{\link{l}}) \coloneqq c(\link{l})$.
Then a link set $F\subseteq L$ is a solution for the WCAP instance $(G,L,c)$ if and only if $\overline{F} \subseteq \overline{L}$ is a solution for the WCAP instance $(\overline{G}, \overline{L}, \overline{c})$ by \cref{thm:cactus-representation}.

Next, we define an instance $(\tilde{G}, \tilde{L},\tilde{c})$ of WRAP that arises from $(\overline{G}, \overline{L}, \overline{c})$ as follows.
We first observe that because $\overline{G}$ is a cactus, it is the union of edge-disjoint cycles and hence it is Eulerian and connected.
Therefore, there exists an Eulerian walk in $\overline{G}$, i.e., a closed walk that uses every edge exactly once, and such an Eulerian walk can be found efficiently.
For every vertex $v\in \overline{V}$ that is visited more than once in the Eulerian walk (where we do not count returning to the starting vertex as a separate visit), we do the following.
If the $k$th visit of the vertex $v$ for $k>1$ consists of entering $v$ by an edge $\{u,v\}$ and leaving it by an edge $\{v,w\}$, then we introduce a new vertex $v_k$ and replace the edges $\{u,v\}$ and $\{v,w\}$ by new edges $\{u,v_k\}$ and $\{v_k,w\}$.
Then we add a link $\{v,v_k\}$ and set its cost to zero.
Iteratively applying this for all vertices visited multiple times by the Eulerian walk leads to an instance of WRAP.
See \cref{fig:cactus-to-cycle}.

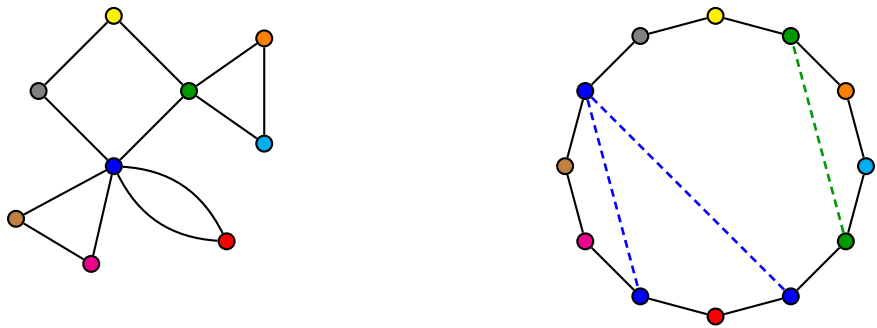
\begin{figure}[!ht]
\begin{center}
\begin{tikzpicture}[scale=1,
ns/.style={thick,draw=black,circle,minimum size=6,inner sep=2pt},
es/.style={thick},
lks/.style={line width=1pt, blue, densely dashed},
ts/.style={every node/.append style={font=\scriptsize}}
]

\begin{scope}[every node/.style={ns}]
\node[fill=yellow] (w1) at (2,2) {};
\node[fill=gray] (w2) at (1,1) {};
\node[fill=blue] (w3) at (2,0) {};
\node[fill=brown] (w4) at (0.7,-0.7) {};
\node[fill=magenta] (w5) at (1.7,-1.3) {};
\node[fill=red] (w7) at (3.5,-1) {};
\node[fill=green!60!black] (w9) at (3,1) {};
\node[fill=cyan] (w10) at (4,0.3) {};
\node[fill=orange] (w11) at (4,1.7) {};
\end{scope}

\begin{scope}[es]
\draw (w1) -- (w2) -- (w3) -- (w4) -- (w5) -- (w3) -- (w9) -- (w10) -- (w11) -- (w9) -- (w1);
\draw [bend right] (w3) to (w7);
\draw [bend right] (w7) to (w3);
\end{scope}

\begin{scope}[shift={(10,0)}]
\def\rad{2}
\def\num{12}

\begin{scope}[every node/.style={ns}]
\foreach \i/\c in {
1/yellow,
2/gray,
3/blue,
4/brown,
5/magenta,
6/blue,
7/red,
8/blue,
9/green!60!black,
10/cyan,
11/orange,
12/green!60!black} 
{
  \pgfmathsetmacro\r{90+(\i-1)*360/\num}
  \node[fill=\c] (v\i) at (\r:\rad) {};
}
\end{scope}

\begin{scope}[es]
\foreach \i in {1,...,\num} {
\pgfmathtruncatemacro\j{1+mod(\i,\num)}
\draw (v\i) -- (v\j);
}
\end{scope}

\begin{scope}[lks]
\draw[blue] (v3) -- (v6);
\draw[blue] (v3) -- (v8);
\draw[green!60!black] (v9) -- (v12);
\end{scope}
\end{scope}

\end{tikzpicture}
 \end{center}
\caption{Example of a cactus $\overline{G}$ (left) and the cycle $\tilde{G}$ (right) constructed from $\overline{G}$. 
For a vertex $v\in \overline{V}$, the vertices $v_k$ with $k>1$ are shown in the same color as $v$.
The links of cost zero that are added to $\overline{L}$ to obtain $\tilde{L}$ are shown as dashed lines.
}
\label{fig:cactus-to-cycle}
\end{figure}

We say that a $2$-cut in $\tilde{G}$ corresponds to a $2$-cut in $\overline{G}$ if the cuts contain the same set of links.
We claim that every $2$-cut in the cycle $\tilde{G}$ either contains a link of cost zero or it corresponds to a $2$-cut in $\tilde{G}$.
Indeed, if a $2$-cut in the cycle $\tilde{G}$ contains no link of cost zero, then for every vertex $v\in \overline{V}$, the vertices $v$ and $v_k$ for $k>1$ are either all contained in $C$ or are all outside of $C$, implying that $C$ corresponds to a $2$-cut of $\overline{G}$.
Moreover, every $2$-cut $\overline{C}$ in the cactus $\overline{G}$ corresponds to the $2$-cut in the cycle $\tilde{G}$ that contains for every vertex $v\in \overline{C}$, the vertex $v$ and all vertices $v_k$ for $k>1$ that exist in $\tilde{G}$.
This shows that every WCAP solution for $(\tilde{G}, \tilde{L},\tilde{c})$ can be transformed into a WRAP solution of the same cost for $(\tilde{G}, \tilde{L},\tilde{c})$ by adding all links of weight zero and, vice versa, every WRAP solution for $(\tilde{G}, \tilde{L},\tilde{c})$ can be transformed into a WCAP solution for $(\tilde{G}, \tilde{L},\tilde{c})$ of the same cost by omitting those links of cost zero that are not contained in~$\overline{L}$.
\end{proof}

Because of \cref{lem:reduction-to-WRAP}, we will focus on WRAP in the rest of this paper.
 \section{Directed Solutions}\label{sec:directed-solutions}

In this section we provide details on our discussion in \cref{sec:ovDirSols}, i.e., the directed setting we are using and how we obtain highly structured directed solutions.

We start by recalling that we move from an undirected WCAP setting to a directed setting by bidirecting the links and also considering shortenings of the thus obtained directed links.
Such shortenings are formally defined as follows (\cref{fig:ov_shortenings}).
\begin{definition}[Shortening of directed link]\label{def:shortening}
Let $(G=(V,E), L, c, r, e_r)$ be a rooted WRAP instance, and let $\vec{\link{l}}=(u,v)\in V\times V$ with $u\neq v$.
A tuple $(s,t) \in V\times V$ with $s \neq t$ is a \emph{shortening} of $\vec{\link{l}}$ if $t = v$ and $s$ is a vertex on the unique $u$-$v$ path in $(V,E\setminus \{e_r\})$.
\end{definition}
Note that a directed link is a shortening of itself.
If we want to refer to a shortening of a directed link that differs from itself, we talk about a $\emph{strict shortening}$.
Observe that a shortening of a directed link $\vec{\link{l}}$ only covers/enters a subset of the cuts in $\mathcal{C}_G$ that $\vec{\link{l}}$ covers/enters.
Hence, a shortening of $\vec{\link{l}}$ is indeed a weakening of $\vec{\link{l}}$.
We recall that a directed WRAP solution $\vec{F}\subseteq V\times V$ is said to be \emph{non-shortenable} if deleting any of its links or replacing it by a strict shortening leads to a link set that is not a directed WRAP solution anymore.

As previously discussed, we aim at highly structured directed WRAP solutions by shortening links until a non-shortenable solution is obtained.
For convenience, we use the following notion of $\emph{shadows}$ to capture directed links obtainable from an undirected one by first directing it one way and then shortening it. (See \cref{fig:shadows}.)
\begin{definition}[Shadow of a link]
Let $(G,L,c,r,e_r)$ be a rooted WRAP instance.
A directed link is a shadow of $\{u,v\}\in \begin{psmallmatrix} V\\ 2 \end{psmallmatrix}$ if it is either a shortening of $(u,v)$ or of $(v,u)$. 
We denote by $\shadows(L)\subseteq V\times V$ the set of all shadows of links in $L$.
\end{definition}
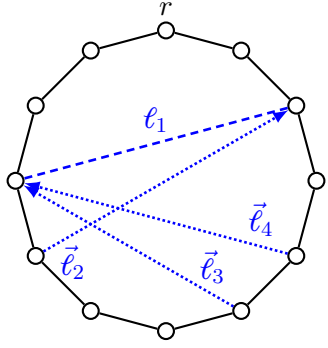
\begin{figure}[!ht]
\begin{center}
\begin{tikzpicture}[scale=1,
ns/.style={thick,draw=black,fill=white,circle,minimum size=6,inner sep=2pt},
es/.style={thick},
lks/.style={line width=1pt, blue, densely dashed},
dlks/.style={lks, -latex},
ts/.style={every node/.append style={font=\scriptsize}}
]

\def\rad{2}
\def\num{12}

\begin{scope}[every node/.style={ns}]
\foreach \i in {1,...,\num} {
  \pgfmathsetmacro\r{90+(\i-1)*360/\num}
  \node (\i) at (\r:\rad) {};
}
\end{scope}

\begin{scope}
\node at (1)[above=2pt] {$r$};
\end{scope}

\begin{scope}[es]
\foreach \i in {1,...,\num} {
\pgfmathtruncatemacro\j{1+mod(\i,\num)}
\draw (\i) -- (\j);
}
\end{scope}

\begin{scope}[lks]
\draw (11) -- node[above,pos=0.5] {$\link{l}_1$} (4);
\end{scope}

\begin{scope}[dlks, densely dotted]
\draw (5)  -- node[below=0pt,pos=0.12] {$\vec{\link{l}}_2$} (11);
\draw (8) -- node[above, pos=0.1] {$\vec{\link{l}}_3$} (4);
\draw (9) -- node[above, pos=0.1] {$\vec{\link{l}}_4$} (4);
\end{scope}

\end{tikzpicture}
 \end{center}
\caption{
Examples of shadows (dotted arrows) of a link (dashed line).
More precisely, the directed links $\vec{\link{l}_2}$, $\vec{\link{l}_3}$, and $\vec{\link{l}_4}$ are all shadows of the (undirected) link $\link{l}_1$.
}
\label{fig:shadows}
\end{figure}

Given a rooted WRAP instance $(G,L,c,r,e_r)$, whenever we take a shadow $\vec{\link{l}}$ of a link $\link{l}\in L$, we assign to it the same cost as the original undirected link, i.e., $c(\vec{\link{l}})\coloneqq c(\link{l})$.
This ensures that any directed solution built from shadows of $L$ can trivially be converted into an undirected WRAP solution of same cost.

A key advantage of the directed setting is that minimum cost directed solutions can be found efficiently.
This has been observed in~\cite{cecchetto_2021_bridging}.
Such a solution, after making it non-shortenable, will be a highly structured $2$-approximation, as we discuss in the following.
\begin{lemma}[{\cite[Lemma~14]{cecchetto_2021_bridging}}]\label{lem:dirWRAPInP}\!\!\footnote{More precisely, \cite[Lemma~14]{cecchetto_2021_bridging} states that when starting with a WCAP instance and replacing any undirected link $\{u,v\}$ by two directed links $(u,v)$ and $(v,u)$ of same cost, one can efficiently compute a minimum cost directed solution. This immediately implies \cref{lem:dirWRAPInP}, which differs only in also allowing shortenings of such links.
However, as discussed, shortenings do not lead to better solutions.
Finally, we highlight that the proof in~\cite{cecchetto_2021_bridging} shows that the natural linear program is integral and does not rely on any structure of the directed links or their costs.
}Given a rooted WRAP instance $(G,L,c,r,e_r)$.
One can efficiently find a minimum cost directed WRAP solution among all directed WRAP solutions consisting only of directed links in $\shadows(L)$.
\end{lemma}

\begin{corollary}\label{cor:dirSolTwoApprox}
Given a rooted WRAP instance $(G,L,c,r,e_r)$, one can efficiently find a directed WRAP solution $\vec{F} \subseteq \shadows(L)$ with $c(\vec{F}) \leq 2 c(\OPT)$, where $\OPT\subseteq L$ is an optimal WRAP solution.
\end{corollary}
\begin{proof}
We claim that a minimum cost directed WRAP solution $\vec{F}$, which can be found efficiently by \cref{lem:dirWRAPInP}, has the desired properties.
Indeed, the directed WRAP solution obtained by starting with $\OPT$ and replacing each link $\ell\in \OPT$ by two antiparallel directed links has cost precisely $2 c(\OPT)$.
Because $\vec{F}$ is a minimum cost directed solution consisting of directed links in $\shadows(L)$, its cost is not more than that; hence, $c(\vec{F})\leq 2 c(\OPT)$, as desired.
\end{proof}

Note that \cref{lem:goodNonShortenableSol} asserts that one can efficiently obtain a directed link set $\vec{F}\subseteq \shadows(L)$ that, on top of fulfilling the properties of \cref{cor:dirSolTwoApprox}, is also non-shortenable.
This follows immediately from \cref{cor:dirSolTwoApprox} and the fact that we can efficiently transform any directed WRAP solution into a non-shortenable one, as stated below.
\begin{lemma}\label{lem:shortening_efficiently}
Let $(G,L,c,r,e_r)$ be a rooted WRAP instance and $\vec{F} \subseteq \shadows(L)$ be a directed solution thereof.
Then one can efficiently construct a non-shortenable directed WRAP solution obtained from $\vec{F}$ by deleting and shortening links.
\end{lemma}
\begin{proof}
Such a non-shortenable directed WRAP solution can be obtained by going one-by-one through the links $\vec{\link{l}}\in \vec{F}$, in an arbitrary order, and deleting $\vec{\link{l}}$ if the remaining links still form a directed WRAP solution, or, if this is not the case, going through the strict shortenings of $\vec{\link{l}}$ and checking if replacing $\vec{\link{l}}$ by such a shortening leads to a directed WRAP solution.
Then we replace $\vec{\link{l}}$ by a shortest feasible shortening, i.e.,  a shortening that itself does not admit a strict shortening that leads again to a directed WRAP solution.
\end{proof}

We now continue by showing \cref{thm:ov_StructureNonShortenable}, i.e., that non-shortenable directed WRAP solutions are highly structured.
One crucial property of non-shortenable solutions that we show, is that they are \emph{$G$-planar}, i.e., they do not contain crossing links.
(See \cref{fig:Gplanar}.)
\begin{definition}[Crossing edges/links, $G$-planarity]
Let $G=(V,E)$ be a cycle and $e_1, e_2\in \begin{psmallmatrix}
V\\
2
\end{psmallmatrix}$.
The two edges $e_1,e_2$ are called \emph{crossing} if they do not share an endpoint and each of the two paths in $G$ between the endpoints of $e_1$ contains an endpoint of $e_2$.
A set $S\subseteq \begin{psmallmatrix}
V\\
2
\end{psmallmatrix}$ is called \emph{$G$-planar} if no pair of edges in $S$ is crossing.\footnote{One can observe that a set $S$ is $G$-planar if and only if, when drawing the ring $G=(V,E)$ on a cycle and connecting the endpoints of each edge in $S$ by a straight line, then a planar embedding of $(V, E\cup S)$ is obtained. 
This is equivalent to $(V,E\cup S)$ admitting an outer-planar embedding.}
\end{definition}
We use the notions of \emph{crossing} and \emph{$G$-planar} analogously for directed links, where directed links are \emph{crossing} or \emph{$G$-planar} if they have these properties when disregarding their orientations.

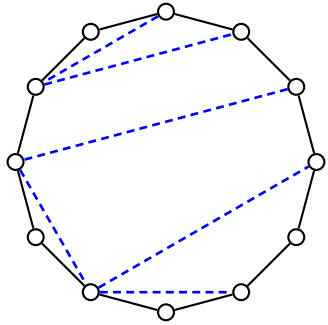
\begin{figure}[!ht]
\begin{center}
\begin{tikzpicture}[scale=1,
ns/.style={thick,draw=black,fill=white,circle,minimum size=6,inner sep=2pt},
es/.style={thick},
lks/.style={line width=1pt, blue, densely dashed},
dlks/.style={lks, -latex},
ts/.style={every node/.append style={font=\scriptsize}}
]

\def\rad{2}
\def\num{12}

\begin{scope}[every node/.style={ns}]
\foreach \i in {1,...,\num} {
  \pgfmathsetmacro\r{90+(\i-1)*360/\num}
  \node (\i) at (\r:\rad) {};
}
\end{scope}

\begin{scope}[es]
\foreach \i in {1,...,\num} {
\pgfmathtruncatemacro\j{1+mod(\i,\num)}
\draw (\i) -- (\j);
}
\end{scope}

\begin{scope}[lks]
\draw (1) -- (3);
\draw (3) -- (12);
\draw (11) -- (4);
\draw (4) -- (6);
\draw (6) -- (8);
\draw (6) -- (10);
\end{scope}

\end{tikzpicture}
 \end{center}
\caption{
Example of a $G$-planar set $S$ (drawn as dashed blue lines).
}
\label{fig:Gplanar}
\end{figure}

Apart from $G$-planarity, non-shortenable directed WRAP solutions have further properties in terms of which directions their links go.
To this end, we introduce notions of left/right and left-going/right-going for rooted WRAP.
These notions can easily be remembered by imagining that we cut open the ring by removing the arc $e_r$ and draw the resulting path from left to right with the root $r$ being the left endpoint. (\cref{fig:rightLeft} illustrates this viewpoint for the notions defined below.)
\begin{definition}[Left/right with respect to rooted WRAP instance]\label{def:leftRight}
Let $(G=(V,E),L,c,r,e_r)$ be a rooted WRAP instance and consider the consecutive numbering $V=\{r=v_0, v_1, \ldots, v_{|V|-1}\}$ of $V$ obtained by traversing the path $(V,E\setminus \{e_r\})$ starting from $r$. We use the following terminology:
\begin{itemize}
\item A vertex $v_i$ is said to be \emph{to the left} of $v_j$ if $i<j$ (and \emph{to the right} if $i>j$);
\item A directed link $\vec{\link{l}}=(v_i,v_j)\in V\times V$ is \emph{left-going} if $i>j$ (and \emph{right-going} if $i<j$);
\item For $v_i,v_j\in V$, the \emph{interval from $v_i$ to $v_j$} is the set $\{v_i,v_{i+1}, \ldots, v_j\}$, which is empty if $i>j$.
\end{itemize}
\end{definition}
\begin{figure}[!ht]
\begin{center}
\begin{tikzpicture}[scale=1,
ns/.style={thick,draw=black,fill=white,circle,minimum size=6,inner sep=2pt},
es/.style={thick},
lks/.style={line width=1pt, blue, densely dashed},
dlks/.style={lks, -latex},
ts/.style={every node/.append style={font=\scriptsize}}
]

\def\rad{2}
\def\num{12}
\def\hd{0.8}

\colorlet{col_r}{blue}

\colorlet{col_l}{orange}

\begin{scope}

\begin{scope}[every node/.style={ns}]
\foreach \i in {1,...,\num} {
  \pgfmathsetmacro\r{90+(\i-1)*360/\num}
  \node (\i) at (\r:\rad) {};
}
\end{scope}

\begin{scope}
\node at ($(1)+(-0.3,0)$)[above=2pt] {$r=v_1$};

\foreach \i in {2,...,\num} {
  \pgfmathsetmacro\r{90+(\i-1)*360/\num}
  \pgfmathsetmacro\radex{\rad+0.4}
  \node at (\r:\radex) {$v_{\i}$};
}

\path (\num) to node[above=-2pt] {$e_r$} (1);
\end{scope}

\begin{scope}[es]
\foreach \i in {1,...,\num} {
\pgfmathtruncatemacro\j{1+mod(\i,\num)}
\draw (\i) -- (\j);
}
\end{scope}

\begin{scope}[dlks]

\begin{scope}[col_r]
\draw (1)  -- (5);
\draw (3) -- (10);
\end{scope}

\begin{scope}[col_l]
\draw  (8) -- (4);
\draw (11) -- (6);
\end{scope}

\end{scope}

\end{scope}

\begin{scope}[shift={(3.5,0)}]

\begin{scope}[every node/.style={ns}]
\foreach \i in {1,...,\num} {
  \node (a\i) at (\i*\hd,0) {};
}
\end{scope}

\begin{scope}
\node at ($(a1)+(-0.3,0)$)[below=3pt] {$r=v_1$};
\pgfmathtruncatemacro\n{\num-1}
\foreach \i in {2,...,\n} {
\node at (a\i)[below=3pt] {$v_{\i}$};
}
\node at ($(a12)+(0.1,0)$)[below=3pt] {$v_{12}$};
\end{scope}

\begin{scope}[es]
\foreach \i in {2,...,\num} {
\pgfmathtruncatemacro\j{\i-1}
\draw (a\j) -- (a\i);
}
\draw (a1) to[out=-50,in=230,looseness=0.8] node[below,pos=0.8] {$e_r$} (a12);
\end{scope}

\begin{scope}[dlks]

\begin{scope}[col_r]
\draw (a1) to[bend left] (a5);
\draw (a3) to[bend left] (a10);
\end{scope}

\begin{scope}[col_l]
\draw  (a8) to[bend right] (a4);
\draw (a11) to[bend right] (a6);
\end{scope}

\end{scope}

\end{scope}

\end{tikzpicture}
 \end{center}
\caption{A rooted WRAP instance with four directed links and a representation of it as the path $(V,E\setminus \{e_r\})$ going from $r$ as the left end of the path to the right end $v_{12}$, together with $e_r$.
Two of the four directed links are right-going (drawn in blue) and two are left-going (drawn in orange).}
\label{fig:rightLeft}
\end{figure}
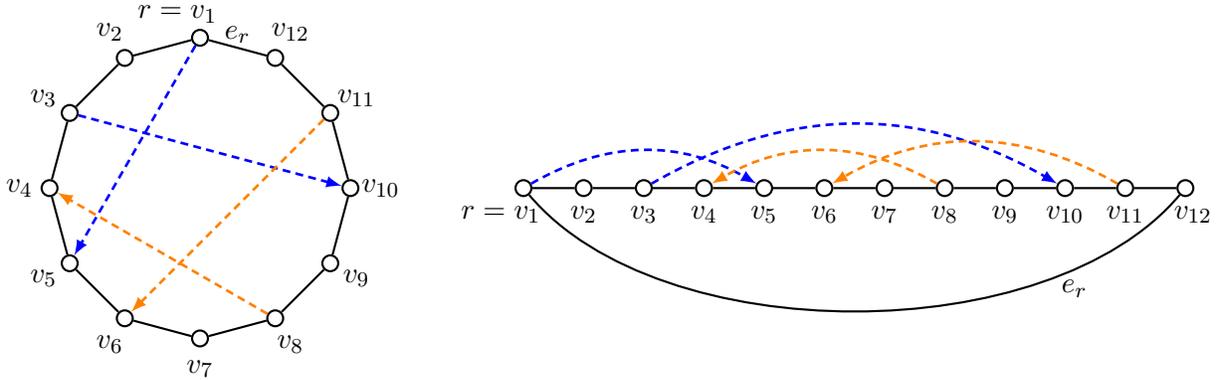
For a vertex $v\in V$, we call a directed link \emph{left-outgoing} (\emph{right-outgoing}) of $v$ if its tail is $v$ and it is left-going (right-going).
Note that the root only has right-outgoing links.
Moreover, we say that a vertex $s\in V$ lies \emph{between two vertices $u,v\in V$}, if it is on the unique $u$-$v$ path in $(V,E\setminus \{e_r\})$; in particular we may have $s\in \{u,v\}$.
To exclude the possibility that $s\in \{u,v\}$, we say that $s$ lies \emph{strictly between $u$ and $v$}.

We are now ready to prove \cref{thm:ov_StructureNonShortenable}, our main structural statement about non-shortenable directed WRAP solutions, which we repeat below for convenience.
\propNonShortenable*
Note that~\ref{item:noTwoSameDir} implies $|\delta^+_{\vec{F}}(v)|\leq 2$ for all $v\in V$, as there can be at most one left-going and one right-going link in $\delta^+_{\vec{F}}(v)$.
Moreover, $|\delta^+_{\vec{F}}(r)|=1$, because any directed link with tail being the root is right-going.
\begin{proof}[Proof of \cref{thm:ov_StructureNonShortenable}]
We first show that any non-shortenable directed solution $\vec{F}\subseteq \shadows(L)$ to a rooted WRAP instance is $G$-planar.
To this end, we split the ways how two directed links $\vec{\link{l}}_1=(u_1,v_1)$ and $\vec{\link{l}}_2=(u_2,v_2)$ can cross into three different types, highlighted in \cref{fig:GplanProof}.
More formally, the types are defined with respect to the order of appearance of the four endpoints of $\vec{\link{l}}_1$ and $\vec{\link{l}}_2$ when traversing the ring from the root to the right:
\begin{enumerate}[itemsep=0.0em,leftmargin=2cm,topsep=0.3em]
\item[type~1:] $u_1$, $v_2$, $v_1$, $u_2$,
\item[type~2:] $v_2$, $u_1$, $u_2$, $v_1$,
\item[type~3:] $v_2$, $v_1$, $u_2$, $u_1$.
\end{enumerate}
There is one more order of endpoints of $\vec{\link{l}}_1$ and $\vec{\link{l}}_2$, up to exchanging their roles, that leads to $\vec{\link{l}}_1$ and $\vec{\link{l}}_2$ being crossing, namely when reversing the directions of both links in a type~3 crossing.
However, this case is symmetric to type~3 crossings, and we therefore do not list it separately.
For type~1 crossings, we show that this constellation is impossible even if $v_1=v_2$ (even though this does not lead to crossing links).
Moreover, for type~3 crossing, we show impossibility of the constellation even if $v_1=v_2$ and/or $u_1=u_2$.
Ruling out constellations with $v_1=v_2$ in type~1 and type~3 crossings implies $|\delta^-_{\vec{F}}(v)|\leq 1$ on top of $G$-planarity.
Moreover, showing that type~3 constellations with $u_1=u_2$ are impossible, immediately implies~\ref{item:noTwoSameDir} of \cref{thm:ov_StructureNonShortenable}.
Hence, whenever we refer to type~1 or type~3 crossings, we do not require $v_1\neq v_2$ and, for type~3, we also do not require $u_1\neq u_2$.

Assume for the sake of deriving a contradiction that there is a type~1 crossing.
Because $\vec{F}$ is non-shortenable, there is a $2$-cut $C_1\subseteq V$ that is covered by $\vec{\link{l}}_1$, i.e., $\vec{\link{l}}_1\in \delta_{\vec{F}}^-(C_1)$, and no other link in $\vec{F}$ or any strict shortening of $\vec{\link{l}}_1$ is covering $C_1$.
In particular,  $\vec{\link{l}}_2$ is not covering $C_1$.
Because either $v_1=v_2$ or $(v_2,v_1)$ is a strict shortening of $\vec{\link{l}}_1$ and hence also not covering $C_1$, we conclude $v_1,v_2,u_2\in C_1$ and $u_1\not\in C_1$ (see \cref{fig:GplanProof}).
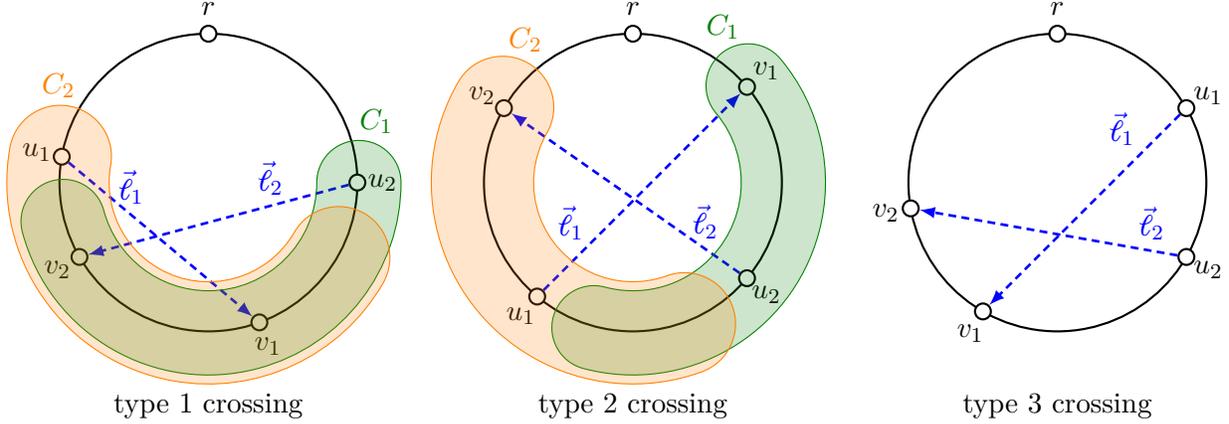
\begin{figure}[!ht]
\begin{center}
\begin{tikzpicture}[scale=0.99,
ns/.style={thick,draw=black,fill=white,circle,minimum size=6,inner sep=2pt},
es/.style={thick},
lks/.style={line width=1pt, blue, densely dashed},
dlks/.style={lks, -latex},
ts/.style={every node/.append style={font=\scriptsize}}
]

\def\rad{2}
\def\vlfac{1.17}

\pgfmathsetmacro\picsep{2*\rad + 1.7}

\pgfkeyssetvalue{/tikz/pics/interval/color}{green!50!black}
\pgfkeyssetvalue{/tikz/pics/interval/radius}{0.3*\rad}
\tikzset{
    pics/interval/.style 2 args={
        code={
               \def\cw{\pgfkeysvalueof{/tikz/pics/interval/radius}}
               \colorlet{col}{\pgfkeysvalueof{/tikz/pics/interval/color}}

               \draw[col, fill=col, fill opacity=0.2] (#1:\rad+\cw) arc (#1:#2:\rad+\cw)
               arc (#2:#2+180:\cw)
               arc (#2:#1:\rad-\cw)
               arc (#1+180:#1+360:\cw);
        }
    },
}

\begin{scope}
\coordinate (c) at (0,0);
\draw[thick] (0,0) circle [radius = \rad];

\begin{scope}[every node/.style={ns}]
\node (r) at (90:\rad) {};

\node (u1) at (170:\rad) {};
\node (v1) at (290:\rad) {};

\node (u2) at (0:\rad) {};
\node (v2) at (210:\rad) {};
\end{scope}

\begin{scope}
\node at ($(c)!\vlfac!(r)$) {$r$};

\node at ($(c)!\vlfac!(u1)$) {$u_1$};
\node at ($(c)!\vlfac!(v1)$) {$v_1$};

\node at ($(c)!\vlfac!(u2)$) {$u_2$};
\node at ($(c)!\vlfac!(v2)$) {$v_2$};
\end{scope}

\begin{scope}[dlks]
\draw (u1) to node[pos = 0.34, above] {$\vec{\link{l}}_1$} (v1);
\draw (u2) to node[pos = 0.3, above] {$\vec{\link{l}}_2$} (v2);
\end{scope}

\begin{scope}
\path (0,0) pic[pics/interval/radius=0.28*\rad] {interval={195}{360}};
\path (0,0) pic[pics/interval/radius=0.34*\rad,pics/interval/color=orange] {interval={170}{330}};
\end{scope}

\begin{scope}
\node[green!50!black] at (20:\rad+0.4) {$C_1$};
\node[orange] at (147:\rad+0.4) {$C_2$};
\end{scope}

\begin{scope}
\node at (270:\rad)[below=7mm] {type 1 crossing};
\end{scope}

\end{scope}

\begin{scope}[xshift=\picsep cm]
\coordinate (c) at (0,0);
\draw[thick] (0,0) circle [radius = \rad];

\begin{scope}[every node/.style={ns}]
\node (r) at (90:\rad) {};

\node (u1) at (230:\rad) {};
\node (v1) at (40:\rad) {};

\node (u2) at (-40:\rad) {};
\node (v2) at (150:\rad) {};
\end{scope}

\begin{scope}
\node at ($(c)!\vlfac!(r)$) {$r$};

\node at ($(c)!\vlfac!(u1)$) {$u_1$};
\node at ($(c)!\vlfac!(v1)$) {$v_1$};

\node at ($(c)!\vlfac!(u2)$) {$u_2$};
\node at ($(c)!\vlfac!(v2)$) {$v_2$};
\end{scope}

\begin{scope}[dlks]
\draw (u1) to node[pos = 0.23, above left=-3pt] {$\vec{\link{l}}_1$} (v1);
\draw (u2) to node[pos = 0.15, above] {$\vec{\link{l}}_2$} (v2);
\end{scope}

\begin{scope}
\path (0,0) pic[pics/interval/radius=0.28*\rad] {interval={255}{400}};
\path (0,0) pic[pics/interval/radius=0.34*\rad,pics/interval/color=orange] {interval={150}{290}};
\end{scope}

\begin{scope}
\node[green!50!black] at (60:\rad+0.42) {$C_1$};
\node[orange] at (127:\rad+0.4) {$C_2$};
\end{scope}

\begin{scope}
\node at (270:\rad)[below=7mm] {type 2 crossing};
\end{scope}

\end{scope}

\begin{scope}[xshift=2*\picsep cm]
\coordinate (c) at (0,0);
\draw[thick] (0,0) circle [radius = \rad];

\begin{scope}[every node/.style={ns}]
\node (r) at (90:\rad) {};

\node (u1) at (30:\rad) {};
\node (v1) at (240:\rad) {};

\node (u2) at (-30:\rad) {};
\node (v2) at (190:\rad) {};
\end{scope}

\begin{scope}
\node at ($(c)!\vlfac!(r)$) {$r$};

\node at ($(c)!\vlfac!(u1)$) {$u_1$};
\node at ($(c)!\vlfac!(v1)$) {$v_1$};

\node at ($(c)!\vlfac!(u2)$) {$u_2$};
\node at ($(c)!\vlfac!(v2)$) {$v_2$};
\end{scope}

\begin{scope}[dlks]
\draw (u1) to node[pos = 0.2, above left=-2pt] {$\vec{\link{l}}_1$} (v1);
\draw (u2) to node[pos = 0.1, above] {$\vec{\link{l}}_2$} (v2);
\end{scope}

\begin{scope}
\node at (270:\rad)[below=7mm] {type 3 crossing};
\end{scope}

\end{scope}

\end{tikzpicture}
 \end{center}
\caption{Visualization of the three different types in which directed links can cross.}
\label{fig:GplanProof}
\end{figure}
An analogous reasoning for $\vec{\link{l}}_2$ leads to the existence of a $2$-cut $C_2$ with $\delta^-_{\vec{F}}(C_2)=\{\vec{\link{l}}_2\}$ satisfying $v_1,v_2,u_1\in C_2$ and $u_2\not\in C_2$ (see \cref{fig:GplanProof}).
We now consider the $2$-cut $C_1\cup C_2$.
Note that this is a $2$-cut because the intervals $C_1$ and $C_2$ overlap.
Because any directed link that enters $C_1\cup C_2$ must also enter $C_1$ or $C_2$, we obtain
\begin{equation*}
\delta_{\vec{F}}^-(C_1\cup C_2) \subseteq \delta_{\vec{F}}^-(C_1) \cup \delta_{\vec{F}}^-(C_2) = \{\vec{\link{l}}_1,\vec{\link{l}}_2\}\enspace.
\end{equation*}
However, as $\vec{\link{l}}_1,\vec{\link{l}}_2\not\in C_1\cup C_2$, the above relation implies $\delta^-_{\vec{F}}(C_1\cup C_2)=\emptyset$, which violates that $\vec{F}$ is a directed WRAP solution.

We now assume for the sake of deriving a contradiction that there is a type~2 crossing.
We follow the same reasoning as in the type~1 case.
Let $C_1$ be a $2$-cut with $\vec{\link{l}}_1\in \delta^-_{\vec{F}}(C_1)$ and such that no other link of $\vec{F}$ and neither any shortening of $\vec{\link{l}}_1$ enters $C_1$.
The interval $C_1$ thus starts at the right neighbor of $u_1$ and expands to the right until (and including) at least $v_1$ (see \cref{fig:GplanProof}).
Analogously, there is a $2$-cut $C_2$ with $\delta_{\vec{F}}^-(C_2) = \{\vec{\link{l}}_2\}$ that starts to the left of $v_2$ and ends at the left neighbor of $u_2$ (see \cref{fig:GplanProof}).
Note that $C_1\cup C_2$ is again a $2$-cut because
\begin{enumerate*}
\item $u_1$ is strictly to the left of $u_2$,
\item the left endpoint of $C_1$ is the right neighbor of $u_1$, and
\item the right endpoint of $C_2$ is the left neighbor of $u_2$.
\end{enumerate*}
As before, we get
\begin{equation*}
\delta_{\vec{F}}^-(C_1\cup C_2) \subseteq \delta_{\vec{F}}^-(C_1) \cup \delta_{\vec{F}}^-(C_2) = \{\vec{\link{l}}_1,\vec{\link{l}}_2\}\enspace,
\end{equation*}
which, because $\vec{\link{l}}_1,\vec{\link{l}}_2\not\in C_1\cup C_2$, implies $\delta^-_{\vec{F}}(C_1\cup C_2)=\emptyset$, contradicting that $\vec{F}$ is a directed WRAP solution.

Finally, a type~3 crossing is impossible, because one could shorten $\vec{\link{l}}_2$ to $(v_1,v_2)$ (or simply delete $\vec{\link{l}}_2$ if $v_1=v_2$).
Indeed, this leads to another directed WRAP solution as any $2$-cut that is covered by $\vec{\link{l}}_1$ or $\vec{\link{l}}_2$ is also covered by either $\vec{\link{l}}_1$ or $(v_1,v_2)$.
To verify that this is the case, one could first replace $\vec{\link{l}}_2$ by the two links $(u_2,v_1)$ and $(v_1,v_2)$, which cover all $2$-cuts that $\vec{\link{l}}_2$ covers because these two links form a path from $u_2$ to $v_2$.
In a second step, one can delete $(u_2,v_1)$ while maintaining a directed WRAP solution, because $(u_2,v_1)$ is a shortening of $\vec{\link{l}}_1$ and therefore only covers $2$-cuts that are also covered by $\vec{\link{l}}_1$.

Hence, we showed that $\vec{F}$ is $G$-planar, $|\delta^-_{\vec{F}}(v)|\leq 1$ for all $v\in V$, and that \ref{item:noTwoSameDir} of \cref{thm:ov_StructureNonShortenable} holds.
It remains to prove that $\vec{F}$ is an $r$-arborescence.
To this end, we start by observing that the in-degrees correspond to those of an $r$-arborescence, i.e., $|\delta_{\vec{F}}^-(v)|=1$ for $v\in V\setminus \{r\}$ and $|\delta_{\vec{F}}^-(r)|=0$.
Indeed, the upper bound of $1$ on the in-degrees together with the fact that every singleton $\{v\}$ for $v\in V \setminus \{r\}$ is a $2$-cut, implies $|\delta_{\vec{F}}^-(v)|=1$ for $v\in V\setminus\{r\}$.
Moreover, $\delta^-(r)=\emptyset$, because any directed link entering $r$ can be deleted as it does not cover any $2$-cut; thus, non-shortenable solutions do not contain such links.

To show that $\vec{F}$ is an $r$-arborescence, we are left to prove that $\vec{F}$ does not contain any cycles (when disregarding orientations).
Suppose for the sake of deriving a contradiction that $\vec{F}$ contains a cycle $\vec{Q} \subseteq \vec{F}$, and let $U\subseteq V$ be the vertices on this cycle.
Because the in-degree of each vertex is at most one, this must be a directed cycle. 
Let $u_1, u_2 \in U$ be the left-most and right-most vertex in $U$, respectively, and let $C\subseteq V$ be the interval from $u_1$ to $u_2$ (we recall that $C$ includes the endpoints $u_1$ and $u_2$). See \cref{fig:noCyclesInNonShortenableSols} for an illustration.
\begin{figure}[!ht]
\begin{center}
\begin{tikzpicture}[scale=1,
ns/.style={thick,draw=black,fill=white,circle,minimum size=6,inner sep=2pt},
es/.style={thick},
lks/.style={line width=1pt, blue, densely dashed},
dlks/.style={lks, -latex},
ts/.style={every node/.append style={font=\scriptsize}}
]

\def\rad{2}
\def\vlfac{1.17}

\pgfkeyssetvalue{/tikz/pics/interval/color}{green!50!black}
\pgfkeyssetvalue{/tikz/pics/interval/radius}{0.3*\rad}
\tikzset{
    pics/interval/.style 2 args={
        code={
               \def\cw{\pgfkeysvalueof{/tikz/pics/interval/radius}}
               \colorlet{col}{\pgfkeysvalueof{/tikz/pics/interval/color}}

               \draw[col, fill=col, fill opacity=0.2] (#1:\rad+\cw) arc (#1:#2:\rad+\cw)
               arc (#2:#2+180:\cw)
               arc (#2:#1:\rad-\cw)
               arc (#1+180:#1+360:\cw);
        }
    },
}

\begin{scope}
\coordinate (c) at (0,0);
\draw[thick] (0,0) circle [radius = \rad];

\begin{scope}[every node/.style={ns}]
\node (r) at (90:\rad) {};

\node (u1) at (130:\rad) {};
\node (u2) at (50:\rad) {};

\node (v1) at (150:\rad) {};
\node (v2) at (205:\rad) {};
\node (v3) at (300:\rad) {};
\node (v4) at (370:\rad) {};
\end{scope}

\begin{scope}
\node at ($(c)!\vlfac!(r)$) {$r$};

\node at ($(c)!\vlfac!(u1)$) {$u_1$};
\node at ($(c)!\vlfac!(u2)$) {$u_2$};
\end{scope}

\begin{scope}[dlks]
\draw (u1) to (v1);
\foreach \i in {1,...,3} {
\pgfmathtruncatemacro\j{\i+1}
\draw (v\i) to (v\j);
}
\draw (v4) to (u2);
\draw (u2) to (u1);
\end{scope}

\begin{scope}
\node[blue] at (105:\rad-0.80) {$\vec{Q}$};
\end{scope}

\begin{scope}
\path (0,0) pic {interval={130}{410}};
\end{scope}

\begin{scope}
\node[green!50!black] at (15:\rad-0.93) {$C$};
\end{scope}

\end{scope}

\end{tikzpicture}
 \end{center}
\caption{Illustration of a directed $G$-planar cycle with $u_1$ being its left-most vertex and $u_2$ its right-most vertex (the orientation of the cycle can be either way).
$C$ is the interval from $u_1$ to $u_2$, including the endpoints.
Any directed link $(s,t)$ with $t\not\in \{u_1,u_2\}$ that enters $C$ must cross the link $(u_1,u_2)$ of the cycle.
}
\label{fig:noCyclesInNonShortenableSols}
\end{figure}
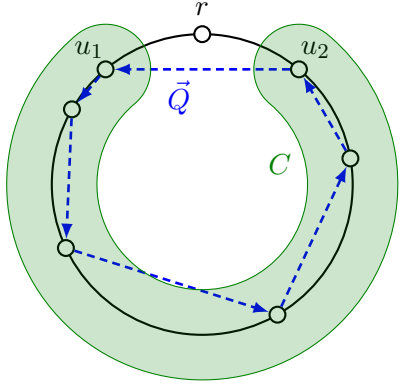
Note that because $\vec{Q}$ is $G$-planar, it either contains the arc $(u_1,u_2)$ or $(u_2,u_1)$.
As $\vec{F}$ is a directed WRAP solution, it contains a directed link $\vec{\link{l}}=(s,t)\in \delta^-_{\vec{F}}(C)$.
The head $t$ of $\vec{\link{l}}$ cannot be in $U$, as each vertex of $U$ already has an incoming link from $\vec{Q}$, and the in-degree of $t$ is one.
Hence, we have in particular $t\not\in \{u_1,u_2\}$, which implies that the link $(s,t)$ is crossing the link in $\vec{Q}$ that connects $u_1$ and $u_2$, thus contradicting $G$-planarity of $\vec{F}$.
\end{proof}

From \cref{thm:ov_StructureNonShortenable} we can quite easily derive the following further basic properties of non-shortenable directed WRAP solutions.
\begin{lemma}\label{lem:furtherNonshortProps}
For any rooted WRAP instance $(G=(V,E),L,c,r, e_r)$ and non-shortenable directed solution $\vec{F} \subseteq \shadows(L)$ thereof, we have:
\begin{enumerate}
\item\label{item:descendants_interval} For $v\in V$, the set of descendants of $v$ in $(V,\vec{F})$ is an interval.
\item\label{item:lca_is_between} For $v_1,v_2\in V$, the least common ancestor of $v_1$ and $v_2$ in $(V,\vec{F})$ lies between $v_1$ and $v_2$.
\end{enumerate}
\end{lemma}
\begin{proof}
\begin{enumerate}[wide=\parindent,itemsep=0.5em]
\item We prove the statement by induction on the distance of $v$ from the root $r$ in the arborescence $(V,\vec{F})$ (going from larger to smaller distance).
If $v$ has a right-outgoing link $(v,w)\in \vec{F}$, then the set of descendants of $w$ is an interval $I_w$. 
Then, because $v\notin I_w$ and $v$ is left of $w$, $v$ is left of the interval $I_w$.
Suppose, for the sake of deriving a contradiction that there exists a vertex that is right of $v$ and left of all the vertices in $I_w$.
Among all such vertices let $t$ be one that is as close to $r$ as possible in the arborescence $(V,\vec{F})$.
We have $t\ne r$ because $v$ is left of $t$ and $r$ is the leftmost vertex in $V$.
Therefore, there exists a directed link $(s,t)\in \vec{F}$.
By \cref{thm:ov_StructureNonShortenable}~\ref{item:noTwoSameDir}, we have $s \ne v$ because $t$ and $w$ are both right of $v$ and $(v,w)\in \vec{F}$.
Moreover, $t\notin I_w$ implies $s\notin I_w$ because $I_w$ is the set of descendants of $w$ and hence does not have an outgoing link in $\vec{F}$.
Finally, we observe that our choice of $t$ implies that $s$ is not strictly between $v$ and $I_w$, i.e., it is not both right of $v$ and left of $I_w$.
We conclude that $s$ is either left of $v$ or right of $I_w$. 
In both of these cases, the links $(v,w)\in \vec{F}$ and $(s,t)\in \vec{F}$ are crossing, contradicting \cref{thm:ov_StructureNonShortenable}~\ref{item:Gplanar}.
This shows that $\{v\} \cup I_w$ is an interval.
A symmetric argument shows that if $v$ has a left-outgoing link $(v,u)$ and $I_u$ is the set of descendants of $u$, then $\{v \} \cup I_u$ is an interval.

If $v$ has no left-outgoing or no right-outgoing link in $\vec{F}$, let $I_u$ or $I_w$, respectively, be the empty set and otherwise let them be defined as above.
By \cref{thm:ov_StructureNonShortenable}~\ref{item:noTwoSameDir}, the set of descendants of $v$ is $I_u \cup \{v\} \cup I_w$.
This is an interval because both $\{v\} \cup I_u$ and $\{v\} \cup I_w$ are intervals.

\item Let $s$ be the least common ancestor of $v_1$ and $v_2$ in $(V, \vec{F})$.
We assume $s \notin \{v_1,v_2\}$; otherwise $s$ is between $v_1$ and $v_2$, as claimed.
Furthermore, let $(s,t_1)$ be the first link of the $s$-$v_1$ path in $\vec{F}$ and let $(s,t_2)$ be the first link of the $s$-$v_2$ path in $\vec{F}$.
By \cref{thm:ov_StructureNonShortenable}~\ref{item:noTwoSameDir}, the links $(s,t_1)$ and $(s,t_2)$ go in opposite directions, say $t_1$ is left of $s$ and $t_2$ is right of $s$.
Applying \cref{lem:furtherNonshortProps}~\ref{item:descendants_interval} to $t_1$ and $t_2$, we get that the descendants of $t_1$ form an interval and so do the descendants of $t_2$.
Because $s$ is a strict ancestor of both $t_1$ and $t_2$, these intervals do not contain $s$.
Hence, all descendants of $t_1$ are left of $s$, whereas all descendants of $t_2$ are right of $s$.
In particular, $v_1$ is left of $s$ and $v_2$ is right of $s$, implying that $s$ lies between $v_1$ and $v_2$, as claimed.
\end{enumerate}
\end{proof}
 \section{Dropping Links from Directed Solutions}\label{sec:drop}

In this section we characterize the set of links that we can drop from a non-shortenable directed WRAP solution when adding some undirected links.
This will be needed both in the relative greedy algorithm presented in \cref{sec:thin-components} and the local search algorithm in \cref{sec:local-search}.
Given a non-shortenable directed WRAP solution $\vec{F}$, we start by defining, for every cut $C\in \mathcal{C}_G$, a directed link $\vec{\link{l}}$ that is responsible for $C$.
In our relative greedy algorithm, we will remove $\vec{\link{l}}$ as soon as we added undirected links covering all cuts $\vec{\link{l}}$ is responsible for.

\begin{definition}\label{def:responsible}
Let $(G,L,c,r, e_r)$ be a rooted WRAP instance and $\vec{F} \subseteq \shadows(L)$ be a non-shortenable directed solution thereof.
We say that a directed link $(u,v)\in \vec{F}$ is \emph{responsible} for a 2-cut $C\in \Cscr_G$ if  $(u,v)$ covers $C$ and no link on the $r$-$u$ path in the arborescence $(V,\vec{F})$ covers $C$.
We denote the set of 2-cuts for which $\vec{\link{l}}\in \vec{F}$ is responsible by
\[
 \Rscr_{\vec{F}}(\vec{\link{l}}) \coloneqq \Bigl\{ C\in \Cscr_G : \vec{\link{l}}\text{ is responsible for }C \Bigr\}.
\]
\end{definition}

Because $\vec{F}$ is a directed solution, for every cut $C\in \Cscr_G$ there is at least  one directed link $\vec{\link{l}}\in \vec{F}$ that is responsible for $C$.
We now show the stronger statement that for every 2-cut $C\in \Cscr_G$ there is \emph{exactly one} link $\vec{\link{l}}\in \vec{F}$ that is responsible for $C$.
See \cref{fig:responsible_cut} for an illustration of the proof.

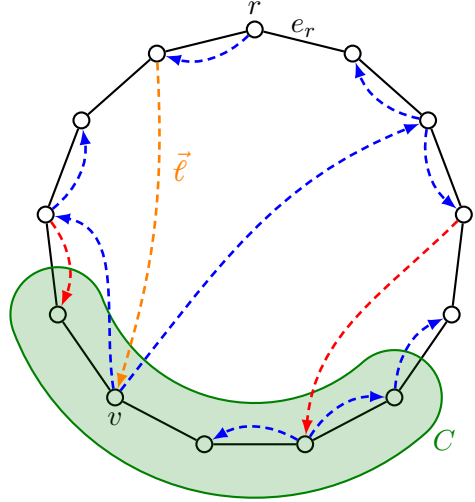
\begin{figure}[!ht]
\begin{center}
\begin{tikzpicture}[scale=1.4,
ns/.style={thick,draw=black,fill=white,circle,minimum size=6,inner sep=2pt},
es/.style={thick},
lks/.style={line width=1pt, blue, densely dashed},
dlks/.style={lks, -latex},
ts/.style={every node/.append style={font=\scriptsize}}
]

\def\rad{2}
\def\num{13}

\begin{scope}[every node/.style={ns}]
\foreach \i in {1,...,\num} {
  \pgfmathsetmacro\r{90+(\i-1)*360/\num}
  \node (\i) at (\r:\rad) {};
}
\end{scope}

\begin{scope}
\node at (1)[above=2pt] {$r$};
\path (\num) to node[above=-2pt] {$e_r$} (1);
\node at (6) [below=2pt] {$v$};
\end{scope}

\begin{scope}[es]
\foreach \i in {1,...,\num} {
\pgfmathtruncatemacro\j{1+mod(\i,\num)}
\draw (\i) -- (\j);
}
\end{scope}

\def\cw{0.45}
\begin{scope}[every path/.style={fill=green!50!black, fill opacity=0.2,draw=green!50!black, thick}]

\newcommand\Interval[2]{
  \pgfmathsetmacro{\ra}{90+(#1-1)*360/\num}
  \pgfmathsetmacro{\rb}{90+(#2-1)*360/\num}
\draw (\ra:\rad+\cw) arc (\ra:\rb:\rad+\cw)
  arc (\rb:\rb+180:\cw)
  arc (\rb:\ra:\rad-\cw)
  arc (\ra+180:\ra+360:\cw)
;
}

\Interval{5}{9}
\end{scope}
\node[green!50!black] () at (1.8,-1.9) {$C$};
\node[orange] () at (-0.7,0.7) {$\vec{\link{l}}$};

\begin{scope}[dlks]
\draw (1) to[bend left] (2);
\draw[orange] (2) to[bend left=10] (6);
\draw (6) to[bend right, in=-120, out=-10] (4);
\draw (4) to[bend right] (3);
\draw[red] (4) to[bend left] (5);
\draw (6) to[bend left, in=160, out=10] (12);
\draw (12) to[bend left] (13);
\draw (12) to[bend right] (11);
\draw[red] (11) to[bend right, out=-10] (8);
\draw (8) to[bend right] (7);
\draw (8) to[bend left] (9);
\draw (9) to[bend left] (10);
\end{scope}

\end{tikzpicture}
 \end{center}
\caption{The dashed edges show a non-shortenable directed WRAP solution.
A cut $C\in \Cscr_G$ is highlighted in green and two directed links that are both covering the cut~$C$ are highlighted in red. 
The vertex $v$ is the least common ancestor of all vertices in $C$ and is contained in $C$ itself.
Thus, directed link $\vec{\link{l}}$ entering $v$ (shown in orange) covers the cut $C$, which implies that the red links  are not responsible for~$C$.}
\label{fig:responsible_cut}
\end{figure}

\begin{lemma}\label{lem:responsibility_unique}
Let $(G,L,c,r, e_r)$ be a rooted WRAP instance and $\vec{F} \subseteq \shadows(L)$ be a non-shortenable directed solution thereof. 
Then the sets $\Rscr_{\vec{F}}(\vec{\link{l}})$ with $\vec{\link{l}}\in \vec{F}$ form a partition of $\Cscr_G$.
\end{lemma}
\begin{proof}
We have to show that, for every $2$-cut $C\in \mathcal{C}_G$, precisely one directed link $\vec{\link{l}}$ is responsible for $C$.
By \cref{lem:furtherNonshortProps}~\ref{item:lca_is_between}, the least common ancestor of any pair of vertices in $C$ lies in $C$.
Hence, $C$ contains a vertex $v\in C$ that is ancestor of all other vertices of $C$.
As $\vec{F}$ is an $r$-arborescence, there is a unique directed link $\vec{\link{l}}\in \vec{F}$ that enters $v$, which also enters $C$.
Because $v$ is ancestor to all other vertices of $C$, any other link $(s,t)$ entering $C$ fulfills that $\vec{\link{l}}$ is on the unique $r$-$s$ path in $(V,\vec{F})$.
Hence, $\vec{\link{l}}$ is responsible for $C$ and it is the only link responsible for $C$.
\end{proof}

For a non-shortenable directed WRAP solution $\vec{F}$ and a set $K\subseteq L$, we now define the set  $\Drop_{\vec{F}}(K)$ of directed links that we will remove from $\vec{F}$ when adding $K$.

\begin{definition}\label{def:Drop}
Let $(G,L,c,r, e_r)$ be a rooted WRAP instance and $\vec{F} \subseteq \shadows(L)$ be a non-shortenable directed solution thereof.
Then, for a link set $K\subseteq L$, we define $\Drop_{\vec{F}}(K)\subseteq \vec{F}$ by
\[
   \Drop_{\vec{F}}(K) \coloneqq \Bigl\{ \vec{\link{l}}\in \vec{F}: \text{ for all }C\in \Rscr_{\vec{F}}(\vec{\link{l}})\text{ we have } \delta_K(C) \ne \emptyset \Bigr\},
\]
i.e., a directed link $\vec{\link{l}}\in \vec{F}$ belongs to $\Drop_{\vec{F}}(K)$ if all cuts $\vec{\link{l}}$ is responsible for are covered by $K$.
\end{definition}

The fact that for every cut $C\in \Cscr_G$ there is a link $\vec{\link{l}}$ that is responsible for $C$, implies that the resulting link set $(\vec{F} \cup K) \setminus \Drop_{\vec{F}}(K)$ is a mixed WRAP solution.

In the remainder of this section we provide a useful characterization for when a link $\vec{\link{l}}\in \vec{F}$ is contained in $\Drop_{\vec{F}}(K)$.
For this, we need the notion of the \emph{link intersection graph}; similar concepts have previously been used in the context of CAP and WCAP in \cite{basavaraju_2014_parameterized, byrka_2020_breaching, nutov_2021_approximation, angelidakis2021node}.
Recall that two links in $L$ are called \emph{intersecting} if they are crossing or have a common endpoint.
An example is shown in \cref{fig:linkIntersection}.
For a fixed WRAP instance $(G,L,c)$, the link intersection graph $H$ is the graph with vertex set $L$ that contains an edge between two links if and only if these links are intersecting (\cref{def:linkIntersectionGraph}).
Moreover, for a set $K\subseteq L$, the link intersection graph of $K$ is the subgraph $H[K]$ of $H$ induced by $K$.

\begin{figure}[!ht]
\begin{center}
\begin{tikzpicture}[scale=1,
ns/.style={thick,draw=black,fill=white,circle,minimum size=6,inner sep=2pt},
es/.style={thick},
lks/.style={line width=1pt, densely dashed},
dlks/.style={lks, -latex},
ts/.style={every node/.append style={font=\scriptsize}}
]

\def\rad{2}
\def\num{12}

\begin{scope}[every node/.style={ns}]
\foreach \i in {1,...,\num} {
  \pgfmathsetmacro\r{90+(\i-1)*360/\num}
  \node (\i) at (\r:\rad) {};
}
\end{scope}

\begin{scope}[es]
\foreach \i in {1,...,\num} {
\pgfmathtruncatemacro\j{1+mod(\i,\num)}
\draw (\i) -- (\j);
}
\end{scope}

\begin{scope}[lks]
\draw[blue] (11) -- node[below,pos=0.5] {$\link{l}_1$} (4);
\draw[red!70!black] (2) -- node[right,pos=0.7] {$\link{l}_2$} (6);
\draw[green!40!black](7) -- node[below right=-3pt,pos=0.5] {$\link{l}_3$} (11);
\end{scope}

\end{tikzpicture}
 \end{center}
\caption{Example illustrating the definition of intersecting links. The link $\link{l}_1$ intersects both $\link{l}_2$ and $\link{l}_3$, but $\link{l}_2$ does not intersect $\link{l}_3$.}
\label{fig:linkIntersection}
\end{figure}
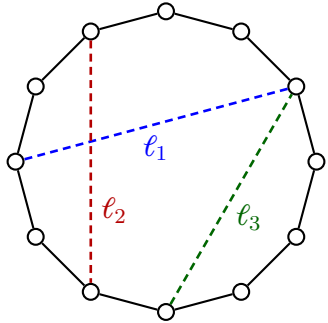

A graph that can arise as the intersection graph of segments of a circle, like link intersection graphs of WRAP instances, is called a \emph{circle graph}. 
The class of circle graphs has been intensively studied in Graph Theory under a variety of aspects, including chromatic numbers~\cite{garey_1980_complexity,unger_1988_k-colouring,gyarfas_1985_chromatic,davies_2021_circle}, cliques~\cite{gavril_1973_algorithms, tiskin_2015_fast}, independent sets~\cite{gavril_1973_algorithms,nash_2010_output}, dominating sets~\cite{keil_1993_complexity}, recognition~\cite{spinrad_1994_recognition}, and treewidth~\cite{kloks_1996_treewidth}. (Due to the extensive work in the field, the above list is highly incomplete; we suggest the interested reader to also check the work cited in the above references.)

Recall from \cref{def:ConnectedInLinkIntersectionGraph} that, for $K\subseteq L$, we say that a vertex $u\in V$ is \emph{connected to} a vertex $v\in V$ \emph{in the link intersection graph $H[K]$ of $K$} if there is a path in $H[K]$ from a link incident to $u$ to a link incident to a $v$.
With this definition we can characterize when a link set $K\subseteq L$ covers a cut $C\subseteq V\setminus \{r\}$.

\begin{lemma}\label{lem:paths_in_intersection_graphs_and_cuts}
Let $(G,L,c,r, e_r)$ be a rooted WRAP instance and $K\subseteq L$.
Moreover, let $C \subseteq V\setminus \{r\}$ be a cut.
Then $\delta_K(C) \neq \emptyset$ if and only if  there exists a vertex $v\in C$  that is connected to a vertex $w \in V\setminus C$ in the link intersection graph of $K$.
\end{lemma}
\begin{proof}
If there exists a link $\link{l}=\{v,w\}\in \delta_K(C)$, where without loss of generality $v\in C$ and $w\notin C$, then the path $H[\{\link{l}\}]$ connects $v$ to $w$ in the link intersection graph $H[K]$.

Now suppose there exists a path $P$ in $H[K]$ from a link $\link{l}$ incident to $v \in C$ to a link $\link{f}$ incident to a vertex $w\in V\setminus C$.
Suppose for the sake of deriving a contradiction that $\delta_K(C) = \emptyset$.
Then every link in $K$ has either both endpoints in $C$ or both endpoints outside of $C$.
Because $v\in C$, the link $\link{l}$ has both endpoints in $C$ and because $w\notin C$, the link $\link{f}$ has both endpoints outside $C$.
Therefore, the path $P$ has an edge $\{\link{l}', \link{f}'\}$ such that $\link{l}'$ has both endpoints inside $C$ and $\link{f}'$ has both endpoints outside $C$.
Because $C$ is a $2$-cut of $G$ and hence an interval, this implies that the links $\link{l}'$ and $\link{f}'$ do not intersect, which contradicts the fact that $\{\link{l}', \link{f}'\}$ is an edge of the intersection graph $H[K]$ of $K$.
\end{proof}

We will show that a directed link $(u,v)\in \vec{F}$ is contained in $\Drop_{\vec{F}}(K)$ if and only if $v$ is connected to a \emph{$v$-good} vertex in the link intersection graph of $K$. 
(Recall that \cref{fig:characterizeDrop} shows an example.)

\begin{definition}\label{def:v-good}
Let $(G=(V,E),L,c,r, e_r)$ be a rooted WRAP instance and $\vec{F} \subseteq \shadows(L)$ be a non-shortenable directed solution thereof.
Then a vertex $u\in V$ is called \emph{$v$-bad} if $u$ is a descendant of $v$ in the arborescence $(V,\vec{F})$.
Otherwise, $u$ is called \emph{$v$-good}.
In particular, $v$ itself is $v$-bad.
\end{definition}

Note that by \cref{lem:furtherNonshortProps}~\ref{item:descendants_interval} the set of $v$-bad vertices is an interval.
The notion of $v$-bad vertices allows us to characterize the $2$-cuts for which a directed link $(u,v)\in \vec{F}$ is responsible as follows.

\begin{lemma}\label{lem:characterize_responsible}
Let $(G=(V,E),L,c,r, e_r)$ be a rooted WRAP instance and $\vec{F} \subseteq \shadows(L)$ be a non-shortenable directed solution thereof.
A directed link $(u,v)\in \vec{F}$ is responsible for a 2-cut $C\in \Cscr_G$ if and only if $v \in C$ and all vertices in $C$ are $v$-bad.
\end{lemma}
\begin{proof}
Consider a $2$-cut $C$ with $v \in C$ and suppose all vertices in $C$ are $v$-bad.
Then, because $C$ consists only of descendants of $v$ in the arborescence $(V,\vec{F})$, the directed link $(u,v)$ is entering $C$ and no other link on the $r$-$u$ path in $(V,\vec{F})$ is entering $C$. This shows that $(u,v)$ is responsible for the cut $C$.

Let us now consider a cut $C\in \mathcal{C}_G$ for which $(u,v)$ is responsible.
The definition of a link being responsible for a cut (\cref{def:responsible}) and the fact that only one link is responsible for any cut (\cref{lem:responsibility_unique}) imply the following.
For any $t\in C$, the unique $r$-$t$ path in $(V,\vec{F})$ contains the arc $(u,v)$.
Hence, all vertices of $C$ are descendants of $v$, i.e., all vertices of $C$ are $v$-bad, as claimed.
\end{proof}

We are now ready to prove the main result of this section, \cref{lem:ov_characterize_drop}, which we repeat below for convenience.
We recall that \cref{fig:characterizeDrop} provides an illustration for this statement.
\ovCharDrop*
\begin{proof}
Let us first consider the case in which $v$ is connected to a $v$-good vertex $w$ in the link intersection graph $H[K]$ of $K$.
Let $C\in \Cscr_G$ be a 2-cut for which the directed link $(u,v)$ is responsible.
By \cref{lem:characterize_responsible}, we have $v\in C$ and all vertices in $C$ are $v$-bad.
In particular, the $v$-good vertex $w$ is not contained in $C$.
By \cref{lem:paths_in_intersection_graphs_and_cuts}, this implies $\delta_K(C) \neq \emptyset$.

Let us now consider the remaining case in which $v$ is not connected to a $v$-good vertex in the link intersection graph $H[K]$ of $K$.
If $\delta_K(v)=\emptyset$, then $\{v\} \in \Cscr_G$ is a 2-cut for which the directed link $(u,v)$ is responsible and that is not covered by $K$.
Otherwise, let $I\subseteq V$ be the interval from the leftmost vertex to which $v$ is connected in $H[K]$ to the rightmost vertex to which $v$ is connected in $H[K]$.
By \cref{lem:furtherNonshortProps}~\ref{item:descendants_interval} and \cref{def:v-good}, all vertices in $I$ are $v$-bad and thus, by \cref{lem:characterize_responsible}, $I$ is a $2$-cut for which $(u,v)$ is responsible. 
We now prove $\delta_K(I) = \emptyset$.
We suppose there is a link $\{w, \bar w\} \in \delta_K(I)$ and derive a contradiction.
Without loss of generality, assume that $w\in I$ and $w$ is to the right of $v$.
(Note that $v\ne w$ because otherwise $v$ would be connected to the vertex $\bar w \notin I$ in the link intersection graph $H[K]$, contradicting the definition of the interval $I$.)
Let $P$ be a path in the link intersection graph $H[K]$ of $K$ from a link incident to $v$ to a link incident to the rightmost vertex in $I$, which exists by the definition of $I$. 
Let $\link{l}$ be the first link in $P$ for which at least one endpoint of $\link{l}$ is not left of $w$.
Because $\link{l}$ is either incident to $v$ or intersects with its predecessor  in $P$, one endpoint of $\link{l}$ is left of $w$; the other endpoint is either $w$ or a vertex in $I$ to the right of $w$.
Therefore, $\link{l}$ intersects the link $\{w, \bar w\}$, implying that $\bar w$ is connected to $v$ in the link intersection graph $H[K]$.
Because $\bar w\notin I$, this contradicts the definition of the interval $I$.
We conclude $\delta_K(I) = \emptyset$.
\end{proof}

We remark that the above characterization of $\Drop_{\vec{F}}(K)$ is closely linked to the formulation of WCAP as a node-weighted Steiner tree problem introduced in \cite{basavaraju_2014_parameterized} and later used in \cite{byrka_2020_breaching, nutov_2021_approximation, angelidakis2021node}.
To see this, we consider the graph $\tilde H$ that arises from the link intersection graph $H$ by adding every $v\in V$ to the vertex set and adding edges $\{v,\link{l}\}$ for every link $\link{l}\in K$ that has $v$ as an endpoint.
We observe that a link set $K\subseteq L$ is a WRAP solution if and only if $\Drop_{\vec{F}}(K) = \vec{F}$.
By \cref{lem:characterize_drop}, this is the case if and only if every vertex $v\in V$ is connected to a $v$-good vertex in $H[K]$, i.e., $v$ is connected to a $v$-good vertex $w\in V$ in $\tilde H[V\cup K]$. 
Unless $w=r$, the $v$-good vertex $w$ to which $v$ is connected must be connected to a $w$-good vertex itself.
Repeating this reasoning ultimately implies that all vertices in $V$ must be connected to the root $r$.
This allows for interpreting WRAP as a node-weighted Steiner tree problem in $\tilde H$ with terminal set $V$, where all vertices in $V$ have zero weight and the weight of a vertex of $\tilde{H}$ that corresponds to a link equals the cost of the link.
 \section{Thin Components and the Relative Greedy Algorithm}\label{sec:thin-components}

This section provides details on our discussion in \cref{sec:ovThinComps}.
In particular, we give a first simple approximation algorithm for WCAP with an approximation ratio that is better than $2$ and introduce several key technical components needed to prove our main result, \cref{thm:main}.
The algorithm presented in this section is a relative greedy algorithm that, on a high level, is inspired by an algorithm for WTAP from \cite{traub_2021_better}.
It achieves an approximation ratio of $1+\ln 2 + \epsilon$ for any fixed $\epsilon >0$.
In Section~\ref{sec:local-search}, we show how the local search technique from \cite{traub_2022_local} allows us to refine the relative greedy algorithm to a $(1.5+\epsilon)$-approximation algorithm.

\subsection{Thin Components}
The high-level idea of relative greedy algorithms is the following. 
We start with a highly structured but rather weak approximation, which in our setting will be a cheapest non-shortenable directed WRAP solution.
Then we iteratively replace parts of this starting solution by stronger components, chosen out of a class $\mathfrak{K}\subseteq 2^L$, which we will specify later.

Throughout our relative greedy algorithm we maintain a mixed WRAP solution $\vec{F}_i \cupp S_i$, where $\vec{F}_i \subseteq \shadows(L)$ and $S_i \subseteq L$.
Initially, we set $S_0 \coloneqq \emptyset$ and let $\vec{F}_0$ be a cheapest non-shortenable directed WRAP solution.
Then, we iteratively choose a component $K \in \mathfrak{K}$ and add $K$ to our mixed WRAP solution while removing links in $\Drop_{\vec{F}_0}(K)$.
More precisely, in the $i$th iteration of the algorithm we choose a component 
$K \in \mathfrak{K}$ and set $S_i \coloneqq S_{i-1} \cup K$ and $\vec{F}_i \coloneqq \vec{F}_{i-1} \setminus \Drop_{\vec{F}_0}(K)$.

As mentioned in \cref{sec:ovThinComps}, the class $\mathfrak{K} \subseteq 2^L$ of components that we will use in our relative greedy algorithm consists of what we call $\alpha$-thin link sets (for some constant $\alpha$).
The definition of $\alpha$-thin link sets relies on the notion of maximal laminar subfamilies $\Lscr$ of the set $\Cscr_G$ of 2-cuts of $G$.
A set $\Lscr \subseteq \Cscr_G$ is laminar if any pair of cuts $C_1,C_2\in \Lscr$ fulfills either $C_1 \cap C_2 =\emptyset$, $C_1 \subseteq C_2$, or $C_2 \subseteq C_1$.
A set $\Lscr \subseteq \Cscr_G$ is a maximal laminar subfamily of $\Cscr_G$ if there is no cut $C\in \Cscr_G \setminus \Lscr$ such that $\Lscr \cup \{C\}$ is laminar.
Alternatively, a maximal laminar subfamily of $\mathcal{C}_G$ can be described as follows.

\begin{lemma}\label{lem:maximal_laminar}
A set $\Lscr \subseteq \Cscr_G$ is a maximal laminar subfamily of $\Cscr_G$ if and only if $\Lscr$ is laminar and
\begin{enumerate}
\item\label{item:laminar_family_contains_singletons} $\{v\}\in \Lscr$ for all $v\in V\setminus \{r\}$,
\item\label{item:laminar_family_contains_full_set} $V\setminus \{r\} \in \Lscr$, and
\item\label{item:laminar_family_is_binary} for all $C\in \Lscr$ with $|C| > 1$, there are two cuts $C_1,C_2\in \Lscr$ with $C = C_1 \cupp C_2$.
\end{enumerate}
\end{lemma}
\begin{proof}
Let $\Lscr \subseteq \Cscr_G$ be a maximal laminar subfamily of $\Cscr_G$.
Because adding any of the sets $\{v\}$ for $v\in V\setminus \{r\}$ or the set $V\setminus \{r\}$ to $\Lscr$ will maintain laminarity, the maximality of $\Lscr$ implies \ref{item:laminar_family_contains_singletons} and \ref{item:laminar_family_contains_full_set}.
Now consider a cut $C\in \Lscr$ with $|C| > 1$.
Let $C_1, \dots, C_p$ be the children of $C$ in the laminar family $\Lscr$, i.e., the maximal sets in $\Lscr$ with $C_i \subseteq C$.
By \ref{item:laminar_family_contains_singletons}, every vertex in $C$ is contained in one of the sets $C_i$ with $i\in\{1,\dots,p\}$ and thus, because $\Lscr$ is laminar, the sets $C_i$ with $i\in \{1,\dots,p\}$ form a partition of $C$.
The cuts $C_1, \dots, C_p \in \Cscr_G$ are disjoint intervals and hence we may assume that they are ordered such that the vertices in $C_i$ are to the left of the vertices in $C_j$ whenever $i <j$.
Then $C_1 \cup C_2 \subseteq V \setminus \{r\}$ is an interval and hence contained in $\Cscr_G$.
Because adding $C_1 \cup C_2$ to $\Lscr$ maintains laminarity, we must have $C= C_1 \cup C_2$, which implies \ref{item:laminar_family_is_binary}.

Let us now consider a laminar family subfamily of $\Cscr_G$ with \ref{item:laminar_family_contains_singletons}, \ref{item:laminar_family_contains_full_set}, and \ref{item:laminar_family_is_binary}.
Let $\overline{C} \in \Cscr_G \setminus \Lscr$.
Then by \ref{item:laminar_family_contains_full_set}, $\overline{C}$ is a subset of one of the sets in $\Lscr$.
Let $C\in\Lscr$ be minimal with that property.
Because $\overline{C} \subsetneq C$, we have $|C| > 1$ and thus by \ref{item:laminar_family_is_binary}, $C= C_1 \cupp C_2$ for some cuts $C_1, C_2 \in \Lscr$.\
By the minimality of $C$, the cut $\overline{C}$ is neither a subset of $C_1$ nor of $C_2$.
This implies that $\Lscr \cupp \{ \overline{C} \}$ is not laminar.
Hence, $\Lscr$ is a maximal laminar subfamily of $\Cscr_G$.
\end{proof}

The definition below recalls our notion of $\alpha$-thin link sets. (See \cref{fig:ov_definition-width}.)
In our relative greedy algorithm, the component class $\mathfrak{K}$ will be all $\alpha$-thin link sets for some $\alpha = O(\sfrac{1}{\epsilon})$.
\alphaThin*

The above notion of thinness of a component $K$ has a geometric interpretation and is also closely linked to the treewidth of the link intersection graph $H[K]$ of $K$.\
More precisely, the minimum number $\alpha$ for which a link set $K$ is $\alpha$-thin has the property that the treewidth of $H[K]$ is between $\alpha$ and $\frac{3}{2}\alpha$.
We will elaborate on these connections in \cref{sec:triangulations} before discussing the key properties of $\alpha$-thin components in \cref{sec:properties-thinness}.

\subsection{A Geometric Interpretation and Connection to Treewidth}\label{sec:triangulations}

In this section we provide a geometric interpretation of \cref{def:alpha-thin} and explain how it is linked to the notion of treewidth.
The content of this section will not be needed to prove our main result, \cref{thm:main}, but is included to provide intuition for the concept of $\alpha$-thinness and to reveal connections to related concepts.

We number the vertices of the cycle $G=(V,E)$ as $r=v_1, \dots, v_n$ in the order in which they appear on the path $(V, E\setminus\{e_r\})$.
We consider a circle and place the vertices in $V$ on distinct locations along this circle in the same order as they appear in the cycle $G$.
Then we place auxiliary vertices $q_1,\dots, q_n$ along the circle such that $q_i$ is placed between $v_i$ and $v_{i+1}$ for $i\in\{1,\dots, n-1\}$ and $q_n$ is placed between $v_n$ and $r=v_1$.

The straight line segments $\overline{pq}$ between auxiliary vertices $p$ and $q$ correspond to the cuts in $\Cscr_G$ and we will therefore call them \emph{cutlines}.
More precisely, for a cutline $\overline{q_iq_j}$ with $i<j$, the interval $I$ from $v_{i+1}$ to $v_j$, which fulfills $I\in \mathcal{C}_G$, is \emph{the cut corresponding to} the cutline $\overline{q_iq_j}$.
This establishes a one-to-one correspondence between cutlines and cuts in $\Cscr_G$.
See \cref{fig:definition-cutline}.

\begin{figure}[!ht]
\begin{center}
\begin{tikzpicture}[scale=1.6,
ns/.style={thick,draw=black,fill=white,circle,minimum size=4,inner sep=2pt},
qs/.style={thick,draw=blue,fill=blue,circle,minimum size=2,inner sep=1pt},
es/.style={thick},
lks/.style={line width=1.5pt, densely dashed},
dlks/.style={lks, -latex},
ts/.style={every node/.append style={font=\scriptsize}}
]

\def\rad{1.5}
\def\num{7}

\draw[thick] (0,0) circle [radius = \rad];

\begin{scope}
\foreach \i in {2,...,\num} {
  \pgfmathsetmacro\r{90+(\i-1)*360/\num}
  \node[ns] (v\i) at (\r:\rad) {};
  \node () at (\r:\rad + 0.22) {$v_{\i}$};
}
\node[ns] (v1) at (90:\rad) {};
\node () at (90:\rad + 0.22) {$r=v_{1}$};
\end{scope}

\begin{scope}
\foreach \i in {1,...,\num} {
  \pgfmathsetmacro\r{90+(2*\i-1)*360/(2*\num)}
  \node[qs] (q\i) at (\r:\rad) {};
  \node[blue] () at (\r:\rad + 0.2) {$q_{\i}$};
}
\end{scope}

\begin{scope}[ultra thick, opacity=0.6]
\draw[blue!60!green] (q1) -- (q4);
\draw[red!80!black] (q2) -- (q5);
\draw[green!60!black] (q5) -- (q6);
\end{scope}

\begin{scope}[shift={(4.5,0)}]

\draw[thick] (0,0) circle [radius = \rad];

\begin{scope}
\foreach \i in {2,...,\num} {
  \pgfmathsetmacro\r{90+(\i-1)*360/\num}
  \node[ns] (v\i) at (\r:\rad) {};
  \node () at (\r:\rad + 0.22) {$v_{\i}$};
}
\node[ns] (v1) at (90:\rad) {};
\node () at (90:\rad + 0.22) {$r=v_{1}$};
\end{scope}

\def\cw{0.45}
\begin{scope}[every path/.style={draw, fill, fill opacity=0.2,thick}]

\newcommand\Interval[2]{
  \pgfmathsetmacro{\ra}{90+(#1-1)*360/\num}
  \pgfmathsetmacro{\rb}{90+(#2-1)*360/\num}
\draw (\ra:\rad+\cw) arc (\ra:\rb:\rad+\cw)
  arc (\rb:\rb+180:\cw)
  arc (\rb:\ra:\rad-\cw)
  arc (\ra+180:\ra+360:\cw)
;
}
\begin{scope}[blue!60!green]
\Interval{2}{4}
\end{scope}
\def\cw{0.35}
\begin{scope}[red!80!black]
\Interval{3}{5}
\end{scope}
\begin{scope}[green!60!black]
\Interval{6}{6}
\end{scope}

\end{scope}

\end{scope}

\end{tikzpicture}
 \end{center}
\caption{
The left picture shows a cycle $G$ with vertices $v_1,\dots, v_7$ and the auxiliary vertices $q_1,\dots, q_7$ with three cutlines $\overline{q_1q_4}$ (blue), $\overline{q_2q_5}$ (red), and $\overline{q_5q_6}$ (green).
The right picture shows the cuts in $\mathcal{C}_G$ that correspond to these cutlines.
Note that two cutlines cross if and only if their corresponding cuts are crossing.
For example, the cutlines $\overline{q_1q_4}$ and $\overline{q_2q_5}$ cross and also their corresponding cuts shown in blue and red in the right picture are crossing.
Moreover, the cutlines $\overline{q_2q_5}$ and $\overline{q_5q_6}$ do not cross and also their corresponding cuts shown in red and green in the right picture are not crossing.
}
\label{fig:definition-cutline}
\end{figure}
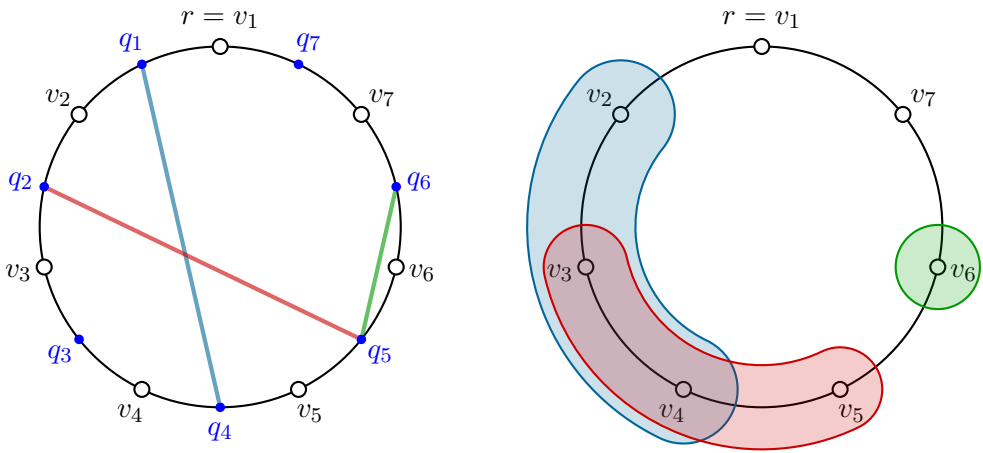

We observe that two cutlines $d_1$ and $d_2$ are crossing (in a geometric sense) if and only if their corresponding cuts $C_1, C_2 \in \mathcal{C}_G$ are crossing, i.e., $C_1 \cap C_2$, $C_1 \setminus C_2$, $C_2\setminus C_1$, and $V\setminus (C_1 \cup C_2)$ are all nonempty.
This implies that the set $D$ of cutlines is planar if and only if the family of cuts corresponding to $D$ is laminar.
We conclude that maximal planar sets of cutlines correspond to maximal laminar subfamilies of $\Cscr_G$.
Finally, we observe that maximal planar sets of cutlines correspond to triangulations of the convex hull $Q$ of the auxiliary vertices $q_1, \dots, q_n$.
See \cref{fig:geometric-definition-width} for an example.

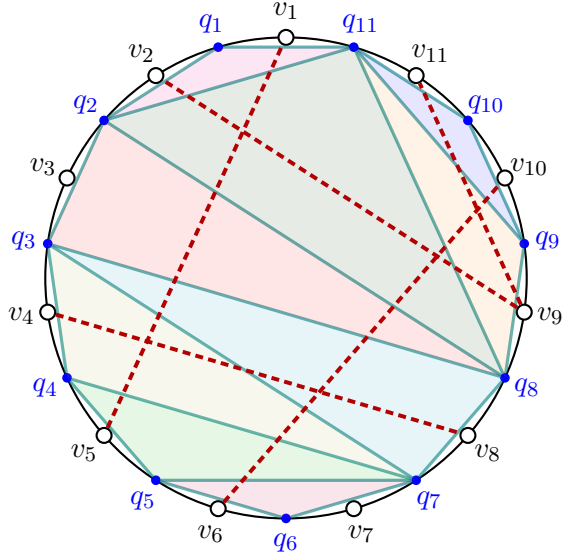
\begin{figure}[!ht]
\begin{center}
\begin{tikzpicture}[scale=1.6,
ns/.style={thick,draw=black,fill=white,circle,minimum size=4,inner sep=2pt},
qs/.style={thick,draw=blue,fill=blue,circle,minimum size=2,inner sep=1pt},
es/.style={thick},
lks/.style={line width=1.5pt, densely dashed},
dlks/.style={lks, -latex},
ts/.style={every node/.append style={font=\scriptsize}}
]

\def\rad{2}
\def\num{11}

\coordinate (c) at (0,0);
\draw[thick] (0,0) circle [radius = \rad];

\begin{scope}
\foreach \i in {1,...,\num} {
  \pgfmathsetmacro\r{90+(\i-1)*360/\num}
  \node[ns] (v\i) at (\r:\rad) {};
  \node () at (\r:\rad + 0.22) {$v_{\i}$};
}
\end{scope}

\begin{scope}\foreach \i in {1,...,\num} {
  \pgfmathsetmacro\r{90+(2*\i-1)*360/(2*\num)}
  \node[qs] (q\i) at (\r:\rad) {};
  \node[blue] () at (\r:\rad + 0.2) {$q_{\i}$};
}
\end{scope}

\begin{scope}[thick, blue]
\foreach \j in {2,...,\num} {
\pgfmathtruncatemacro\i{\j-1}
\draw[very thick, blue!50!green, opacity=0.6] (q\i) to (q\j);
}
\draw[very thick, blue!50!green, opacity=0.6] (q\num) to (q1);
\end{scope}

\begin{scope}[very thick, blue!50!green, opacity=0.6]
\draw (q2) -- (q11);
\draw (q2) -- (q8);
\draw (q11) -- (q8);
\draw (q3) -- (q8);
\draw (q3) -- (q7);
\draw (q4) -- (q7);
\draw (q5) -- (q7);
\draw (q11) -- (q9);
\end{scope}

\begin{scope}[opacity=0.1]
\fill[purple] (q5.center) -- (q6.center) -- (q7.center) -- cycle;
\fill[green!70!black] (q4.center) -- (q5.center) -- (q7.center) -- cycle;
\fill[yellow!70!black] (q3.center) -- (q4.center) -- (q7.center) -- cycle;
\fill[cyan!70!black] (q3.center) -- (q7.center) -- (q8.center) -- cycle;
\fill[red] (q2.center) -- (q3.center) -- (q8.center) -- cycle;
\fill[green!20!black] (q2.center) -- (q11.center) -- (q8.center) -- cycle;
\fill[magenta] (q1.center) -- (q2.center) -- (q11.center) -- cycle;
\fill[orange] (q11.center) -- (q8.center) -- (q9.center) -- cycle;
\fill[blue] (q9.center) -- (q11.center) -- (q10.center) -- cycle;
\end{scope}

\begin{scope}[lks,red!70!black]
\draw (v4) -- (v8);
\draw (v5) -- (v1);
\draw (v6) -- (v10);
\draw (v9) -- (v2);
\draw (v9) -- (v11);
\end{scope}

\end{tikzpicture}
 \end{center}
\caption{
The picture shows a cycle $G$ and the corresponding polygon $Q$ with a triangulation.
This triangulation certifies that the link set $K$ shown in red is $3$-wide because any triangle side is crossed by at most $3$ links.
Moreover, the triangulation satisfies that the link intersection graph $H[K]$ has treewidth at most $3$ because also every triangle is crossed by at most $3$ links.
We remark that the triangulation shown in this figure corresponds to the maximal laminar subfamily of $\Cscr_G$ that is shown in \cref{fig:ov_definition-width} and also the depicted link set $K$ is the same as in \cref{fig:ov_definition-width}.
}
\label{fig:geometric-definition-width}
\end{figure}

This allows us to rephrase \cref{def:alpha-thin} as follows.
A link set $K\subseteq L$ is $\alpha$-thin if and only if there exists a triangulation of the polygon $Q$ such that every side of a triangle is crossed by at most $\alpha$ many links (when embedding each link $\{v,w\}$ as the straight line segment between $v$ and $w$).

We will now use this geometric interpretation of $\alpha$-thinness to establish a connection between $\alpha$-thinness of $K\subseteq L$ and the treewidth of $H[K]$.
\citeauthor{kloks_1996_treewidth}~\cite{kloks_1996_treewidth} gave a characterization of the treewidth of circle graphs, which implies that for $K\subseteq L$, the treewidth of the link intersection graph $H[K]$ is at most $\gamma$  if and only if there exists a triangulation of the polygon $Q$ such that every triangle is crossed by at most $\gamma$ many links from $K$.
Using that every link that crosses a triangle crosses exactly two of its sides and that every link crossing a side of a triangle also crosses the triangle itself, we get the following relations:
\begin{itemize}
\item If $K\subseteq L$ is $\alpha$-thin, then the treewidth of $H[K]$ is at most $\tfrac{3}{2}\alpha$.
\item If, for $K\subseteq L$, the link intersection graph $H[K]$ has treewidth $\alpha$, then $K$ is $\alpha$-thin.
\end{itemize}

\subsection{Key Properties of Thin Components}\label{sec:properties-thinness}

Recall that during our relative greedy algorithm we maintain a mixed WRAP solution $\vec{F_i} \cupp S_i$.
In the $i$th iteration we choose an $\alpha$-thin component $K\subseteq L$ that minimizes
\begin{equation}\label{eq:choice-of-component}
\frac{c(K)}{c(\Drop_{\vec{F_0}}(K)\cap \vec{F}_{i-1})},
\end{equation}
where we again use the convention $\frac{0}{0}\coloneqq 1$ and $\frac{x}{0}\coloneqq \infty$ for $x > 0$.
Then we add $K$ to our current solution and remove $\Drop_{\vec{F}_0}(K)\cap \vec{F}_{i-1}$, i.e., we set $S_{i} \coloneqq S_{i-1} \cup K$ and $\vec{F}_{i} \coloneqq \vec{F}_{i-1} \setminus  \Drop_{\vec{F}_0}(K)$.
For any constant $\alpha$, we can find an $\alpha$-thin link set $K$ minimizing \eqref{eq:choice-of-component} efficiently by a dynamic program.
This is the result of our optimization theorem, which we recall below for convenience.

\ovFindCompForRelGreedy*
We defer the proof of \cref{thm:ov_find-component-for-relative-greedy} to \cref{sec:dp}.
In order to prove that our relative greedy algorithm has an approximation ratio of $2-\delta$ for some $\delta > 0$, we need that at any iteration of our algorithm the following holds.
If our mixed solution $\vec{F}_i \cup S_i$ is not already a $(2-\delta)$-approximation, then there is an $\alpha$-thin component $K$ with $\sfrac{c(K)}{c(\Drop_{\vec{F}_0}(K) \cap \vec{F}_i)}<1$. 
(More precisely, we want this ratio to be bounded away from $1$ to be able to make significant progress.)
The main statement that we use to prove this is our decomposition theorem, which is a key technical contribution of this paper.
It is recalled below for convenience.

\ovDecompositionTheorem*

We defer the proof of \cref{thm:ov_decomposition-theorem} to \cref{sec:decomposition-thm} and first provide full details on how we apply it to analyze our relative greedy algorithm.
As long as the cost $c(\vec{F}_{i-1})$ of the remaining directed links in our mixed solution is much higher than the cost $c(\OPT) = \sum_{K\in \mathcal{K}} c(K)$ of an optimal WRAP solution, \cref{thm:ov_decomposition-theorem} implies that we can find an $4\lceil\sfrac{2}{\epsilon}\rceil$-thin component $K$ with cost significantly less than the cost of the directed links we can remove, thus making our mixed solution overall cheaper.
Indeed, we have
\begin{equation}\label{eq:average_over_components}
\min_{K\in\mathcal{K}} \frac{c(K)}{c(\Drop_{\vec{F}_0}(K) \cap \vec{F}_{i-1})} 
\ \le\ \frac{\sum_{K\in\mathcal{K}}c(K)} {\sum_{K\in\mathcal{K}} c(\Drop_{\vec{F}_0}(K) \cap \vec{F}_{i-1})}\ \le\  \frac{c(\OPT) }{c(\vec{F}_{i-1})-c(\vec{R})},
\end{equation}
where  $c(\vec{R}) \le \sfrac{\epsilon}{2} \cdot c(\vec{F}_0) \le \epsilon \cdot c(\OPT)$.
This allows for quantifying the amount by which our mixed solution $\vec{F}_i \cup S_i$ gets cheaper in the course of the algorithm.
Note that as long as $\vec{F}_{i-1}\ne \emptyset$, we can always find a $1$-thin component $K$ such that $\sfrac{c(K)}{c(\Drop_{\vec{F}_0}(K) \cap \vec{F}_i)}\le 1$, because for a shadow $\vec{\link{l}} \in \vec{F}_{i-1}$ of a link $\link{l}\in L$, we can choose $K=\{\link{l}\}$, which implies $\vec{\link{l}} \in \Drop_{\vec{F}_0}(K)$.
Then our mixed solution does not become cheaper anymore, but it does not harm to run the algorithm until $\vec{F}_i$ is empty and thus $S_i$ is a WRAP solution.

\subsection{The Relative Greedy Algorithm}

In this section we formally state the relative greedy algorithm for WRAP and show that is has an approximation ratio of $1+\ln 2 +\epsilon$.
This will not be needed to prove our main result, \cref{thm:main}, but it yields a simple proof of why \cref{thm:ov_find-component-for-relative-greedy} and \cref{thm:ov_decomposition-theorem} imply the existence of a better-than-2 approximation algorithm for WRAP and hence for WCAP.

\begin{algorithm2e}[H]
\KwIn{A rooted WRAP instance $(G=(V,E),L,c,r, e_r)$.}
\KwOut{A WRAP solution $S\subseteq L$ with $c(S) \le (1+\ln(2)+\epsilon)\cdot c(\OPT)$.}
\vspace*{2mm}
\begin{enumerate}[label=\arabic*.,ref=\arabic*,rightmargin=7mm]\itemsep5pt
\item Set $i\coloneqq 0$ and $S_0 \coloneqq \emptyset$.
\item\label{item:compute_starting_solution} Compute a non-shortenable directed WRAP solution $\vec{F}_0\subseteq \shadows(L)$ with $c(\vec{F}_0) \le 2 \cdot c(\OPT)$.
\item\label{item:choose_comp} While $\vec{F_i} \neq \emptyset$:
\begin{itemize}[itemsep=0.3em]
\item Increase $i$ by $1$.
\item Compute a minimizer $\displaystyle K_i \in \argmin\left\{\frac{c(K)}{c(\Drop_{\vec{F}_0}(K) \cap \vec{F}_{i-1})}\colon K\subseteq L \text{ is $4\lceil\tfrac{2}{\epsilon}\rceil$-thin}\right\}$.

\item Set $S_{i} \coloneqq S_{i-1} \cup K$ and $\vec{F}_{i} \coloneqq \vec{F}_{i-1} \setminus \Drop_{\vec{F}_0}(K)$.
\end{itemize}
\smallskip
\item Return $S\coloneqq S_i$.
\end{enumerate}
\caption{Relative greedy algorithm for WRAP}\label{algo:relative_greedy}
\end{algorithm2e}

Algorithm~\ref{algo:relative_greedy} can be seen as the analogue for WRAP of the relative greedy algorithm from \cite{traub_2021_better}.
We can analyze it in the same way as done in \cite{traub_2021_better} and other relative greedy algorithms, see e.g., \cite{zelikovsky_1996_better,gropl_2001_approximation,cohen_2013_approximation}.
We include the analysis here for completeness.

\begin{theorem}
For every $\epsilon >0$, Algorithm~\ref{algo:relative_greedy} is a $(1+\ln 2 +\epsilon)$-approximation algorithm for WRAP.
\end{theorem}
\begin{proof}
First, we observe that for every iteration $i$ of Algorithm~\ref{algo:relative_greedy}, $\vec{F}_i \cupp S_i$ is a mixed WRAP solution.
This follows from the definition of the set $\Drop_{\vec{F}_0} (K)$ (\cref{def:Drop}) together with the fact that for every 2-cut $C\in \Cscr_G$ there is at least one link $\vec{\link{l}} \in \vec{F}_0$ that is responsible for $C$ (by \cref{lem:responsibility_unique}).
By the decomposition theorem (\cref{thm:ov_decomposition-theorem}) and \eqref{eq:average_over_components}, we get that the component $K_i$ chosen in the $i$th iteration of Algorithm~\ref{algo:relative_greedy} satisfies
\begin{equation}\label{eq:ratio_bound_from_decomposition_theorem}
\frac{c(K_i)}{c(\Drop_{\vec{F}_0}(K) \cap \vec{F}_{i-1})} \ \le\ \frac{c(\OPT)}{c(\vec{F}_{i-1})-c(\vec{R})}
\end{equation}
for some set $\vec{R} \subseteq \shadows(L)$ with $c(\vec{R}) \le \epsilon \cdot c(\OPT)$.
Consider an arbitrary link $\vec{\link{l}} \in \vec{F}_{i-1}$.
Then $\vec{\link{l}}$ is a shadow of some link $\link{l} \in L$ and we have $\vec{\link{l}} \in \Drop_{\vec{F}_0}(\{\link{l}\})$ by \cref{def:Drop}.
Because the set $\{\link{l}\}$ is $4\left\lceil\sfrac{2}{\epsilon}\right\rceil$-thin and thus the algorithm could have chosen $K_i= \{\link{l}\}$, this shows 
$ \sfrac{c(K_i)}{c(\Drop_{\vec{F}_0}(K) \cap \vec{F}_{-1})} \ \le\ 1$ and also implies $K\neq \emptyset$ as well as $\Drop_{\vec{F}_0}(K) \cap \vec{F}_{-1}) \neq \emptyset$.
This shows that \cref{algo:relative_greedy} terminates after at most $|F_0| = |V|-1$ iterations.
Moreover, combining $\sfrac{c(K_i)}{c(\Drop_{\vec{F}_0}(K) \cap \vec{F}_{i-1})} \ \le\ 1$ with \eqref{eq:ratio_bound_from_decomposition_theorem}, we obtain
\begin{equation*}
c(K_i) \ \le\ \min\left\{\frac{c(\OPT)}{c(\vec{F}_{i})- c(\vec{R})}, 1\right\} \cdot c(\vec{F}_{i} \setminus \vec{F}_{i+1})
\ \le\ \int_{c(\vec{F}_{i+1})}^{c(\vec{F}_{i})} \min\left\{\frac{c(\OPT)}{x - c(\vec{R})},1\right\}\,dx\enspace,
\end{equation*}
where we used $c(\vec{F}_{i+1}) \le c(\vec{F}_i)$ and that $\min\{\sfrac{c(\OPT)}{(x -c (\vec{R}))},1\}$ is nonincreasing in $x$.

Let $m$ denote the number of iterations of the while loop.
Then the WRAP solution we return is $S= S_m$ and, moreover, we have $\vec{F}_m=\emptyset$ and 
\begin{align*}
c(S)\ =\ \sum_{i=1}^m c(K_i) \ 
\le&\ \sum_{i=1}^m  \int_{c(\vec{F}_i)}^{c(\vec{F}_{i-1})} \min\left\{\frac{c(\OPT)}{x- c(\vec{R})},1\right\}\,dx\\[2mm]
=&\ \int_0^{c(\OPT)+c(\vec{R})} 1 \,dx + \int_{c(\OPT)+c(\vec{R})}^{c(\vec{F}_0)}\frac{c(\OPT)}{x-c(\vec{R})} \,dx \\[2mm]
=&\ c(\OPT)+c(\vec{R}) + \ln\left(\frac{c(\vec{F}_0)-c(\vec{R})}{c(\OPT)}\right) \cdot c(\OPT) \\[2mm]
\le&\ (1 +\epsilon)\cdot c(\OPT) + \ln 2 \cdot c(\OPT)\enspace,
\end{align*}
where the last inequality follows from $c(\vec{R}) \le \epsilon \cdot c(\OPT)$ and $c(\vec{F}_0) \le 2 \cdot c(\OPT)$.
\end{proof}
 \section{Proving the Decomposition Theorem}\label{sec:decomposition-thm}

In this section we prove the decomposition theorem (\cref{thm:ov_decomposition-theorem}).
To this end, we fix a rooted WRAP instance $(G,L,c,r,e_r)$.
Moreover, we fix a WRAP solution $S$ and a non-shortenable directed solution $\vec{F}$.
The notion of $v$-good vertices for $v\in V\setminus\{r\}$ is defined with respect to this fixed solution $\vec{F}$ throughout this section.
Recall that our goal is to show that, for any $\epsilon > 0$, there exists a partition $\mathcal{K}$ of the solution $S$ into $4\lceil \sfrac{1}{\epsilon} \rceil$-thin components such that for all but an $\epsilon$-fraction of the links $\vec{\link{l}}\in\vec{F}$ (in terms of cost), we have $\vec{\link{l}}\in \Drop_{\vec{F}}(K)$ for some component $K\in \mathcal{K}$.
Our proof of this statement proceeds in several steps.

First, we construct a partition $\Xscr$ of the solution $S$ into so-called festoons.
Festoons are very structured link sets which have the nice property that adding a festoon to a link set can increase its thinness by at most $4$.
These festoons will be the building blocks for the components $K\in \mathcal{K}$ that we will construct.
More precisely, every component $K\in \mathcal{K}$ will be the union of some festoons from $\Xscr$.
We formally introduce the notion of festoons and some of their basic properties in \cref{sec:festoons}.
Then, in \cref{sec:family_of_festoons}, we describe how we construct the partition $\Xscr$ of the solution $S$ into festoons.

Next, we choose for every link $(u,v)\in\vec{F}$, a minimal subset $\Xscr_{v}$ of the festoons in $\Xscr$ that allows for dropping the link $(u,v)$, i.e., such that
\[
\textstyle
 (u,v) \in \Drop_{\vec{F}}\Big(\bigcup_{X \in \Xscr_v} X \Big).
\]
(Note that every vertex $v\in V$ has at most one incoming arc $(u,v)$ in the arborescence $(V,\vec{F})$.)
We remark that it will be crucial to choose the sets $\Xscr_{v}$ of festoons carefully.
This is described in \cref{sec:dependency_graph}, where we also prove that the sets $\Xscr_v$ are highly structured.
We will show that for all but an $\epsilon$-fraction of the links $\vec{\link{l}}\in\vec{F}$ (in terms of cost), we can ensure that all festoons in $\Xscr_v$ are part of the same component in $\mathcal{K}$, which will complete the proof of the decomposition theorem.

To this end, we will model the sets $\Xscr_{v}$ of festoons, which we ideally would like to be part of the same component in $\mathcal{K}$, by a graph with vertex set being the set $\Xscr$ of festoons.
The edges of this graph, which we call the \emph{dependency graph}, model the sets $\Xscr_{v}$.
We introduce this graph in \cref{sec:dependency_graph} and show that it is a branching, i.e., every connected component is an arborescence.
We can define this graph for any possible set $\vec{R} \subseteq \vec{F}$, where the edges of the dependency graph model only the sets $\Xscr_{v}$ for links $(u,v)\in \vec{F}\setminus\vec{R}$.
For the set $\vec{R}$ that we will choose, the vertex sets of the connected components of the dependency graph will correspond to the partition $\mathcal{K}$.
More precisely, two festoons will be part of the same component $K\in \mathcal{K}$ if and only if they are part of the same connected component of the dependency graph.

At this point it remains to choose the set $\vec{R}\subseteq \vec{F}$.
To this end, we relate properties of the dependency graph to the thinness of the corresponding components (\cref{sec:thinness_dependency_graph}).
Using these properties, we can then in \cref{sec:choosing_R} complete the proof of the decomposition theorem by showing that it is possible to choose a set $\vec{R} \subseteq \vec{F}$ of sufficiently small cost such that, in the partition $\Kscr$ corresponding to the connected components of the dependency graph, all sets $K\in \mathcal{K}$ are $4\lceil \sfrac{1}{\epsilon} \rceil$-thin.

\subsection{Festoons}\label{sec:festoons}

In this section we introduce festoons, which will be the building blocks for the components into which we will decompose the WRAP solution $S$.
For a link $\link{l}\in L$, we denote by $\leftp(\link{l})$ the left endpoint of $\link{l}$ and by $\rightp(\link{l})$ the right endpoint of~$\link{l}$.
\cref{fig:festoon-example} shows an example of a festoon, which is defined as follows.

\begin{definition}[festoon, festoon interval]\label{def:festoon}
A link set $X\subseteq L$ is a \emph{festoon} if the link intersection graph $H[X]$ of $X$ is a path and we can number the links in $X$ as $\link{l}_1,\dots, \link{l}_p$ such that
\begin{itemize}
\item the links $\link{l}_1,\dots, \link{l}_p$ are visited in this order by the path $H[X]$,
\item for $i,j\in\{1,\dots,p\}$ with $i<j$, the vertex $\leftp(\link{l}_i)$ is to the left of the vertex $\leftp(\link{l}_j)$, and
\item for $i,j\in\{1,\dots,p\}$ with $i<j$, the vertex $\rightp(\link{l}_i)$ is to the right of the vertex $\rightp(\link{l}_j)$.
\end{itemize}
Then the interval from $\leftp(\link{l}_1)$ to $\rightp(\link{l}_p)$ is called the \emph{festoon interval} of the festoon $X$ and is denoted by $I_X$, and the ordering $\link{l}_1,\ldots, \link{l}_p$ of the links in $X$ is called the \emph{festoon order}.
\end{definition}
Note that the festoon order is clearly unique due to the second (or also third) condition of \cref{def:festoon}, and that the left and right endpoints of $I_X$ are $\leftp(\link{l}_1)$ and $\leftp(\link{l}_p)$, respectively.
Also, the definition of festoons is symmetric in terms of exchanging ``left'' and ``right''.
Moreover, one can observe that the second condition of \cref{def:festoon} almost implies the third one (and vice versa);
more precisely, the only additional property implied by the third point is $\rightp(\link{l}_{p-1}) \neq \rightp(\link{l}_{p})$.

\begin{figure}[!ht]
\begin{center}
\begin{tikzpicture}[scale=1,
ns/.style={thick,draw=black,fill=white,circle,minimum size=6,inner sep=2pt},
es/.style={thick},
lks/.style={line width=1.5pt, blue, densely dashed},
dlks/.style={lks, -latex},
ts/.style={every node/.append style={font=\scriptsize}}
]

\def\num{14}
\def\hd{1.1}

\begin{scope}[every node/.style={ns}]
\foreach \i in {1,...,\num} {
  \node (a\i) at (\i*\hd,0) {};
}
\end{scope}

\begin{scope}
\node at ($(a1)+(-0.1,0)$)[above=3pt] {$r$};
\pgfmathtruncatemacro\n{\num-1}
\foreach \i in {2,...,\n} {
\node at (a\i)[above=3pt] {};
}
\node at ($(a\num)+(0.1,0)$)[above=3pt] {};
\end{scope}

\begin{scope}[es]
\foreach \i in {2,...,\num} {
\pgfmathtruncatemacro\j{\i-1}
\draw (a\j) -- (a\i);
}
\draw (a1) to[out=50,in=-230,looseness=0.3] node[above,pos=0.8] {$e_r$} (a\num);
\end{scope}

\begin{scope}[lks, darkred]
\draw (a1)  to[bend right=30] node[below] {$\link{l}_1$} (a5);
\draw (a3)  to[bend right=50] node[below] {$\link{l}_2$} (a6);
\draw (a6)  to[bend right=50] node[below] {$\link{l}_3$} (a9);
\draw (a8)  to[bend right=40] node[below] {$\link{l}_4$} (a11);
\draw (a10) to[bend right=60] node[below] {$\link{l}_5$} (a12);
\end{scope}

\end{tikzpicture}
 \end{center}
\caption{
The link set shown in red is an example of a festoon with a numbering of the links according to the festoon order.
}
\label{fig:festoon-example}
\end{figure}
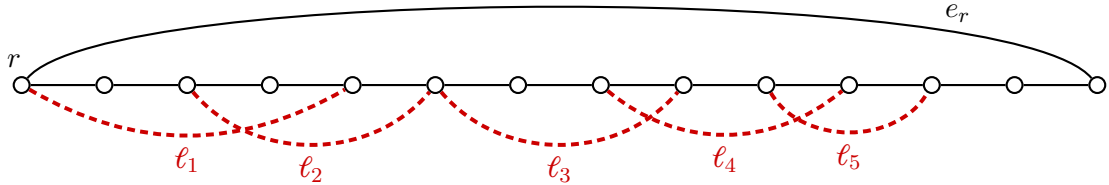

We now prove some basic properties of festoons.
\begin{lemma}\label{lem:endpoint_order_festoon}
If $X=\{\link{l}_1,\dots, \link{l}_p\}$ is a festoon with links numbered according to the festoon order, then, for $i < j -1$, the vertex $\rightp(\link{l}_i)$ is to the left of the vertex $\leftp(\link{l}_j)$.
\end{lemma}
\begin{proof}
We first recall that, for $i < j-1$, the vertex $\leftp(\link{l}_i)$ is to the left of  $\leftp(\link{l}_j)$ and $\rightp(\link{l}_i)$ is to the left of $\rightp(\link{l}_j)$.
Because by \cref{def:festoon} the links $\link{l}_i$  and $\link{l}_j$ are not adjacent in $H[X]$, i.e., they are not intersecting, this implies that $\rightp(\link{l}_i)$ is to the left of $\leftp(\link{l}_j)$.
\end{proof}

The following property of festoons will be crucial for showing thinness of the components we construct.
It implies that not only every festoon is $4$-thin, but we even have the much stronger property that adding a festoon to any link set cannot increase the thinness by more than~$4$.
\begin{lemma}\label{lem:festoon_light}
Let $X$ be a festoon and $C\in \mathcal{C}_G$ be a $2$-cut.
Then $|\delta_X(C)|\leq 4$.
\end{lemma}
\begin{proof}
Let $X=\{\link{l}_1,\dots, \link{l}_p\}$ be the festoon order numbering of the links in $X$.
Because $C\in\Cscr_G$, the cut $C$ is an interval, say from $u$ to $v$.
Every link in $\delta_X(C)$ either has its left endpoint to the left of $u$, or it has its right endpoint to the right of $v$.
We show that there are at most two links in $\delta_X(C)$ of each of these two types.
If there exists at least one link in $\delta_X(C)$ that has its left endpoint to the left of $u$, let $i$ be minimal such that $\link{l}_i \in \delta_X(C)$ has this property.
Because $\link{l}_i \in \delta_X(C)$, we have $\rightp(\link{l}_i)\in C$.
In particular, $\rightp(\link{l}_i)$ is not left of $u$ and hence, by \cref{lem:endpoint_order_festoon}, for every $j > i+1$ the link $\link{l}_j$ does not have its left endpoint to the left of $u$.
A symmetric argument shows that there are also at most two link in $\delta_X(C)$ that have their right endpoint to the right of $v$. 
\end{proof}

\subsection{Partitioning the solution $S$ into festoons}\label{sec:family_of_festoons}

Next, we describe how we partition the WRAP solution $S$ into festoons.
A crucial property of the resulting partition $\Xscr$ will be that the family $\{I_X : X\in \Xscr\}$ of the corresponding festoon intervals is laminar (\cref{lem:festoon-intervals_laminar}).
To obtain the partition $\Xscr$ we proceed as follows.
Starting with $\overline{S}\coloneqq S$ and $\Xscr = \emptyset$, we iteratively compute a festoon $X\subseteq \overline{S}$ such that the festoon interval $I_X$ is maximal, add $X$ to $\Xscr$, and replace $\overline{S}$ by $\overline{S}\setminus X$.
When $\overline{S}$ is empty, the set $\Xscr$ is a partition of $S$ into festoons.
We now show that we can efficiently compute a festoon with maximal festoon interval.
In fact, we prove the slightly stronger statement that we can choose $X$ such that $|I_X|$ is maximum, but we will not need this stronger property.

\begin{lemma}\label{lem:compMaxFestoon}
Given a link set $\overline{S}\subseteq L$, we can efficiently compute a festoon $X\subseteq \overline{S}$ such that $|I_X|$ is maximum.
\end{lemma}
\begin{proof}
For a fixed pair $\link{l},\link{f} \in \overline{S}$, we can efficiently check if there is a festoon for which $\link{l}$ is the leftmost link, i.e., the first link in the festoon order, and $\link{f}$ is the rightmost link, i.e., the last link in the festoon order, as follows.
We consider a directed auxiliary graph with vertex set $\overline{S}$ where we have a directed arc $(\link{l}_1,\link{l}_2)$ if
\begin{enumerate}
\item $\link{l}_1$ and $\link{l}_2$ are intersecting,
\item $\leftp(\link{l}_1)$ is to the left of $\leftp(\link{l}_2)$, and
\item $\rightp(\link{l}_1)$ is to the right of $\rightp(\link{l}_2)$.
\end{enumerate}
Then every festoon $X\subseteq \overline{S}$ corresponds to the directed path in the auxiliary graph that visits the links in $X$ in the festoon order.
Thus, if the auxiliary graph contains no $\link{l}$-$\link{f}$ path, then there is no festoon in $\overline{S}$ that contains both $\link{l}$ and $\link{f}$.
Otherwise, we can compute a shortest $\link{l}$-$\link{f}$ path in the auxiliary graph with vertices $\link{l}_1 = \link{l}, \dots, \link{l}_p = \link{f}$ visited in this order.
The vertex set $\{\link{l}_1,\dots,\link{l}_p\}$ of this path is a festoon because $P$ being a shortest path implies that $\link{l}_i$ and $\link{l}_j$ are not intersecting if $|i-j|>1$.
By enumerating over all possible choices of $\link{l}$ and $\link{f}$ we can find a desired festoon $X$ maximizing $|I_X|$.
\end{proof}
Alternatively, one can also find a festoon $X\subseteq \overline{S}$ with maximum $|I_X|$ through a dynamic program over the directed graph constructed in the proof of \cref{lem:compMaxFestoon}.
More precisely, this directed graph is acyclic by construction, and a desired festoon $X$ can be found using dynamic programming techniques analogous to how a longest path can be computed in an acyclic graph.\footnote{This dynamic programming techniques to obtain $X$ with maximum festoon interval leads to a faster procedure than the one used in the proof of \cref{lem:compMaxFestoon}.
However, in the interest of simplicity we refrain from optimizing running times and strive to present conceptually simple approaches.}

We now show that the family of festoon intervals corresponding to $\Xscr$ is indeed laminar.
\begin{lemma}\label{lem:festoon-intervals_laminar}
The family $\{ I_X: X\in\Xscr\}$ of festoon intervals is laminar.
\end{lemma}
\begin{proof}
Suppose for the sake of deriving a contradiction that there exist festoons $X,Y\in \Xscr$ such that the festoon intervals $I_X$ and $I_Y$ are crossing, i.e., we neither have $I_X \subseteq I_Y$, $I_Y \subseteq I_X$, nor $I_X \cap I_Y =\emptyset$.
Without loss of generality we assume that the festoon $X$ was constructed before $Y$ in the construction of the festoon family $\Xscr$.
In particular, when the festoon $X$ was constructed, we had $X\cup Y \subseteq \overline{S}$.
Let $\link{l}_1, \dots, \link{l}_p$ be the links in the festoon $X$, numbered according to the festoon order.

Because $I_X$ and $I_Y$ are crossing, either the leftmost or the rightmost vertex of $I_Y$ is contained in $I_X$, but not both.
Therefore, \cref{lem:paths_in_intersection_graphs_and_cuts} implies that there exists a link $\{u,v\} \in \delta_Y(I_X)$, say with $u\in I_X$.
Then the vertex $v$ is lies outside of $I_X$ and, because the definition of a festoon is symmetric, we assume without loss of generality that $v$ is to the right of $I_X$.
Thus, the rightmost vertex in $I_X$ is contained in the interval $I_{uv}$ from $u$ to $v$ and the leftmost vertex in $I_X$ is not.
Hence, by \cref{lem:paths_in_intersection_graphs_and_cuts}, there exists a link in $\delta_X(I_{uv})$ and this link intersects the link $\{u,v\}\in Y$.
Let $j$ be the smallest index in $\{1,\dots, p\}$ such that $\link{l}_j \in X$ intersects $\{u,v\}$.
See \cref{fig:proof_laminarity}.

\begin{figure}[!ht]
\begin{center}
\begin{tikzpicture}[scale=1,
ns/.style={thick,draw=black,fill=white,circle,minimum size=6,inner sep=2pt},
es/.style={thick},
lks/.style={line width=1.5pt, blue, densely dashed},
dlks/.style={lks, -latex},
ts/.style={every node/.append style={font=\scriptsize}}
]

\def\num{15}
\def\hd{1}

\begin{scope}[every node/.style={ns}]
\foreach \i in {1,...,\num} {
  \node (a\i) at (\i*\hd,0) {};
}
\end{scope}

\begin{scope}
\node at ($(a1)+(-0.1,0)$)[above=3pt] {$r$};
\pgfmathtruncatemacro\n{\num-1}
\foreach \i in {2,...,\n} {
\node at (a\i)[above=3pt] {};
}
\node at ($(a\num)+(0.1,0)$)[above=3pt] {};
\end{scope}

\begin{scope}[es]
\foreach \i in {2,...,\num} {
\pgfmathtruncatemacro\j{\i-1}
\draw (a\j) -- (a\i);
}
\draw (a1) to[out=50,in=-230,looseness=0.3] node[above,pos=0.8] {$e_r$} (a\num);
\end{scope}

\begin{scope}[lks, darkred]
\draw (a2)  to[bend right=50] node[below] {$\link{l}_1$} (a4);
\draw (a4)  to[bend right=60] node[below] {$\link{l}_2$} (a9);
\draw (a6)  to[bend right=30]  (a11);
\node[darkred] () at (9.5,-1) {$\link{l}_3$};
\draw (a10) to[bend right=60] node[below] {$\link{l}_4$} (a12);
\end{scope}
\begin{scope}[lks, blue]
\draw (a5)  to[bend right=50] node[below] {} (a7);
\draw (a7)  to[bend right=45] node[below] {$\{u,v\}\in \delta_Y(I_X)$} (a14);
\draw (a13)  to[bend right=40] node[below] {} (a15);
\end{scope}

\node[darkred] (X) at (1.5,-1.4) {festoon $X$};
\node[blue] (Y) at (1.5,-1.9) {festoon $Y$};

\node[above=2pt] () at (a7) {$u$};
\node[above=2pt] () at (a14) {$v$};
\end{tikzpicture}
 \end{center}
\caption{An example illustrating the proof of \cref{lem:festoon-intervals_laminar}. 
The festoon $X$ is shown in red and the festoon $Y$ is shown in blue.
In this example $j=2$ and the links set $\{\link{l}_1,\link{l}_2, \link\{u,v\}\}$ is a festoon with a festoon interval that is strictly larger than $I_X$.
\label{fig:proof_laminarity}}
\end{figure}
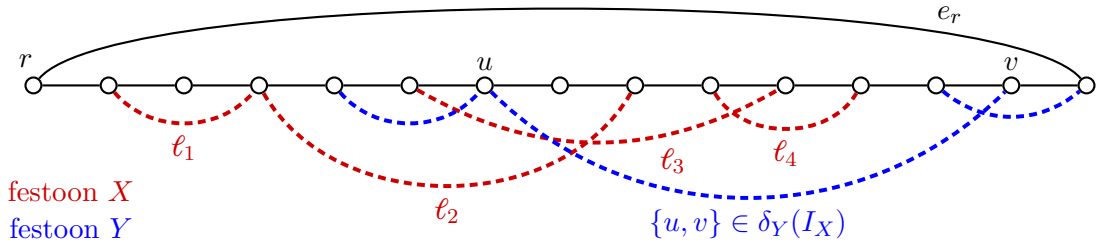

Because $v$ is to the right of the festoon interval $I_X$ and the links $\{u,v\}$ and $\link{l}_j\in X$ are intersecting, the vertex $\leftp(\link{l}_j)$ must be to the left of $u = \leftp(\{u,v\})$.
Moreover, by the choice of the index $j$, the links $\link{l}_i$ and $\{u,v\}$ are intersecting for $i=j$, but not for $i <j$ and hence we conclude that $\{\link{l}_1,\dots,\link{l}_j, \{u,v\}\}$ is a festoon.
However, the festoon interval of this festoon is strictly larger than $I_X$, contradicting the choice of $X$ in the construction of $\Xscr$.
\end{proof}

Recall that in the partition of the WRAP solution $S$ into $4\lceil \sfrac{1}{\epsilon}\rceil$-thin components that we construct, every component will be the union of some festoons from $\Xscr$.
In order to use \cref{lem:characterize_drop} to analyze which links are contained in the sets $\Drop_{\vec{F}}(K)$ for the components $K$ that we construct, we need to understand which links are connected to each other in the link intersection graph $H[K]$ of these components.
For this reason it will be useful to understand when the link intersection graph $H[X\cup Y]$ of the union of two festoons $X,Y \in \Xscr$ is connected.
This is the case if the festoons $X$ and $Y$ are \emph{tangled}. 

\begin{definition}[tangled festoons]\label{def:tangled}
Two festoons $X,Y \in \Xscr$ are \emph{tangled} if there is a link in $X$ and a link in $Y$ that intersect.
\end{definition}

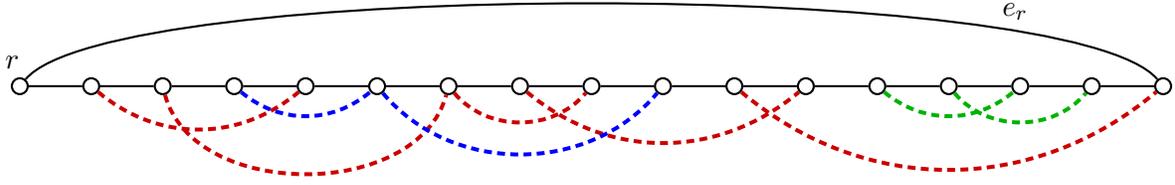
\begin{figure}[!ht]
\begin{center}
\begin{tikzpicture}[scale=1,
ns/.style={thick,draw=black,fill=white,circle,minimum size=6,inner sep=2pt},
es/.style={thick},
lks/.style={line width=1.5pt, blue, densely dashed},
dlks/.style={lks, -latex},
ts/.style={every node/.append style={font=\scriptsize}}
]

\def\num{17}
\def\hd{0.95}

\begin{scope}[every node/.style={ns}]
\foreach \i in {1,...,\num} {
  \node (a\i) at (\i*\hd,0) {};
}
\end{scope}

\begin{scope}
\node at ($(a1)+(-0.1,0)$)[above=3pt] {$r$};
\pgfmathtruncatemacro\n{\num-1}
\foreach \i in {2,...,\n} {
\node at (a\i)[above=3pt] {};
}
\node at ($(a\num)+(0.1,0)$)[above=3pt] {};
\end{scope}

\begin{scope}[es]
\foreach \i in {2,...,\num} {
\pgfmathtruncatemacro\j{\i-1}
\draw (a\j) -- (a\i);
}
\draw (a1) to[out=50,in=-230,looseness=0.3] node[above,pos=0.8] {$e_r$} (a\num);
\end{scope}

\begin{scope}[lks, darkred]
\draw (a2)  to[bend right=40]  (a5);
\draw (a3)  to[bend right=75]  (a7);
\draw (a7)  to[bend right=50]  (a9);
\draw (a8)  to[bend right=40]  (a12);
\draw (a11) to[bend right=40]  (a17);
\end{scope}
\begin{scope}[lks, blue]
\draw (a4)  to[bend right=40]  (a6);
\draw (a6)  to[bend right=50]  (a10);
\end{scope}
\begin{scope}[lks, green!70!black]
\draw (a13)  to[bend right=40]  (a15);
\draw (a14)  to[bend right=50]  (a16);
\end{scope}
\end{tikzpicture}
 \end{center}
\caption{
The picture shows three festoons colored in red, blue, and green, respectively.
The red and the blue festoons are tangled, while the red and green one are not.
Also the blue and the green festoons are not tangled.
}
\label{fig:tangled-example}
\end{figure}

Note that festoons $X,Y\in \Xscr$ with disjoint festoon intervals are never tangled.
Otherwise, by \cref{lem:festoon-intervals_laminar}, we either have $I_X \subseteq I_Y$ or $I_Y \subseteq I_X$.
We now characterize when the festoons $X$ and $Y$ are tangled in this case.
We use the following notation.

\begin{definition}[partial order $\prec$]
For two festoons $X,Y\in \Xscr$, we write $X \prec Y$ if and only if $I_X \subsetneq I_Y$ and we write $X \preceq Y$ if and only if $I_X \subseteq I_Y$.
\end{definition}

\begin{lemma}\label{lem:characterize_tangled}
Let $X,Y$ be festoons with $X \preceq Y$. Then $X$ and $Y$ are tangled if and only if there exists a link in $Y$ that has an endpoint in the festoon interval $I_X$. 
\end{lemma}
\begin{proof}
Clearly, for $X$ and $Y$ to be tangled, $Y$ needs to have an endpoint in the festoon interval $I_X$.
Now suppose there exists a link in $Y$ that has an endpoint in $I_X$.
If $I_X = I_Y$, then both $X$ and $Y$ have a link incident to the leftmost (and also rightmost) vertex of the interval $I_X=I_Y$ and these links are intersecting, implying that $X$ and $Y$ are tangled.
Thus, let us now assume $X \prec Y$.
Then $Y$ contains a link that has an endpoint not contained in $I_X$.
Because $Y$ is a festoon and hence the link intersection graph $H[Y]$ is a path, \cref{lem:paths_in_intersection_graphs_and_cuts} implies that there exists a link $\{u,v\} \in \delta_Y(I_X)$, where without loss of generality $u$ is to the left of $v$.
Because $\{u,v\}\in \delta_Y(I_X)$, either the leftmost vertex of $I_X$ is contained in the interval $I_{uv}$ from $u$ to $v$ and the rightmost vertex of $I_X$ is not, or the other way around.
Therefore, because the leftmost and rightmost vertex of $I_X$ are incident to links in $X$ and because $H[X]$ is a path, \cref{lem:paths_in_intersection_graphs_and_cuts} implies that $X$ contains a link $\link{l}\in\delta_X(I_{uv})$.
Then, the links $\{u,v\}\in X$ and $\link{l}\in Y$ are intersecting.
Hence, $X$ and $Y$ are tangled.
\end{proof}

As a corollary of \cref{lem:characterize_tangled}, we obtain the following useful property.

\begin{corollary}\label{cor:sandwiched_festoons}
Let $X,Y, Z$ be festoons with $X \preceq Y \preceq Z$. 
If $X$ and $Z$ are tangled, then $Y$ and $Z$ are tangled.
\end{corollary}
\begin{proof}
Suppose $X$ and $Z$ are tangled.
Then by \cref{lem:characterize_tangled}, there exists a link $\link{l} \in Z$ that has an endpoint in $I_X \subseteq I_Y$.
Therefore, \cref{lem:characterize_tangled} implies that also $Y$ and $Z$ are tangled. 
\end{proof}

\subsection{Constructing the Dependency Graph}\label{sec:dependency_graph}

Next, we analyze the structure of (inclusion-wise) minimal sets $\Xscr_v \subseteq \Xscr$ that allow for dropping the incoming arc of $v$ in the arborescence $(V,\vec{F})$.
By \cref{lem:characterize_drop}, these are the minimal sets $\Xscr_v \subseteq \Xscr$ that connect $v$ to a $v$-good vertex.

\begin{definition}[connecting set of festoons]
We say that a set $\Xscr_v \subseteq \Xscr$ of festoons \emph{connects a vertex $v$ to a vertex $w$} if $v$ is connected to $w$ in the link intersection graph of $\bigcup_{X\in \Xscr_v} X$.
\end{definition}

We first show that whether a festoon $X$ contains a link incident to a $v$-good vertex depends solely on the festoon interval $I_X$.

\begin{lemma}\label{lem:v-good_in_interval_suffices}
Let $v\in V\setminus \{r\}$ and let $X\subseteq L$ be a festoon.
Then there exists a $v$-good vertex that is incident to a link from $X$ if and only if the festoon interval $I_X$ contains a $v$-good vertex.
\end{lemma}
\begin{proof}
Because the set of $v$-bad vertices is an interval (\cref{lem:furtherNonshortProps}~\ref{item:descendants_interval}), we have that $I_X$ contains a $v$-good vertex if and only if one of the endpoints of $I_X$ is $v$-good.
Moreover, by the definition of festoon intervals (\cref{def:festoon}), both endpoints of $I_X$ are incident to a link in $X$, which implies the result.
\end{proof}

We are now  ready to prove our main structural result on minimal sets of festoons connecting a vertex $v$ to a $v$-good vertex.
An example is given in \cref{fig:minimal_connecting_set_of_festoons}.

\begin{figure}[!ht]
\begin{center}
\begin{tikzpicture}[yscale=1.9,
ns/.style={thick,draw=black,fill=white,circle,minimum size=4,inner sep=1.3pt},
es/.style={thick},
lks/.style={line width=1.5pt, densely dashed},
dlks/.style={lks, -latex},
ts/.style={every node/.append style={font=\scriptsize}}
]

\def\num{32}
\def\hd{0.5}

\begin{scope}[every node/.style={ns}]
\foreach \i in {1,...,\num} {
 \ifthenelse{\i < 2 \OR \i > 29}
  { 
      \node[fill=black] (a\i) at (\i*\hd,0) {};
  }{
      \node (a\i) at (\i*\hd,0) {};
  }

}
\end{scope}

\begin{scope}
\node at ($(a1)+(-0.1,0)$)[above=3pt] {$r$};
\node at ($(a15)$)[above=1pt] {$v$};
\end{scope}

\begin{scope}[es]
\foreach \i in {2,...,\num} {
\pgfmathtruncatemacro\j{\i-1}
\draw (a\j) -- (a\i);
}
\draw (a1) to[out=-80,in=-100,looseness=0.45] node[below right,pos=0.9] {$e_r$} (a\num);
\end{scope}

\begin{scope}[lks]
\begin{scope}[darkred]
\draw (a12)  to[bend right=25]  (a15);
\draw (a14)  to[bend right=40]  (a16);
\draw (a16)  to[bend right=40]  (a18);
\draw (a18) to[bend right=20] (a20);
\end{scope}
\begin{scope}[blue]
\draw (a8) to[bend right=40] (a11);
\draw (a10) to[bend right] (a13);
\draw (a13) to[bend right] (a19);
\draw (a17) to[bend right=40] (a20);
\end{scope}
\begin{scope}[orange!90!black]
\draw (a6) to[bend right] (a10);
\draw (a9) to[bend right] (a22);
\draw (a21) to [bend right] (a23);
\draw (a23) to [bend right] (a25);
\end{scope}
\begin{scope}[green!50!black]
\draw (a3) to[bend right] (a7);
\draw (a5) to [bend right] (a24);
\draw (a24) to[bend right] (a27);
\draw (a26) to [bend right] (a29);
\end{scope}
\begin{scope}[magenta]
\draw (a2) to[bend right] (a5);
\draw (a4) to[bend right] (a28);
\draw (a28) to[bend right] (a30);
\draw (a30) to[bend right] (a32);
\end{scope}
\end{scope}

\def\ih{0.07}
\begin{scope}[every path/.style={draw, fill, fill opacity=0.1,very thick}]
\newcommand\Interval[3]{
  \pgfmathsetmacro{\xa}{(#1)*\hd - \hd / 3}
  \pgfmathsetmacro{\xb}{(#2)*\hd + \hd / 3}
  \pgfmathsetmacro{\ya}{(#3) - \ih / 2 }
  \pgfmathsetmacro{\yb}{(#3)}
  \pgfmathsetmacro{\yc}{(#3) + \ih / 2 }
  
  \draw (\xa,\yb) -- (\xb,\yb);
  \draw (\xa,\ya) -- (\xa, \yc);
  \draw (\xb,\ya) -- (\xb, \yc);
}
\begin{scope}[darkred]
\Interval{12}{20}{0.25}
\end{scope}
\begin{scope}[blue]
\Interval{8}{20}{0.35}
\end{scope}
\begin{scope}[orange!90!black]
\Interval{6}{25}{0.45}
\end{scope}
\begin{scope}[green!50!black]
\Interval{3}{29}{0.55}
\end{scope}
\begin{scope}[magenta]
\Interval{2}{32}{0.65}
\end{scope}
\end{scope}

\node[right] () at (0.3,-2.3) { festoons $\textcolor{darkred}{X_1} \prec \textcolor{blue}{X_2} \prec \textcolor{orange!90!black}{X_3} \prec \textcolor{green!50!black}{X_4} \prec \textcolor{magenta}{X_5}$};

\begin{scope}[every node/.style={ns}]
      \node[fill=black] (g) at (13.5,-1.9) {};
      \node (b) at (13.5,-2.1) {};
\end{scope}
\node[right=3pt] () at (g) {$v$-good vertex};
\node[right=3pt] () at (b) {$v$-bad vertex};

\end{tikzpicture}
 \end{center}
\caption{An example of a minimal set of festoons connecting $v$ to a $v$-good vertex. 
Different festoons are drawn in different colors and the top part of the picture shows the corresponding festoon intervals.
The filled vertices are those that are $v$-good.
\label{fig:minimal_connecting_set_of_festoons}}
\end{figure}
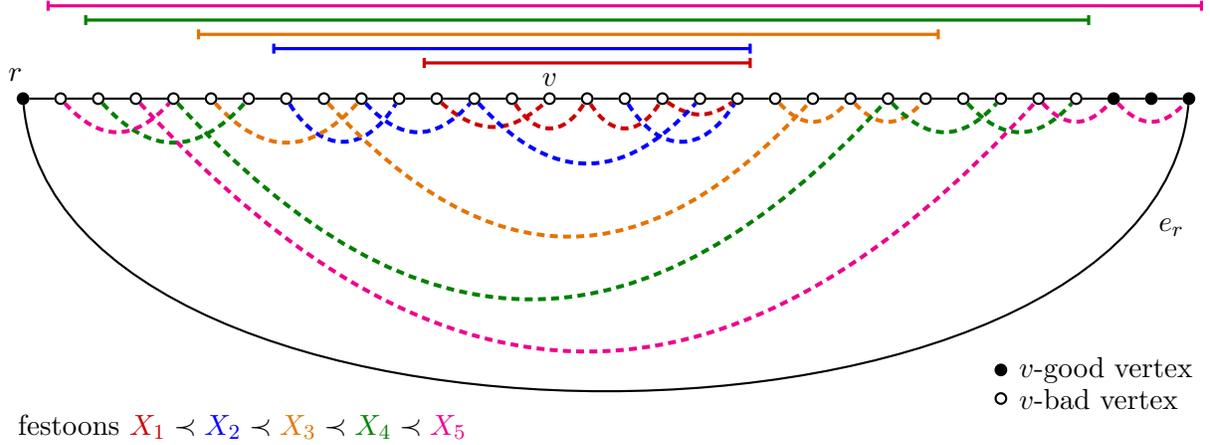

\begin{lemma}\label{lem:properties_dependency_paths}
Let $v\in V\setminus \{r\}$ and let $\mathcal{X}_{v} \subseteq \mathcal{X}$ be a minimal set of festoons that connects $v$ to a $v$-good vertex.
Then the family $\{I_X : X \in \mathcal{X}_{v}\}$ is a chain, i.e., we can number the festoons in $\Xscr_v$ such that $X_1 \prec X_2 \prec \dots \prec X_p$.
Moreover:
\begin{enumerate}
    \item\label{item:consecutive_tangled} $X_i$ and $X_j$ are tangled if and only if $|i-j|=1$.
    \item\label{item:first_contains_v} $X_1$ is the only festoon in $\Xscr_v$ with a link incident to $v$.
    \item\label{item:last_contains_v_good} $X_p$ is the only festoon in $\Xscr_v$ with a link incident to a $v$-good vertex.
\end{enumerate}
\end{lemma}
This structural result of minimal subsets $\Xscr_v$ that allow for dropping a link $(u,v)\in\vec{F}$ is one of the main motivations for considering festoons.
While it might at first sight seem more natural to consider minimal subsets of the solution $S$ that allow for dropping a link, \cref{fig:DropLinkExampleJumping} provides an example illustrating that such link sets seem to have little structure.
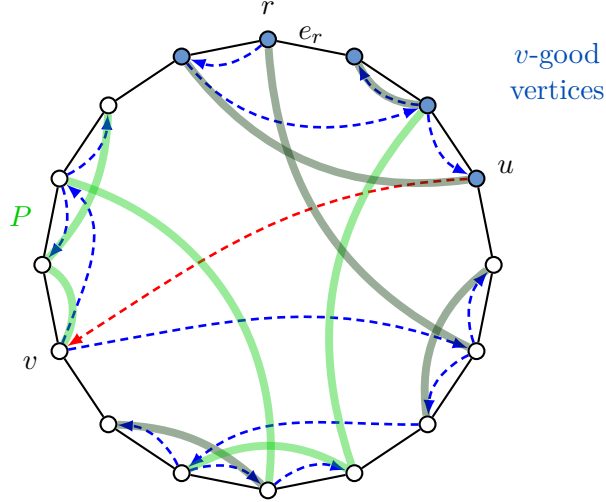
\begin{figure}[!ht]
\begin{center}
\begin{tikzpicture}[scale=1.5,
ns/.style={thick,draw=black,fill=white,circle,minimum size=6,inner sep=2pt},
es/.style={thick},
lks/.style={line width=1pt, blue, densely dashed},
dlks/.style={lks, -latex},
ts/.style={every node/.append style={font=\scriptsize}}
]

\def\rad{2}
\def\num{16}

\def\vlfac{1.14}

\colorlet{goodcol}{blue!70!green}
\colorlet{dacol}{red}
\colorlet{Scol}{green!20!black}
\colorlet{Kcol}{green!80!black}

\pgfkeyssetvalue{/tikz/pics/interval/color}{green!50!black}
\pgfkeyssetvalue{/tikz/pics/interval/radius}{0.15*\rad}
\tikzset{
    pics/interval/.style 2 args={
        code={
               \def\cw{\pgfkeysvalueof{/tikz/pics/interval/radius}}
               \colorlet{col}{\pgfkeysvalueof{/tikz/pics/interval/color}}

               \draw[col, fill=col, fill opacity=0.2] (#1:\rad+\cw) arc (#1:#2:\rad+\cw)
               arc (#2:#2+180:\cw)
               arc (#2:#1:\rad-\cw)
               arc (#1+180:#1+360:\cw);
        }
    },
}
\tikzset{
    pics/intervalN/.style 2 args={
        code={
               \def\cw{\pgfkeysvalueof{/tikz/pics/interval/radius}}
               \colorlet{col}{\pgfkeysvalueof{/tikz/pics/interval/color}}

               \pgfmathanglebetweenpoints{\pgfpoint{0cm}{0cm}}{\pgfpointanchor{#1}{center}}
               \edef\angA{\pgfmathresult}
               \pgfmathanglebetweenpoints{\pgfpoint{0cm}{0cm}}{\pgfpointanchor{#2}{center}}
               \edef\angB{\pgfmathresult}

               \tikzmath{
                   if \angB < \angA then {
                       \angB = \angB + 360;
                   };
               }

               \draw[col, fill=col, fill opacity=0.2] (\angA:\rad+\cw) arc (\angA:\angB:\rad+\cw)
               arc (\angB:\angB+180:\cw)
               arc (\angB:\angA:\rad-\cw)
               arc (\angA+180:\angA+360:\cw);
        }
    },
}

\coordinate (c) at (0,0);

\begin{scope}[every node/.style={ns}]
\foreach \i in {1,...,\num} {
  \pgfmathsetmacro\r{90+(\i-1)*360/\num}
  \ifthenelse{\i < 3 \OR \i > 13}
  {
    \node[fill=goodcol, fill opacity=0.6] (\i) at (\r:\rad) {};
  }{
    \node (\i) at (\r:\rad) {};
  }
}
\end{scope}

\begin{scope}
\path (\num) to node[above=-2pt] {$e_r$} (1);
\end{scope}

\begin{scope}[es]
\foreach \i in {1,...,\num} {
\pgfmathtruncatemacro\j{1+mod(\i,\num)}
\draw (\i) -- (\j);
}
\end{scope}

\begin{scope}[dlks]
\draw (1) to[bend left] (2);
\draw[dacol] (14) to[bend right, in=-170, out=-20] (6);
\draw (6) to[bend right, in=-130, out=-20] (4);
\draw (4) to[bend right] (3);
\draw (4) to[bend left] (5);
\draw (6) to[bend left, in=160, out=10] (12);
\draw (12) to[bend left] (13);
\draw (12) to[bend right] (11);
\draw (11) to[bend right, out=-10] (8);
\draw (8) to[bend right] (7);
\draw (8) to[bend left] (9);
\draw (9) to[bend left] (10);
\draw (2) to[bend right, in=-150, out=-40] (15);
\draw (15) to[bend left] (16);
\draw (15) to[bend right] (14);
\end{scope}

\begin{scope}[lks,solid,line width=3pt,relative,Scol,opacity=0.4]
\draw[Kcol] (5) to[out=-20,in=-160] (3);
\draw[Kcol] (5) to[out=50,in=140] (6);
\draw (7) to[out=20,in=160] (9);
\draw[Kcol] (4) to[out=40,in=140] (9);
\draw[Kcol] (8) to[out=30,in=150] (10);
\draw (11) to[out=30,in=150] (13);
\draw[Kcol] (10) to[out=30,in=150] (15);
\draw (12) to[out=30,in=150] (1);
\draw (2) to[out=-30,in=-150] (14);
\draw (15) to[out=30,in=150] (16);
\end{scope}

\begin{scope}
\node at ($(c)!\vlfac!(1)$) {$r$};
\node at ($(c)!\vlfac!(14)$) {$u$};
\node at ($(c)!\vlfac!(6)$) {$v$};
\end{scope}

\begin{scope}[Kcol]
\coordinate (tn) at ($(4)!0.5!(5)$);
\node at ($(c)!\vlfac!(tn)$) {$P$};
\end{scope}

\begin{scope}
\def\irad{0.22*\rad}

\begin{scope}[goodcol]
\node[align=center,right] at (40:\rad cm + \irad + 4ex) {$v$-good\\ vertices};
\end{scope}
\end{scope}

\end{tikzpicture}
 \end{center}
\caption{Example of a non-shortenable directed WRAP solution $\vec{F}$ (dashed blue and red arrows) together with an undirected WRAP solution $S$ (thick lines in light and dark green).
One arc of $\vec{F}$ is highlighted in red, with tail $u$ and head $v$.
The $v$-good vertices are colored in blue.
A smallest subset $P\subseteq S$ with $(u,v)\in \Drop_{\vec{F}}(K)$ is shown in light green.
Note that while the links in $P$ induce a path in the dependency graph $H$ and thus have a natural ordering, this ordering neither sorts the links ``from left to right'' or vice versa (as a festoon ordering does), nor is it related to the arborescence $(V,\vec{F})$.
In contrast, \cref{lem:properties_dependency_paths} shows that any minimal set $\mathcal{X}_{v} \subseteq \mathcal{X}$ of festoons that allows for dropping a link $(u,v)\in\vec{F}$ has an order related to the structure of the laminar family $\{I_X : X \in \Xscr\}$.
}
\label{fig:DropLinkExampleJumping}
\end{figure}
\begin{proof}[Proof of \cref{lem:properties_dependency_paths}]
Because $\Xscr_v$ connects $v$ to a $v$-good vertex, there exists a path $P$ in the link intersection graph of $\bigcup_{X \in \Xscr_v} X$ from a link incident to $v$ to a link incident to a $v$-good vertex.
Using that the link intersection graph of every festoon is connected, we may assume that the path $P$ visits every festoon $X\in \Xscr_v$ only once, i.e., if $P$ visits two links $\link{l},\link{f} \in X$, then all links visited by $P$ between $\link{l}$ and $\link{f}$ are contained in $X$.
If this was not the case, we could replace the subpath between $\link{l}$ and $\link{f}$ by a path in the link intersection graph of $X$ (and remove cycles if the resulting walk is not a path).
By minimality of $\Xscr_v$ we can thus conclude that the path $P$ visits every festoon in $\Xscr_v$ exactly once.
This induces an order $X_1,\dots, X_p$ on the festoons in $\Xscr_v$ where, for $i<j$, the festoon $X_i$ is visited before $X_j$ by the path $P$.
Because $P$ goes from a link incident to $v$ to one incident to a $v$-good vertex, we have that $X_1$ contains a link incident to $v$ and $X_p$ a link incident to a $v$-good vertex.
We now show that $X_1 \prec X_2 \prec \dots \prec X_p$ and \labelcref{item:consecutive_tangled,item:first_contains_v,item:last_contains_v_good} are fulfilled.

To prove \ref{item:consecutive_tangled}, we first observe that for $i\in \{1,\dots,p-1\}$ the path $P$ contains an edge from a link in $X_i$ to a link in $X_{i+1}$.
Because $P$ is a path in the link intersection graph, these links are intersecting and hence $X_i$ and $X_j$ are tangled.
Moreover, for $i,j\in \{1,\dots,p\}$ with $i< j-1$, the festoons $X_i$ and $X_j$ cannot be tangled as otherwise we could modify the path $P$ (without changing its endpoints) such that it visits the festoon $X_j$ directly after the festoon $X_i$, implying that the set $\{X_1, \dots, X_i, X_j, \dots, X_p\}$ of festoons connects $v$ to a $v$-good vertex and thus contradicting the minimality of $\Xscr_v$.
Such a modification of the path $P$ is possible because the link intersection graph of every festoon is connected.

Next, we show \ref{item:first_contains_v}.
Indeed, if there was a festoon $X_i \in \Xscr_v$ with $i>1$ that contains a link $\link{l}$ incident to $v$, then the festoons $\{X_i, \dots, X_p\}$ are such that $\cup_{j=i}^p X_j$ are connected set of links containing a link incident to $v$ (in $X_i$) one one incident to a $v$-good vertex (in $X_p$).
This contradicts minimality of $\mathcal{X}_v$.
An analogous argument implies \ref{item:last_contains_v_good}.

Finally we prove $X_1 \prec X_2 \prec \dots \prec X_p$.
Let $i$ be the smallest index such that $X_i \prec X_{i+1} \prec \dots \prec X_p$.
Suppose for the sake of deriving a contradiction that $i>1$.
Then, by \ref{item:consecutive_tangled}, the festoons $X_{i-1}$ and $X_i$ are tangled and thus the intersection $I_{X_{i-1}} \cap I_{X_i}$ of their festoon intervals is nonempty.
Because  $\{ I_X: X\in\Xscr\}$  is laminar (\cref{lem:festoon-intervals_laminar}), the minimality of the index $i$ implies $X_i \preceq X_{i-1}$.
Thus, $I_{X_i} \subseteq I_{X_p} \cap I_{X_{i-1}}$ and, because $\{ I_X: X\in\Xscr\}$  is laminar, this implies $X_{i-1} \preceq X_p$ or $X_p \preceq X_{i-1}$.
By \ref{item:last_contains_v_good} and \cref{lem:v-good_in_interval_suffices}, the festoon interval $I_{X_{p}}$ contains a $v$-good vertex but the festoon interval $I_{X_{i-1}}$ does not.
Hence, $X_{i-1} \prec X_p$.

We have shown $X_i \preceq X_{i-1} \prec X_p$.
Because $X_i \prec X_{i+1} \prec \dots \prec X_p$ and $\{ I_X: X\in\Xscr\}$  is laminar, there exists an index $j> i$ such that 
$X_{j-1} \preceq X_{i-1} \prec X_j$.
By \ref{item:consecutive_tangled}, the festoons $X_{j-1}$ and $X_j$ are tangled, which by \cref{cor:sandwiched_festoons} implies that $X_{i-1}$ and $X_j$ are tangled.
This contradicts \ref{item:consecutive_tangled}.
\end{proof}

For every vertex $v\in V\setminus \{r\}$, we will now carefully choose a minimal set $\Xscr_v \subseteq \Xscr$ of festoons connecting $v$ to a $v$-good vertex.
Such a set exists because $\Xscr$ is a partition of a WRAP solution and hence $\Xscr$ connects $v$ to a $v$-good vertex.
In order to prove the decomposition theorem (\cref{thm:ov_decomposition-theorem}), we will then show that one can partition the solution $S$ in such a way into $4\lceil\sfrac{1}{\epsilon}\rceil$-components that, for all but an $\epsilon$-fraction of the links $(u,v)\in \vec{F}$ (in terms of cost), the festoons in $\Xscr_v$ are subsets of the same component $K$, in which case $(u,v)$ belongs to $\Drop_{\vec{F}}(K)$.

By \cref{lem:properties_dependency_paths}, we can number the festoons in a minimal set $\Xscr_v \subseteq \Xscr$ of festoons connecting $v$ to a $v$-good vertex as $ X_1,\dots, X_p$ such that $X_1 \prec X_2 \prec \dots \prec X_p$.
Among all minimal subsets of $\Xscr$ that connect $v$ to a $v$-good vertex, we choose $\Xscr_v$ such that $X_p$ is minimal with respect to the order~$\prec$.
Then we can represent $\Xscr_v$ by the directed path with vertex set $\Xscr_v$ and arc set 
\[
P_v \coloneqq \Big\{ (X_{i}, X_{i-1}) : i \in \{2,3,\dots,p\} \Big\}.
\]
We call the union of these paths the \emph{dependency graph} (\cref{def:dependency_graph}).
Because we do not require all links in $\vec{F}$ to be contained in a set $\Drop_{\vec{F}}(K)$ for some component $K$, but only those in $\vec{F}\setminus \vec{R}$ for some cheap set $\vec{R}$, it will be useful to consider the $U$-dependency graph for a subset $U\subseteq V\setminus \{r\}$, which is the union of only the paths $P_v$ with $v\in U$.

\begin{definition}[$U$-dependency graph]\label{def:dependency_graph}
For $U\subseteq V\setminus \{r\}$, the \emph{$U$-dependency graph} is the directed graph with vertex set $\Xscr$ and arc set
\[
A \coloneqq\ \bigcupp_{u\in U} P_u
\]
being the disjoint union of the sets $P_u$ with $u\in U$.
\end{definition}

A similar concept has been used in the context of WTAP \cite{cohen_2013_approximation, traub_2021_better}, but there the vertices of the dependency graph were links rather than festoons.
Next, we show that the $U$-dependency graph is a branching for any $U\subseteq V\setminus\{r\}$, i.e., every connected component of the $U$-dependency graph is an arborescence.
For this we need the following simple yet crucial observation.

\begin{lemma}\label{lem:not_bad_in_both_directions}
For any distinct $v,w \in V\setminus \{r\}$, either $v$ is $w$-good or $w$ is $v$-good.
\end{lemma}
\begin{proof}
If $v$ is not $w$-good, then $v$ is a descendant of $w$ and thus $w$ is $v$-good.
\end{proof}

\begin{lemma}\label{lem:dependency_graph_is_branching}
For every $U\subseteq V\setminus\{r\}$, the $U$-dependency graph is a branching.
\end{lemma}
\begin{proof}
Because $X \prec Y$ for every arc $(Y,X)\in A$, it suffices to show that no festoon $X\in \Xscr$ has more than one incoming arc in $A$.
Suppose for the sake of deriving a contradiction that there exists a festoon $X\in \Xscr$ with two distinct incoming arcs $(Y,X)\in A$ and $(Z,X) \in A$.
Let $v,w\in V$ such that $(Y,X)\in P_v$ and $(Z,X)\in P_w$.
Because $(\Xscr_v, P_v)$ is a path, we have $v\neq w$.

By \cref{lem:not_bad_in_both_directions}, we may assume without loss of generality that $w$ is $v$-good.
Then \cref{lem:properties_dependency_paths} implies $w\in I_X$.
Therefore, by \cref{lem:v-good_in_interval_suffices}, the festoon $X$ contains a link incident to a $v$-good vertex.
This contradicts \cref{lem:properties_dependency_paths}~\ref{item:last_contains_v_good} applied to $\Xscr_v$ because $X$ has an incoming arc in $P_v$.
\end{proof}

\subsection{Connecting Thinness to the Dependency Graph}\label{sec:thinness_dependency_graph}

Consider a set $U\subseteq V\setminus\{r\}$.
Then if we define a partition $\Kscr$ of $S$ by 
\[
 \textstyle \Kscr \coloneqq \Big\{ \bigcup_{X\in \Xscr'} X : (\Xscr',A')\text{ is a connected component of the $U$-dependency graph} \Big\},
\]
every link $(v,u)\in \vec{F}$ with $u\in U$ is contained in $\Drop_{\vec{F}}(K)$ for some component $K\in \Kscr$, because all festoons in $\Xscr_u$ are part of the same connected component of the $U$-dependency graph.
The goal of this section is to develop a sufficient criterion for when we can guarantee the components in $\Kscr$ to be $\alpha$-thin.

\begin{lemma}\label{lem:few_tangled_imply_thinness}
Let $\alpha\in \mathbb{Z}_{\geq 0}$ and $U \subseteq V\setminus \{r\}$. 
Let $(\Xscr',A')$ be a connected component of the $U$-dependency graph.
If every festoon $X\in \Xscr'$ is tangled with at most $\alpha$ festoons in 
\[
\left\{ Z \in \Xscr'\setminus\{X\}: X \preceq Z\right\},
\]
then the link set $\bigcup_{X\in \Xscr'} X$ is $4(\alpha +1)$-thin.
\end{lemma}
\begin{proof}
By \cref{lem:festoon-intervals_laminar}, the family $\{ I_X : X\in \Xscr'\}$ of festoon intervals is laminar.
All intervals in this family except for the interval $V$ are 2-cuts contained in $\Cscr_G$.
We extend the family $\{ I_X : X\in \Xscr'\} \setminus \{V\}$ to a maximal laminar subfamily of $\Cscr_G$, which we denote by $\Lscr$.
Let $K \coloneqq \bigcup_{X\in \Xscr'} X$.
We will prove that, for every cut $C\in \Lscr$, we have $|K\cap \delta(C)| \le 4(\alpha +1)$, which implies that $K$ is $4(\alpha +1)$-thin.

Hence, we fix a cut $C\in \Lscr$. 
First, we observe that, for every festoon $X\in \Xscr'$ with $I_X \cap C = \emptyset$ or $I_X \subseteq C$, we have $\delta_X(C) = \emptyset$.
The construction of $\Lscr$ guarantees that $C$ does not cross any festoon interval $I_X$ with $X\in \Xscr'$.
Thus, every festoon $X\in \Xscr'$ with $\delta_X(C) \neq \emptyset$ fulfills $C\subseteq I_X$.
Let $Y \in \Xscr'$ be minimal with respect to the order $\prec$ among all festoons $X \in \Xscr'$ with $C\subseteq I_X$.
Now consider a festoon $Z\in \Xscr'\setminus\{X\}$ with $\delta_Z(C) \neq \emptyset$.
Then, because $C\subseteq I_Y \cap I_Z$, laminarity of the festoon intervals (\cref{lem:festoon-intervals_laminar}) implies $Z \preceq Y$ or $Y \preceq Z$ and hence the minimality of $Y$ with respect to $\prec$ yields $Y \preceq Z$.
Combining $\delta_Z(C) \neq \emptyset$ and $C\subseteq I_Y$, we conclude that the festoon $Z$ contains a link that is incident to a vertex in $I_Y$.
By \cref{lem:characterize_tangled}, this implies that $Z$ and $Y$ are tangled.

We have shown that every festoon $Z\in \Xscr'\setminus\{Y\}$ with $\delta_Z(C) \neq \emptyset$ is tangled with the festoon $Y$ and fulfills $Y\preceq Z$.
Thus, by assumption of the lemma, there are at most $\alpha$ such festoons $Z$ and hence there are at most $\alpha+1$ festoons $X\in \Xscr'$ with $\delta_X(C) \neq \emptyset$.
Moreover, by \cref{lem:festoon_light}, we have $|\delta_X(C)|\le 4$ for each such festoon $X$, implying $|K\cap \delta(C)| \le 4(\alpha +1)$ as desired.
\end{proof}

Consider a connected component $(\Xscr',A')$ of the dependency graph. 
In order to get a bound on the number of festoons $Z\in \Xscr'$ that are tangled with a festoon $Y\in \Xscr'$ with $Y \preceq Z$, we next show that any two festoons in $\Xscr'$ that are tangled must have an ancestry relation in the arborescence $(\Xscr', A')$.
The proof of this lemma is where we will use our particular choice of $\Xscr_v$, i.e., we will exploit that we did not choose $\Xscr_v$ as an arbitrary minimal set of festoons from $\Xscr$ that connect $v$ to a $v$-good vertex.

\begin{lemma}\label{lem:ancestry_relation}
Let $U\subseteq V\setminus \{r\}$.
If two festoons in the same connected component of the $U$-dependency graph $D=(\Xscr, A)$ are tangled, then they have an ancestry relation in $D$.
\end{lemma}
\begin{proof}
Let $X,Y \in \Xscr$ be two distinct tangled festoons that are part of the same connected component of the dependency graph $D$, which is is an arborescence by \cref{lem:dependency_graph_is_branching}.
Then $I_X \cap I_Y \ne \emptyset$ and hence we have $X \preceq Y$ or $Y \preceq X$, say without loss of generality $X \preceq Y$.
Let $Z\in \Xscr$ be the root of the connected component of $D$ containing $X$ and $Y$.
If $Z= Y$, then $Y$ is an ancestor of $X$ and we are done.
Thus, assume from now on that this is not the case, which implies $Y\prec Z$ because the dependency graph contains a $Z$-$Y$ path.

We consider the path $P$ in the dependency graph from the root $Z$ to the festoon $X$.
For every festoon $\overline{X}$ on the path $P$, we have $X \preceq \overline{X}$ by the construction of the dependency graph.
Because $I_X \subseteq I_Y \cap I_{\overline{X}}$, laminarity of the festoon intervals (\cref{lem:festoon-intervals_laminar}) implies $Y \preceq \overline{X}$ or $\overline{X} \preceq Y$ for every festoon $\overline{X}$ on the path $P$.
Hence, $X \preceq Y \prec Z$ implies that the $Z$-$X$ path $P$ contains an arc $(X^{\rm{above}}, X^{\rm{below}})$ such that $X \preceq X^{\rm{below}}\preceq Y \prec X^{\rm{above}}$.
See \cref{fig:proof_lemma_ancestry_relation}.
\begin{figure}[!ht]
\begin{center}
\begin{tikzpicture}[scale=1,
fs/.style={very thick,draw,rectangle,minimum size=4,inner sep=2pt, outer sep=1pt},
as/.style={very thick, -latex},
]

\def\num{15}
\def\hd{0.5}

\def\ih{0.25}
\begin{scope}[every path/.style={draw, fill, very thick}]
\newcommand\Interval[3]{
  \pgfmathsetmacro{\xa}{(#1)*\hd - \hd / 3}
  \pgfmathsetmacro{\xb}{(#2)*\hd + \hd / 3}
  \pgfmathsetmacro{\ya}{(#3) - \ih / 2 }
  \pgfmathsetmacro{\yb}{(#3)}
  \pgfmathsetmacro{\yc}{(#3) + \ih / 2 }
  
  \draw (\xa,\yb) -- (\xb,\yb);
  \draw (\xa,\ya) -- (\xa, \yc);
  \draw (\xb,\ya) -- (\xb, \yc);
}
\begin{scope}[darkred]
\Interval{5}{6}{0.5}
\node[right] () at (4.7,0.5) {$I_X$};
\end{scope}
\begin{scope}[black]
\Interval{4}{7}{1.5}
\node[right] () at (4.7, 1.5) {$I_{X^{\rm{below}}}$};
\end{scope}
\begin{scope}[blue]
\Interval{4}{7}{2}
\node[right] () at (4.7,2) {$I_Y$};
\end{scope}
\begin{scope}[black]
\Interval{3}{8}{3}
\node[right] () at (4.7, 3) {$I_{X^{\rm{above}}}$};
\end{scope}
\begin{scope}[orange!80!black]
\Interval{2}{9}{4}
\node[right] () at (4.7,4) {$I_Z$};
\end{scope}
\end{scope}

\begin{scope}[every node/.style={fs}, shift={(1,0)}]
\node[orange!80!black] (z) at (9.5,4) {};
\node[darkred] (x) at (11.5,0.5) {};
\node[black] (xb) at (10.77,1.5) {};
\node[black] (xa) at (10.17,3) {};
\node[blue] (y) at (8,2) {};
\end{scope}

\node[darkred, right] at (x)  {$X$};
\node[orange!70!black, above=2pt]  at (z)  {$Z$};
\node[black, right] at (xb) {$X^{\rm{below}}$};
\node[black, above right] at (xa) {$X^{\rm{above}}$};
\node[blue, left] at (y) {$Y$};

\begin{scope}[as]
\draw (xa) to node[right, pos=0.4] {$\in P_v$} (xb);
\draw (9.6,3) to node[right, pos=0.35] {$\in P_w$} (y);
\end{scope}

\draw[as, dotted, gray] (z) to[out=-90,in=110, looseness=1.8] (xa);
\draw[as, dotted, gray] (xb) to[out=-90,in=110, looseness=1.8] (x);

\end{tikzpicture}
 \end{center}
\caption{An illustration of the proof of \cref{lem:ancestry_relation}.
The left part of the figure shows the festoon intervals of the festoons $X \preceq X^{\rm{below}} \preceq Y \prec X^{\rm{above}} \preceq Z$.
The right part of the figure illustrates the connected component of the dependency graph that contains these festoons.
It contains an arc $(X^{\rm{above}}, X^{\rm{below}})$ on the $Z$-$X$ path and the festoon $Y$ has an incoming arc, which does not lie on the $Z$-$X$ path. \label{fig:proof_lemma_ancestry_relation}}
\end{figure}
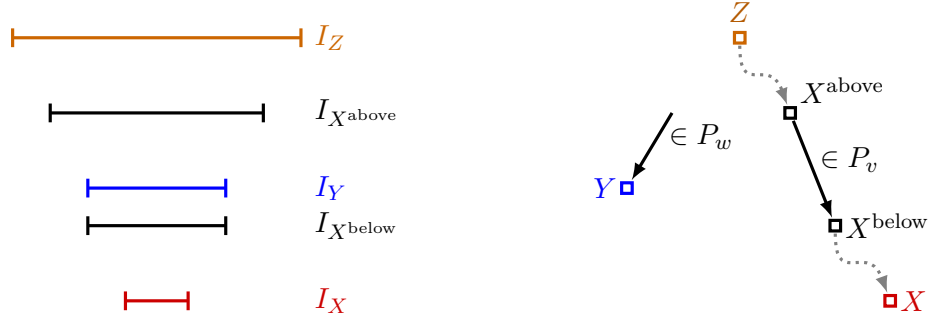
Because $X$ and $Y$ are tangled, \cref{cor:sandwiched_festoons} implies that  $X^{\rm{below}}$ and $Y$ are tangled. 
Let $v\in U$ such that $(X^{\rm{above}}, X^{\rm{below}}) \in P_v$ and let $X_1 \prec X_2 \prec \dots \prec X_p$ be the festoons in $\Xscr_v$.

Next, we show that the festoon interval $I_Y$ contains no $v$-good vertex.
Let the index $i$ be such that $X_i = X^{\rm{below}}$; then $X_{i-1} = X^{\rm{above}}$.
Now suppose that $I_Y$ contains a $v$-good vertex.
Then by \cref{lem:v-good_in_interval_suffices}, $Y$ contains a link incident to a $v$-good vertex .
This implies that the set $\{X_1, \dots, X_i, Y\}$ of festoons connects $v$ to a $v$-good vertex because by \cref{lem:properties_dependency_paths}~\ref{item:first_contains_v} $X_1$ contains a link incident to $v$, by \cref{lem:properties_dependency_paths}~\ref{item:consecutive_tangled} $X_j$ and $X_{j+1}$ are tangled for $j\in \{1,\dots, i-1\}$, and also $X_i = X^{\rm{below}}$ and $Y$ are tangled.
Because $Y \prec X^{\rm{above}}$, this contradicts the choice of $\Xscr_v$.
This shows that indeed the festoon interval $I_Y$ contains no $v$-good vertex.

Because the festoon $Y$ is not the root of the connected component of the dependency graph $D$ that contains $X$ and $Y$, it has an incoming arc $a_Y \in \delta^-_A(Y)$.
Let $w\in U$ such that $a_Y \in P_w$.
Then by \cref{lem:properties_dependency_paths}, we have $w\in I_Y$ which implies that $w$ is not $v$-good.
Thus by \cref{lem:not_bad_in_both_directions}, the vertex $v$ is $w$-good.
Because $v\in I_{X^{\rm{below}}} \subseteq I_Y$, this shows that $I_Y$ contains a $w$-good vertex, which by \cref{lem:v-good_in_interval_suffices} implies that the festoon $Y$ contains a link incident to a $w$-good vertex.
By \cref{lem:properties_dependency_paths}~\ref{item:last_contains_v_good}, this contradicts the fact that $Y$ has an incoming arc in $P_w$.
\end{proof}

Using \cref{lem:ancestry_relation}, we can now give the following upper bound on the number of festoons $Y$ in a connected component $(\Xscr', A')$ of the dependency graph that are tangled with $X\in \Xscr'$ and fulfill $X\preceq Y$.

\begin{lemma}\label{lem:few_colors_imply_few_tangled_festoons}
Let $\alpha\in \mathbb{Z}_{\geq 0}$ and $U \subseteq V\setminus \{r\}$. 
Let $(\Xscr',A')$ be a connected component of the $U$-dependency graph.
If for every directed path with arc set $P\subseteq A'$ we have
\begin{equation*}
\left|\left\{ u\in U \colon P \cap P_u \ne \emptyset \right\}\right| \le \alpha,
\end{equation*}
then, for every festoon $X\in \Xscr'$, there are at most $\alpha$ festoons in $\{Y \in \Xscr'\setminus \{X\}: X \preceq Y\}$ that are tangled with $X$.
\end{lemma}
\begin{proof}
Let $X\in \Xscr'$.
By \cref{lem:dependency_graph_is_branching}, the connected component $(\Xscr',A')$ is an arborescence.
Moreover, by \cref{lem:ancestry_relation}, every festoon $Y\in \Xscr'$ with $X \preceq Y$ that is tangled with $X$ must be an ancestor of $X$, i.e., it must be contained in the path $P$ from the root of the arborescence $(\Xscr',A')$ to $X$.
(Here we used that, for every descendant $Z$ of $X$ in the arborescence $(\Xscr',A')$ with $Z\neq X$, we have $X \succ Z$.)

Consider a vertex $u\in U$ such that $P_u \cap P \neq \emptyset$.
Let $X_1,\dots, X_q$ be the festoons that have an outgoing arc in  $P_u \cap P$, where without loss of generality $X_1 \prec X_2 \prec \dots \prec X_q$.
We will show that the festoons $X_2,\dots,X_{q}$ are not tangled with $X$, which implies the statement of the lemma.
See \cref{fig:proof_number_tangled_festoons}.
\begin{figure}[!ht]
\begin{center}
\begin{tikzpicture}[scale=1,
fs/.style={very thick,draw,rectangle,minimum size=5,inner sep=2pt, outer sep=1pt},
as/.style={very thick, -latex},
]

\def\num{10}
\def\hd{1}
\def\vd{0.5}

\begin{scope}
\foreach \i in {1,3,6,7,8} {
  \pgfmathsetmacro\x{\i*\hd}
  \pgfmathsetmacro\y{(\num -\i)*\vd}
  \node[fs] (v\i) at (\x,\y) {};
}
\foreach \i in {2,4,5,9,10} {
  \pgfmathsetmacro\x{\i*\hd}
  \pgfmathsetmacro\y{(\num -\i)*\vd}
  \node[fs, fill=black] (v\i) at (\x,\y) {};
}
\node[fs] (v11) at (2.5, 2.5) {};
\node[fs] (v12) at (2, 1.5) {};
\node[fs] (v13) at (1.5, 0.5) {};
\node[fs] (v14) at (7,2.25) {};
\end{scope}

\node [right] () at (v10) {$X$};

\begin{scope}[as]
\begin{scope}[darkred]
\draw (v1) to (v2);
\draw (v2) to (v3);
\draw (v3) to (v11);
\end{scope}
\begin{scope}[orange]
\draw (v11) to (v12);
\draw (v12) to (v13);
\end{scope}
\begin{scope}[blue]
\draw (v3) to (v4);
\draw (v4) to (v5);
\end{scope}
\begin{scope}[green!80!black]
\draw (v5) to (v6);
\draw (v6) to (v14);
\end{scope}
\begin{scope}[magenta]
\draw (v6) to (v7);
\draw (v7) to (v8);
\draw (v8) to (v9);
\draw (v9) to (v10);
\end{scope}
\end{scope}

\end{tikzpicture}
 \end{center}
\caption{An illustration of the proof of \cref{lem:few_colors_imply_few_tangled_festoons}.
The figure shows a connected component $(\Xscr'A')$ of the dependency graph, where the sets $P_u$ for different vertices $u\in U$ are shown in different colors.
Only the festoons shown as filled squares might be tangled with the festoon $X$ (the festoon at the bottom right). \label{fig:proof_number_tangled_festoons}}
\end{figure}

Let $(X_1 , X_{0})$ be the outgoing arc of $X_1$ in $P_u \cap P$.
Then $(\{X_0, \dots, X_{q}\}, P_u \cap P)$ is a subpath of the path $(\Xscr_u, P_u)$.
Suppose that for some $i\in\{2,\dots, q\}$, the festoons $X_i$ and $X$ are tangled. 
Then, we have $X \preceq X_0 \prec X_i$, which by \cref{cor:sandwiched_festoons} implies that $X_i$ and $X_0$ are tangled, contradicting \cref{lem:properties_dependency_paths}~\ref{item:consecutive_tangled}.
\end{proof}

Combining \cref{lem:few_tangled_imply_thinness} and \cref{lem:few_colors_imply_few_tangled_festoons}, we obtain the main result of this section.

\begin{lemma}\label{lem:few_colors_imply_thinness}
Let $\alpha\in \mathbb{Z}_{\geq 0}$ and $U \subseteq V\setminus \{r\}$. 
Let $(\Xscr',A')$ be a connected component of the $U$-dependency graph.
If for every directed path with arc set $P\subseteq A'$ we have
\begin{equation*}
\left|\left\{ u\in U \colon P \cap P_u \ne \emptyset \right\}\right| \le \alpha,
\end{equation*}
then the link set $\bigcup_{X\in \Xscr'} X$ is $4(\alpha +1)$-thin.
\end{lemma}
\begin{proof}
By \cref{lem:few_colors_imply_few_tangled_festoons}, for every $X\in \Xscr'$ there are at most $\alpha$ festoons in $\{Y\in \Xscr'\setminus \{X\}: X \preceq Y\}$ that are tangled with $X$.
By \cref{lem:few_tangled_imply_thinness}, this implies that  $\bigcup_{X\in \Xscr'} X$ is $4(\alpha +1)$-thin.
\end{proof}

\subsection{Completing the Proof of the Decomposition Theorem}\label{sec:choosing_R}

Finally, we use \cref{lem:few_colors_imply_thinness} to prove the decomposition theorem (\cref{thm:ov_decomposition-theorem}), which we restate here for convenience.
Having established \cref{lem:few_colors_imply_thinness}, the remaining part of the proof is identical to the proof of the decomposition theorem for WTAP in \cite{traub_2021_better}.

\ovDecompositionTheorem*
\begin{proof}
We let $D\coloneqq (\Xscr, A)$ be the $(V\setminus\{r\})$-dependency graph, which is a branching by \cref{lem:dependency_graph_is_branching}.
Then we define a labeling $\lambda\colon A \to \mathbb{Z}_{>0}$ as follows.
For any arc $a=(X,Y)\in A$ for which $\delta^-_A(X) = \emptyset$, we set $\lambda(a)\coloneqq 1$.
For any other arc $a=(X,Y)$ there is a unique incoming arc $\overline{a}$ of $X$ in $A$ and we set 
\[
\lambda(a) \coloneqq 
   \begin{cases}
      \lambda(\overline{a})   & \text{ if there is a vertex }v\in V\setminus\{r\}\text{ s.t. } a,\overline{a}\in P_v, \\
      \lambda(\overline{a})+1 & \text{ otherwise}.
   \end{cases}
\]
Note that, for every vertex $v\in V\setminus \{r\}$, all arcs in $P_v$ have been assigned the same label.
Therefore, the following partition $\vec{R}_0, \vec{R}_1,  \dots , \vec{R}_{q - 1}$ of $\vec{F}$ is well defined, where $q\coloneqq \lceil \sfrac{1}{\epsilon} \rceil$.
For $i\in \{0, 1,  \dots , q - 1\}$, we define
\[
\vec{R}_i \coloneqq \Big\{ (u,v)\in \vec{F} \ : \ \lambda(a) \equiv i \pmod{q} \text{ for } a\in P_v \Big\}.
\]
Let $\vec{R} \in \{\vec{R}_0, \vec{R}_1,  \dots , \vec{R}_{q - 1}\}$ be such that $c(\vec{R})$ is minimum. 
Then $c(\vec{R}) \le \frac{1}{q} \cdot c(\vec{F}) \le \epsilon \cdot c(\vec{F})$.

Let 
\[
U \coloneqq \Big\{ v\in V: \delta^-_{\vec{F}\setminus\vec{R}}(v) \neq \emptyset \Big\}.
\]
Then the $U$-dependency graph arises from $D$ by deleting the arcs in $\bigcup_{(v,w)\in \vec{R}} P_w$.
Thus, for every path in the dependency graph $D$, we deleted the arcs of every $q$th label and hence along every remaining path in the $U$-dependency graph, there occur at most $q-1$ distinct labels.
By the construction of the labeling $\lambda$, this means that, for every directed path with arc set $P$ in the $U$-dependency graph, there are at most $q-1$ vertices $u\in U$ with $P\cap P_u \neq \emptyset$.
Hence, \cref{lem:few_colors_imply_thinness} implies that, for every connected component $(\Xscr', A')$ of the $U$-dependency graph, the component $K = \bigcup_{X\in \Xscr'} X$ is $4q$-thin. 
Because $\Xscr$ is a partition of $S$, also the components $K$ corresponding to the different connected components  $(\Xscr', A')$ of the $U$-dependency graph form a partition $\mathcal{K}$ of $S$.

Finally, we show that, for every link $\vec{\link{l}}\in \vec{F} \setminus \vec{R}$, there exists a component $(\Xscr', A')$ of the $U$-dependency graph such that 
$\vec{\link{l}} \in \Drop_{\vec{F}}(K)$ for $K = \bigcup_{X\in \Xscr'} X$.
Let $\vec{\link{l}} = (u,v)$ and consider the path $(\Xscr_v, P_v)$.
This path is contained in one of the connected components $(\Xscr', A')$ of the $U$-dependency graph.
Because $\Xscr_v \subseteq \Xscr'$ connects $v$ to a $v$-good vertex, also $K = \bigcup_{X\in \Xscr'} X$ connects $v$ to a $v$-good vertex.
Thus, by \cref{lem:characterize_drop}, we have $\vec{\link{l}} \in \Drop_{\vec{F}}(K)$.
\end{proof}
 \section{Finding Components via Dynamic Programming}\label{sec:dp}
In this section we show \cref{thm:ov_find-component-for-relative-greedy}, which is repeated below for convenience.
Throughout the section, $\alpha\in \mathbb{Z}_{\geq 1}$ is a fixed constant.
We recall that, for a given rooted WRAP instance $(G,L,c,r,e_r)$, we denote by $\mathfrak{K}$ the family of all $\alpha$-thin link sets.
Moreover, such link sets $K\in \mathfrak{K}$ are also called \emph{components}.
\ovFindCompForRelGreedy*

Recall that we interpret the ratio $\sfrac{c(K)}{c(\Drop_{\vec{F}_0}(K)\cap \vec{F})}$ as $1$ if both $c(K)=0$ and $c(\Drop_{\vec{F}_0}(K)\cap \vec{F})=0$.
This reflects the intuition that the ratio captures how much we pay through added links versus how much we gain through links of $\vec{F}$ we can remove.
If both the added links and the ones to be removed have zero cost, we neither gain nor lose anything by adding the component $K$.
If only $c(\Drop_{\vec{F}_0}(K)\cap \vec{F})=0$ but $c(K) > 0$, then we use the usual convention that $\sfrac{c(K)}{c(\Drop_{\vec{F}_0}(K)\cap \vec{F})} = \infty$.

We denote the optimal/minimal ratio by
\begin{equation}\label{eq:rhoStar}
\rho^* \coloneqq \min \left\{ \frac{c(K)}{c(\Drop_{\vec{F}_0}(K) \cap \vec{F})}\colon K \subseteq L \text{ is $\alpha$-thin}\right\}.
\end{equation}
Note that we always have $\rho^*\leq 1$ because one can choose $K=\emptyset$ in the above minimization problem.
We compute a minimizer $K\in \mathfrak{K}$ of \eqref{eq:rhoStar} through a procedure that allows for deciding whether some value $\rho \in \mathbb{R}$ is larger or smaller than $\rho^*$.
This can be achieved by maximizing another objective that is parameterized by $\rho$, which we call \emph{slack} and define in \cref{sec:slackMax}.
We then apply binary search with respect to $\rho$ to obtain the optimal ratio $\rho^*$ together with a desired minimizer.

An analogous approach to reduce the problem of finding a component minimizing a ratio as in~\eqref{eq:rhoStar} to maximizing an auxiliary slack function was already used in~\cite{traub_2021_better} in the context of WTAP.
For WCAP we can use a largely identical reduction.
For completeness, we repeat it in \cref{sec:slackMax}, adapted to our notation.
The way the method we present in this section differs compared to prior work is in how we set up a dynamic program to maximize the auxiliary slack function.
We expand on this in \cref{sec:maxSlackThroughDP}.

\subsection{Reducing to Slack Maximization}\label{sec:slackMax}

Whether some value $\rho\in \mathbb{R}$ is larger or smaller than $\rho^*$ reduces to maximizing the following slack function $\slack_\rho (K)$ over all $K \in \mathfrak{K}$.
\begin{definition}[slack function]
Let $\rho \in \mathbb{R}$ and let $(G,L,c,r,e_r)$ be a rooted WRAP instance with non-shortenable directed solution $\vec{F}_0\subseteq \shadows(L)$ and $\vec{F}\subseteq \vec{F}_0$. For any $K \subseteq L$, we define the \emph{slack} of $K$ (with respect to the WRAP instance, $\vec{F}_0$, $\vec{F}$, and $\rho$) by
\begin{equation*}
\slack_{\rho}(K) \coloneqq \rho \cdot c\bigl(\Drop_{\vec{F}_0}(K) \cap \vec{F}\bigr) - c(K).
\end{equation*}
\end{definition}

The following statement formalizes that comparing a given $\rho$ to $\rho^*$ reduces to maximizing $\slack_{\rho}(K)$ over all $\alpha$-thin components $K\in \mathfrak{K}$.
It implies that $\rho^*$ is the largest $\rho$ for which the maximum value of $\slack_\rho(K)$ among $K\in \mathfrak{K}$ is zero.
\begin{lemma}\label{lem:basicSlackRatioRel}
Let $\rho \in [0,1]$. Then
\begin{equation*}
\max\{\slack_{\rho}(K) \colon K\in \mathfrak{K}\} >0 \iff \rho > \rho^*.
\end{equation*}
\end{lemma}
\begin{proof}
We start with the left-to-right implication.
Hence, let $\overline{K}\in \mathfrak{K}$ with $\slack_\rho(\overline{K})>0$.
This implies in particular $c(\Drop_{\vec{F_0}}(\overline{K})\cap \vec{F})>0$.
We now obtain $\rho > \rho^*$ from 
\begin{equation*}
0 < \frac{\slack_{\rho}(\overline{K})}{c(\Drop_{\vec{F}_0}(\overline{K})\cap \vec{F})} = \rho - \frac{c(\overline{K})}{c(\Drop_{\vec{F}_0}(\overline{K})\cap \vec{F})} \leq \rho - \rho^*,
\end{equation*}
where the equality follows from the definition of $\slack_{\rho}(K)$, and the final inequality from the definition of $\rho^*$.

Conversely, assume $\rho > \rho^*$ and let $K^*\in \argmin\{\sfrac{c(K)}{c(\Drop_{\vec{F}_0}(K)\cap \vec{F})} \colon K \in \mathfrak{K}\}$.
Hence,
\begin{equation}\label{eq:rhoStarRealizedByKStar}
\rho^* = \frac{c(K^*)}{c(\Drop_{\vec{F}_0}(K^*)\cap \vec{F})}.
\end{equation}
Because we assume $\rho > \rho^*$ and $\rho \leq 1$, we have $\rho^*<1$ and thus $c(\Drop_{\vec{F}_0}(K^*)\cap \vec{F})>0$.
(Having $\rho^* < 1$ excludes the case $\rho^*=1$ which could be achieved with a zero numerator and denominator in~\eqref{eq:rhoStarRealizedByKStar}.)
The desired statement now follows from
\begin{align*}
\max\{\slack_{\rho}(K) \colon K\in \mathfrak{K}\}
  &\geq \slack_{\rho}(K^*)\\
& = \rho \cdot c(\Drop_{\vec{F}_0}(K^*)\cap \vec{F}) - c(K^*)\\
&> \rho^* \cdot c(\Drop_{\vec{F}_0}(K^*)\cap \vec{F}) - c(K^*)\\
&= 0,
\end{align*}
where the strict inequality follows from $\rho > \rho^*$ and $c(\Drop_{\vec{F}_0}(K^*)\cap \vec{F})>0$, and the last equality follows from \eqref{eq:rhoStarRealizedByKStar}.
\end{proof}

Thus, to use binary search to determine $\rho^*$, it suffices to be able to solve, for a given $\rho\in [0,1]$, the slack maximization problem
\begin{equation}\label{eq:slackMaxProblem}
\max\{\slack_{\rho}(K) \colon K \in \mathfrak{K}\}.
\end{equation}
We show in \cref{sec:maxSlackThroughDP} how this slack maximization problem can be solved efficiently via a dynamic program.
Before expanding on this dynamic program, we summarize the result we get from it below and show that, indeed, it allows for determining an $\alpha$-thin component minimizing $\sfrac{c(K)}{c(\Drop_{\vec{F}_0}(K)\cap \vec{F})}$.
\begin{lemma}\label{lem:slack_maximization}
Given $\rho\in [0,1]$, a rooted WRAP instance $(G,L,c,r,e_r)$ with a non-shortenable directed solution $\vec{F}_0\subseteq \shadows(L)$, and $\vec{F}\subseteq \vec{F}_0$, we can compute in polynomial time a maximizer of $\max\{\slack_{\rho}(K) \colon K \in \mathfrak{K}\}$.
\end{lemma}
We now show how \cref{lem:slack_maximization} implies \cref{thm:ov_find-component-for-relative-greedy}.
\begin{proof}[Proof of \cref{thm:ov_find-component-for-relative-greedy}]
Without loss of generality, we assume that the link costs $c$ are integral, i.e., $c\colon L \to \mathbb{Z}_{\geq 0}$, which can be achieved by scaling up the costs if necessary.
Moreover, we assume $\vec{F}\neq \emptyset$; for otherwise we can simply return $\overline{K}\coloneqq \emptyset$ as a maximizer of $\sfrac{c(K)}{c(\Drop_{\vec{F}_0}(K)\cap \vec{F})}$ among all $K\in \mathfrak{K}$.
Observe that $\vec{F}\neq \emptyset$ implies existence of trivial components $\overline{K}\in \mathfrak{K}$ with $\sfrac{c(\overline{K})}{c(\Drop_{\vec{F}_0}(K)\cap \vec{F})}\leq 1$ and $\Drop_{\vec{F}_0}(K)\cap \vec{F}\neq\emptyset$.
Indeed, let $\vec{\link{l}}\in \vec{F}$ and let $\link{l}\in L$ be the link with respect to which $\vec{\link{l}}$ is a shadow.
Hence, $c(\link{l}) = c(\vec{\link{l}})$, and therefore $\overline{K}=\{\link{l}\}$ is an $\alpha$-thin component with $\sfrac{c(\overline{K})}{c(\Drop_{\vec{F}_0}(\overline{K})\cap \vec{F})}\leq 1$ because $\vec{\link{l}}\in \Drop_{\vec{F}_0}(K)$, which also implies $\Drop_{\vec{F}_0}(K)\cap \vec{F}\neq \emptyset$.
Thus, we can furthermore assume $c(\vec{F}) > 0$; for otherwise, one can return $\overline{K}=\{\link{l}\}$.

To employ binary search, observe that, for each $\rho\in [0,1]$, \cref{lem:slack_maximization} allows for efficiently computing $\eta \coloneqq \max\{\slack_{\rho}(K)\colon K \in \mathfrak{K}\}$, and, by \cref{lem:basicSlackRatioRel}, we know that $\rho^* < \rho$ if $\eta >0$ and $\rho^*\geq \rho$ if $\eta \leq 0$.
Hence, we can use binary search to determine in polynomial time an interval $[a,b]\subseteq [0,1]$ with $\rho^*\in [a,b]$ and $b-a < \sfrac{1}{c(\vec{F})^2}$, together with an $\alpha$-thin component $\overline{K}\in \mathfrak{K}$ satisfying $\Drop_{\vec{F}_0}(\overline{K})\cap \vec{F}\neq\emptyset$ and 
\begin{equation}\label{eq:upperBoundComp}
\frac{c(\overline{K})}{c(\Drop_{\vec{F}_0}(\overline{K})\cap \vec{F})} \leq b.
\end{equation}
That we can obtain a component $\overline{K}\in \mathfrak{K}$ with $\Drop_{\vec{F}_0}(\overline{K})\cap \vec{F}\neq\emptyset$ and $\sfrac{c(\overline{K})}{c(\Drop_{\vec{F}_0}(\overline{K})\cap \vec{F})} \leq b$ follows by observing that either $b=1$, in which case we can set $\overline{K}=\{\link{l}\}$ as discussed above, or the upper bound $b$ was obtained through the computation of $\eta\coloneqq \max\{\slack_{b}(K)\colon K \in \mathfrak{K}\}$ and we had $\eta >0$.
In the latter case, we also computed a maximizer $\overline{K}$ of $\max\{\slack_{b}(K)\colon K \in \mathfrak{K}\}$.
Because $0 < \eta = \slack_{b}(\overline{K}) = b\cdot c(\Drop_{\vec{F}_0}(\overline{K})\cap \vec{F}) - c(\overline{K})$, we have that $\overline{K}$ is an $\alpha$-thin component fulfilling \eqref{eq:upperBoundComp}, as claimed.
(Note that $\eta > 0$ implies $c(\Drop_{\vec{F}_0}(\overline{K})\cap \vec{F})>0$, which also implies $\Drop_{\vec{F}_0}(\overline{K})\cap \vec{F} \neq\emptyset$.)

We claim that $\overline{K}$ minimizes $\sfrac{c(K)}{c(\Drop_{\vec{F}_0}(K)\cap \vec{F})}$ among all $\alpha$-thin components, as desired.
Suppose by the sake of deriving a contradiction that this is not the case, i.e., $\sfrac{c(\overline{K})}{c(\Drop_{\vec{F}_0}(\overline{K})\cap \vec{F})} > \rho^*$, and let $K^*\in \mathfrak{K}$ be such that $\sfrac{c(K^*)}{c(\Drop_{\vec{F}_0}(K^*)\cap \vec{F})} = \rho^*$.
The following contradiction now follows from integrality of $c$:
\begin{align*}
\frac{1}{c(\vec{F})^2}
  &> b-a \\
  &\geq \frac{c(\overline{K})}{c(\Drop_{\vec{F}_0}(\overline{K})\cap \vec{F})} - \frac{c(K^{*})}{c(\Drop_{\vec{F}_0}(K^*)\cap \vec{F})} \\
  &\geq \frac{1}{c(\Drop_{\vec{F}_0}(\overline{K})\cap \vec{F})\cdot c(\Drop_{\vec{F}_0}(K^{*})\cap \vec{F})} \\
  &\geq \frac{1}{c(\vec{F})^2},
\end{align*}
where the penultimate inequality follows from integrality of $c$.
Hence, we have $\sfrac{c(\overline{K})}{c(\Drop_{\vec{F}_0}(\overline{K})\cap \vec{F})}=\rho^*$ as desired.
\end{proof}

Instead of binary search, as used above, one can also employ \citeauthor{megiddo_1979_combinatorial}'s~\cite{megiddo_1979_combinatorial} parametric search technique to determine the value of $\rho^*$ in strongly polynomial time.
This works because the dynamic program we use to obtain a maximizer of $\max\{\slack_{\rho}(K)\colon K \in \mathfrak{K}\}$ as claimed in \cref{lem:slack_maximization} is a strongly polynomial time algorithm, where the value of $\rho$ appears linearly in each comparison that we perform during the algorithm.

\subsection{Maximizing Slack Through a Dynamic Program}\label{sec:maxSlackThroughDP}

We now prove \cref{lem:slack_maximization} by presenting a dynamic program to efficiently maximize the slack $\slack_{\rho}(K)$ among all $\alpha$-thin components $K\in \mathfrak{K}$.
In fact, we prove the following slightly stronger version of \cref{lem:slack_maximization}, which will also be useful for the local search procedure we discuss in \cref{sec:local-search}.

\begin{theorem}\label{thm:dp-main}
Given a rooted WRAP instance $(G,L,c,r,e_r)$  with a non-shortenable directed solution $\vec{F}_0 \subseteq L$ and a cost function $\tilde c \colon \vec{F}_0 \to \mathbb{R}_{\ge 0}$, we can compute in polynomial time a maximizer of
\begin{equation}\label{eq:thm-dp-objective}
   \max \Big\{ \tilde c(\Drop_{\vec{F}_0}(K))  -  c(K)  : K \in  \mathfrak{K} \Big\}.
\end{equation}
\end{theorem}

\cref{lem:slack_maximization} follows from \cref{thm:dp-main} by setting $\tilde c(\vec{\link{l}}) \coloneqq \rho \cdot c(\vec{\link{l}})$ for $\vec{\link{l}} \in \vec{F}\cap \vec{F}_0$ and $\tilde c(\vec{\link{l}}) \coloneqq 0$ for $\vec{\link{l}} \in \vec{F}_0 \setminus \vec{F}$. 
To prove \cref{thm:dp-main}, we fix, within this subsection, a rooted WRAP instance $(G,L,c,r,e_r)$ with a non-shortenable directed solution $\vec{F}_0 \subseteq L$ and a cost function $\tilde c \colon \vec{F}_0 \to \mathbb{R}_{\ge 0}$.

We start by providing some intuition before formally introducing our dynamic programming approach.
Recall that the notion of an $\alpha$-thin component $K$ (\cref{def:alpha-thin}) is defined through the existence of a maximal laminar family $\mathcal{L} \subseteq \mathcal{C}_G$ of $2$-cuts such that each of the $2$-cuts in $\mathcal{L}$ is crossed by at most $\alpha$ many links of $K$.
Our dynamic program (DP) aims at constructing an optimal solution $K^*$ bottom-up with respect to the unknown laminar family $\mathcal{L}$.
More precisely, for a $2$-cut $C\in \mathcal{C}_G$ (think of $C$ as being part of the laminar family $\mathcal{L}$) and a set $B\subseteq\delta_L(C)$ with $|B|\leq \alpha$, our DP aims at constructing ``best possible'' $\alpha$-thin link sets $S\subseteq L$ such that
\begin{enumerate}
\item each $\link{l}\in S$ has at least one endpoint in $C$, and
\item $\delta_S(C) = B$.
\end{enumerate}
In the propagation step of the DP, we then successively combine link sets $S_1,S_2\in L$, for pairs of $2$-cuts $C_1,C_2\in \mathcal{C}_G$ whose vertices are neighboring on the ring, to a link set for the $2$-cut $C_1\cup C_2$.\footnote{Formally, $C_1,C_2\in \mathcal{C}_G$ are \emph{neighboring} if $C_1\cap C_2=\emptyset$ and there is a pair of vertices $v_1\in C_1, v_2\in C_2$ connected by an edge.
Equivalently, $C_1\cap C_2=\emptyset$ and $C_1\cup C_2\in \mathcal{C}_G$.
}
When using a link set $S$ in such a propagation step to construct good solutions, we consider further parameters on top of just the crossing links $B=\delta_S(C)$.
More precisely, we will save how much in terms of objective value we gained already by $S$ (we define this formally in a moment), but also what kind of directed links in $\vec{F}_0$ with heads in $C$ could potentially be dropped later depending on how we extend the link set $S$.

Before discussing more formally the key parameters that we consider in the DP, it is helpful to better understand the structure of the droppable links $\Drop_{\vec{F}_0}(S)$ for a link set $S\subseteq L$.
To this end, we need the notion of \emph{least common ancestor} $\lca(U)$ of a vertex set $U\subseteq V$ (with respect to the arborescence $\vec{F}_0$), which, we recall, is the vertex $v\in V$ lowest in $\vec{F}_0$, i.e., farthest away from the root $r$ in $\vec{F}_0$, such that all vertices of $U$ are descendants of $v$.
See \cref{fig:least_common_ancestor}.

\begin{figure}[H]
\begin{center}
\begin{tikzpicture}[scale=1,
ns/.style={thick,draw=black,fill=white,circle,minimum size=6,inner sep=2pt},
es/.style={thick},
lks/.style={line width=1pt, blue, densely dashed},
dlks/.style={lks, -latex},
ts/.style={every node/.append style={font=\scriptsize}}
]

\def\num{13}
\def\rad{2.5}
\def\vlfac{1.14}

\pgfmathsetmacro\picsep{2*\rad + 1.7}

\pgfkeyssetvalue{/tikz/pics/interval/color}{green!50!black}
\pgfkeyssetvalue{/tikz/pics/interval/radius}{0.15*\rad}
\tikzset{
    pics/interval/.style 2 args={
        code={
               \def\cw{\pgfkeysvalueof{/tikz/pics/interval/radius}}
               \colorlet{col}{\pgfkeysvalueof{/tikz/pics/interval/color}}

               \draw[col, fill=col, fill opacity=0.2] (#1:\rad+\cw) arc (#1:#2:\rad+\cw)
               arc (#2:#2+180:\cw)
               arc (#2:#1:\rad-\cw)
               arc (#1+180:#1+360:\cw);
        }
    },
}
\tikzset{
    pics/intervalN/.style 2 args={
        code={
               \def\cw{\pgfkeysvalueof{/tikz/pics/interval/radius}}
               \colorlet{col}{\pgfkeysvalueof{/tikz/pics/interval/color}}

               \pgfmathanglebetweenpoints{\pgfpoint{0cm}{0cm}}{\pgfpointanchor{#1}{center}}
               \edef\angA{\pgfmathresult}
               \pgfmathanglebetweenpoints{\pgfpoint{0cm}{0cm}}{\pgfpointanchor{#2}{center}}
               \edef\angB{\pgfmathresult}

               \draw[col, fill=col, fill opacity=0.2] (\angA:\rad+\cw) arc (\angA:\angB:\rad+\cw)
               arc (\angB:\angB+180:\cw)
               arc (\angB:\angA:\rad-\cw)
               arc (\angA+180:\angA+360:\cw);
        }
    },
}

\coordinate (c) at (0,0);

\begin{scope}[every node/.style={ns}]
\foreach \i in {1,...,\num} {
  \pgfmathsetmacro\r{90+(\i-1)*360/\num}
  \node (\i) at (\r:\rad) {};
}
\end{scope}

\begin{scope}[es]
\foreach \i in {1,...,\num} {
\pgfmathtruncatemacro\j{1+mod(\i,\num)}
\draw (\i) -- (\j);
}
\end{scope}

\begin{scope}
\node at ($(c)!\vlfac!(1)$) {$r$};
\end{scope}

\begin{scope}[dlks]
\draw (1) to[bend left] (2);
\draw (2) to[bend left=10] (6);
\draw (6) to[bend right, in=-120, out=-10] (4);
\draw (4) to[bend right] (3);
\draw (4) to[bend left] (5);
\draw (6) to[bend left, in=160, out=10] (12);
\draw (12) to[bend left] (13);
\draw (12) to[bend right] (11);
\draw (11) to[bend right, out=-10] (8);
\draw (8) to[bend right] (7);
\draw (8) to[bend left] (9);
\draw (9) to[bend left] (10);
\end{scope}

\begin{scope}
\path (c) pic {intervalN={3}{4}};
\path (c) pic {intervalN={7}{9}};
\path (c) pic {intervalN={11}{12}};
\end{scope}

\begin{scope}[red!50!black]
\node[ns,fill=red!50!black] at (6) {};
\node at (6)[below left] {$\mathrm{lca}(U)$};
\end{scope}

\begin{scope}
\node[green!50!black] at (50:\rad+0.7) {$U$};
\end{scope}

\end{tikzpicture}
 \end{center}
\caption{Example of the least common ancestor $\lca(U)$ of a vertex set $U\subseteq V$.}\label{fig:least_common_ancestor}
\end{figure}
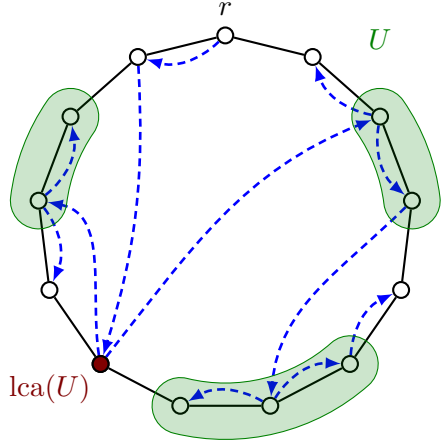

We now recall \cref{lem:ov_dropOfConnectedS}, which provides a clean description of which links of $\vec{F}_0$ are droppable with respect to a link set $S\subseteq L$ that is connected in the intersection graph.
More precisely, all directed links of $\vec{F}_0$ whose heads are endpoints of links in $S$ are droppable, except for the directed link entering the least common ancestor $\lca(V(S))$ of the endpoints $V(S)$ of $S$.
\AllExceptOneDroppable*
\begin{proof}
We recall that, by \cref{lem:characterize_drop} and the fact that $H[S]$ is connected, a directed link $(s,t)\in\vec{F}$ is part of $\Drop_{\vec{F}}(S)$ if and only if $t\in V(S)$ and $V(S)$ contains a vertex that is $t$-good.
Hence, in particular, we have $\Drop_{\vec{F}}(S) \subseteq \bigcup_{v\in V(S)} \delta^-_{\vec{F}}(v)$.

Consider a link $(s,t)\in \bigcup_{v\in V(S)} \delta^-_{\vec{F}}(v)$; hence, $t\in V(S)$.
We have that $(s,t)\not\in \Drop_{\vec{F}}(S)$ if and only if all vertices in $V(S)$ are $t$-bad, i.e., all vertices in $V(S)$ are descendants of $t$ in $\vec{F}$.
However, as $t\in V(S)$, this is equivalent to $t=\lca(V(S))$, as desired.
\end{proof}
That only the single link $\vec{\link{l}}$ in $\delta^-_{\vec{F}_0}(\lca(V(S)))$ is not droppable is a key fact that we exploit in our dynamic program, as it heavily limits the number of configuration for which we need to do bookkeeping.

Before expanding on the table entries for our DP, we briefly highlight a property of the least common ancestor $\lca(U)$ for $U\subseteq V$, which is simple and helpful for intuition, and we will use it later on.
More precisely, $\lca(U)$ depends only on the leftmost and rightmost vertex of $U$, and it always lies between them.
\begin{lemma}\label{lem:lca_from_leftmost_rightmost}
Let $(G=(V,E),L,c,r,e_r)$ be a rooted WRAP instance with non-shortenable directed solution $\vec{F}\subseteq \shadows(L)$, and let $U\subseteq V$. Then $\lca(U)$ is the least common ancestor of the leftmost vertex $s$ and rightmost vertex $t$ in $U$ and lies between $s$ and $t$.
\end{lemma}
\begin{proof}
The fact that $\lca(\{s,t\})$ lies between $s$ and $t$ is an immediate consequence of \cref{lem:furtherNonshortProps}~\ref{item:lca_is_between}.
Hence, it remains to show that $\lca(U)=\lca(\{s,t\})$.
This holds because the set of descendants of $\lca(\{s,t\})$ contains $s$ and $t$, and is an interval by \cref{lem:furtherNonshortProps}~\ref{item:descendants_interval}.
Hence, all vertices of $U$ are descendants of $\lca(\{s,t\})$, because $s$ and $t$ are the leftmost and rightmost vertex of $U$, respectively.
Together with the fact that $s,t\in U$, we have $\lca(\{s,t\})=\lca(U)$, as desired.
\end{proof}

We now formalize the precise parameters for which we create a DP table entry.
We recall that a DP table entry captures key properties of how an $\alpha$-thin link set $S\subseteq L$ interacts with a $2$-cut $C\in \mathcal{C}_G$ in a laminar family certifying $\alpha$-thinness of $S$ (as described in \cref{def:alpha-thin}).
More precisely, we are interested in describing the relevant properties of the links of $S$ with at least one endpoint in $C$, i.e., the link set
\begin{equation*}
S_C\coloneqq \{\{u,v\}\in S \colon \{u,v\}\cap C \neq \emptyset\}.
\end{equation*}
Moreover, when considering a $2$-cut $C\in \mathcal{C}_G$ with a link set $S_C$, we do not only want $S_C$ to be $\alpha$-thin, but that there is a maximal laminar family $\mathcal{L}$ with $C\in \mathcal{L}$ that certifies $\alpha$-thinness of $S_C$, because we want the DP to consider the set $C$ as being part of an inclusion-wise maximal laminar family certifying $\alpha$-thinness of $S$.
As the DP will consider at each step only a partial solution, we use the following notion of $(\alpha,C)$-thinness to captures $\alpha$-thinness on the vertices considered so far.
\begin{definition}[$(\alpha,C)$-thin link set]\label{def:alphaC-thin}
Let $(G=(V,E),L,c,r,e_r)$ be a rooted WRAP instance, $C\in \mathcal{C}_G$, and $\alpha\in \mathbb{Z}_{\geq 1}$.
A set $S\subseteq \bigcup_{v\in C}\delta_L(v)$ is called $(\alpha,C)$-thin if there is an inclusion-wise maximal laminar family $\mathcal{L}\subseteq \mathcal{C}_G$ over the ground set $C$ such that $|\delta_S(\overline{C})|\leq \alpha$ for all $\overline{C}\in \mathcal{L}$.
\end{definition}
Note that $(\alpha,V\setminus \{r\})$-thinness is the same as $\alpha$-thinness.

The parameters of $S_C$ that we use as dynamic programming table entries are described by a quintuple $(C,B,\mathcal{T},\phi,\psi)$, where
\begin{itemize}
\item $C\in \mathcal{C}_G$;
\item $B\subseteq \delta_L(C)$ with $|B|\leq \alpha$;
\item $\mathcal{T}\subseteq 2^B$ is a partition of $B$;
\item $\phi \colon \mathcal{T}\to C$; and
\item $\psi\colon \mathcal{T} \to \{0,1\}$,
\end{itemize}
with the following interpretation, highlighted in \cref{def:realizingLinkSet}, where we denote by $\mathcal{Q}$ the family of all quintuples of the above type, which we also call \emph{patterns}.
(Note that the number $|\mathcal{Q}|$ of patterns is polynomially bounded.)
\begin{definition}[Link set realizing $(C,B,\mathcal{T},\phi,\psi)$]\label{def:realizingLinkSet}
Let $(C,B,\mathcal{T},\phi,\psi)\in \mathcal{Q}$ and $S\subseteq L$ be $(\alpha,C)$-thin.
We denote by $S_1,\ldots, S_q\subseteq S$ the partition of $S$ according to the connected components of $H[S]$.
We say tat $S$ \emph{realizes} the pattern $(C,B,\mathcal{T},\phi,\psi)$ if all of the following conditions hold:
\begin{enumerate}
\item\label{item:realizerOneEndpointInC} All links in $S$ have at least one endpoint in $C$.
\item\label{item:realizerBOk} $B=\delta_S(C)$.
\item\label{item:realizerTOk} $\mathcal{T}$ is the family of all nonempty sets among $S_1\cap B, \ldots, S_q\cap B$.
\item\label{item:realizerPhiOk} For any set $S_i\cap B \in \mathcal{T}$, we have $\phi(S_i\cap B)=\lca(V(S_i))$.
\item\label{item:realizerPsiOk} For any set $S_i\cap B \in \mathcal{T}$, we have $\psi(S_i\cap B)=1$ if $\phi(S_i\cap B)\in V(S_i)$; otherwise $\psi(S_i\cap B)=0$.
\end{enumerate}
\end{definition}
See \cref{fig:patternRealization} for an example of a link set realizing a pattern $Q\in \mathcal{Q}$.
Note that only $(\alpha,C)$-thin link sets can realize a pattern in $(C,B,\mathcal{T},\phi,\psi)\in \mathcal{Q}$.
We call a pattern $Q\in \mathcal{Q}$ \emph{realizable} if there is a link set that realizes $Q$.
Hence, one can think of realizable patterns as feasible patterns.
Because the $2$-cut $C\in \mathcal{C}_G$ plays a crucial role in the way we consider patterns $Q=(C,B,\mathcal{T},\phi,\psi)$ in our dynamic program, we call a \emph{$C$-pattern}, for $C\in \mathcal{C}_G$, a pattern in $\mathcal{Q}$ whose first entry in its quintuple is $C$.

\begin{figure}[!ht]
\begin{center}
\begin{tikzpicture}[scale=1,
ns/.style={thick,draw=black,fill=white,circle,minimum size=6,inner sep=2pt},
es/.style={thick},
lks/.style={line width=1pt, blue, densely dashed},
dlks/.style={lks, -latex},
ts/.style={every node/.append style={font=\scriptsize}}
]

\def\num{50}
\def\rad{5.8}
\def\vlfac{1.17}

\colorlet{s1c}{darkgreen!100!white}
\colorlet{s2c}{red!70!white}
\colorlet{s3c}{cyan!80!black}
\colorlet{s4c}{green!80!black!80!blue}
\colorlet{s5c}{orange}

\colorlet{Ccol}{violet!60!white}

\pgfmathsetmacro\picsep{2*\rad + 1.7}

\pgfkeyssetvalue{/tikz/pics/interval/color}{green!50!black}
\pgfkeyssetvalue{/tikz/pics/interval/radius}{0.15*\rad}
\tikzset{
    pics/interval/.style 2 args={
        code={
               \def\cw{\pgfkeysvalueof{/tikz/pics/interval/radius}}
               \colorlet{col}{\pgfkeysvalueof{/tikz/pics/interval/color}}

               \draw[col, fill=col, fill opacity=0.2] (#1:\rad+\cw) arc (#1:#2:\rad+\cw)
               arc (#2:#2+180:\cw)
               arc (#2:#1:\rad-\cw)
               arc (#1+180:#1+360:\cw);
        }
    },
}
\tikzset{
    pics/intervalN/.style 2 args={
        code={
               \def\cw{\pgfkeysvalueof{/tikz/pics/interval/radius}}
               \colorlet{col}{\pgfkeysvalueof{/tikz/pics/interval/color}}

               \pgfmathanglebetweenpoints{\pgfpoint{0cm}{0cm}}{\pgfpointanchor{#1}{center}}
               \edef\angA{\pgfmathresult}
               \pgfmathanglebetweenpoints{\pgfpoint{0cm}{0cm}}{\pgfpointanchor{#2}{center}}
               \edef\angB{\pgfmathresult}

               \tikzmath{
                   if \angB < \angA then {
                       \angB = \angB + 360;
                   };
               }

               \draw[col, fill=col, fill opacity=0.2] (\angA:\rad+\cw) arc (\angA:\angB:\rad+\cw)
               arc (\angB:\angB+180:\cw)
               arc (\angB:\angA:\rad-\cw)
               arc (\angA+180:\angA+360:\cw);
        }
    },
}

\coordinate (c) at (0,0);

\begin{scope}[every node/.style={ns}]
\foreach \i in {1,...,\num} {
  \pgfmathsetmacro\r{90+(\i-1)*(360/\num)}
  \node (\i) at (\r:\rad) {};
}
\end{scope}

\begin{scope}[es]
\foreach \i in {1,...,\num} {
\pgfmathtruncatemacro\j{1+mod(\i,\num)}
\draw (\i) -- (\j);
}
\end{scope}

\begin{scope}
\end{scope}

\begin{scope}[dlks,relative,opacity=0.5]

\tikzset{
r1/.style={out=45,in=130},
r2/.style={out=45,in=110},
r3/.style={out=45,in=110},
r4/.style={out=45,in=110},
r5/.style={out=55,in=110},
r6/.style={out=55,in=110},
r8/.style={out=35,in=120},
r11/.style={out=35,in=140},
r15/.style={out=35,in=150},
l1/.style={out=-45,in=-130},
l2/.style={out=-45,in=-110},
l3/.style={out=-45,in=-110},
l4/.style={out=-45,in=-110},
l5/.style={out=-45,in=-110},
l6/.style={out=-45,in=-110},
}

\draw (1) to[r11] (12);
\draw (12) to[r15] (27);

\draw (27) to[l6] (21);
\draw (21) to[r4] (25);
\draw (25) to[l2] (23);
\draw (23) to[r1] (24);
\draw (23) to[l1] (22);
\draw (25) to[r1] (26);

\draw (21) to[l5] (16);
\draw (16) to[l2] (14);
\draw (14) to[r1] (15);
\draw (14) to[l1] (13);

\draw (16) to[r2] (18);
\draw (18) to[l1] (17);
\draw (18) to[r1] (19);
\draw (19) to[r1] (20);

\draw (27) to[r4] (31);
\draw (31) to[l2] (29);
\draw (29) to[r1] (30);
\draw (29) to[l1] (28);

\draw (31) to[r6] (37);
\draw (37) to[l5] (32);
\draw (32) to[r3] (35);
\draw (35) to[r1] (36);
\draw (35) to[l2] (33);
\draw (33) to[r1] (34);

\draw (37) to[r8] (45);
\draw (45) to[l4] (41);
\draw (41) to[r2] (43);
\draw (43) to[r1] (44);
\draw (43) to[l1] (42);

\draw (41) to[l2] (39);
\draw (39) to[r1] (40);
\draw (39) to[l1] (38);

\draw (12) to[l6] (6);
\draw (6) to[l2] (4);
\draw (4) to[l1] (3);
\draw (3) to[l1] (2);
\draw (4) to[r1] (5);

\draw (6) to[r1] (7);
\draw (7) to[r1] (8);
\draw (8) to[r2] (10);
\draw (10) to[r1] (11);
\draw (10) to[l1] (9);

\draw (45) to[r5] (50);
\draw (50) to[l4] (46);
\draw (46) to[r2] (48);
\draw (48) to[r1] (49);
\draw (48) to[l1] (47);

\end{scope}

\begin{scope}[lks,solid,line width=2pt,relative]
\begin{scope}[s1c]
\draw (13) to[out=55,in=125] node[right] {$S_1$} (19);
\draw (15) to[out=55,in=135] (18);
\draw (17) to[out=55,in=135] (22);
\end{scope}

\begin{scope}[s2c]
\draw (35) to[out=55,in=125] (37);
\draw (36) to[out=55,in=125] (39);
\draw (38) to[out=55,in=125] (47);
\draw (40) to[out=55,in=125] node[left] {$S_2$} (49);
\end{scope}

\begin{scope}[s3c]
\draw (6) to[out=40,in=155] node[right] {$S_3$} (24);
\draw (23) to[out=40,in=145] (28);
\draw (24) to[out=40,in=145] (27);
\draw (25) to[out=45,in=135] (34);
\end{scope}

\begin{scope}[s4c]
\draw (41) to[out=55,in=115] node[below left=0.18em,pos=0.40] {$S_4$} (44);
\draw (43) to[out=55,in=145] (45);
\end{scope}

\begin{scope}[s5c]
\draw (29) to[out=45,in=125] node[pos=0.55,above] {$S_5$} (32);
\draw (31) to[out=55,in=145] (33);
\end{scope}

\end{scope}

\begin{scope}[every node/.style={ns}]
\node[fill=s1c] (ps1) at (21) {};
\node[fill=s2c] (ps2) at (37) {};
\node[fill=s3c] (ps3) at (12) {};
\node[fill=s4c] (ps4) at (45) {};
\end{scope}

\begin{scope}

\begin{scope}[s1c]
\node[align=center] (l1) at ($(ps1)+(-0.13*\rad,-0.17*\rad)$) {$\phi(S_1\cap B)$\\[0.3ex]$\psi(S_1\cap B)\!=\!0$};
\end{scope}

\begin{scope}[s2c]
\node[align=center] (l2) at ($(ps2)+(0.19*\rad,-0.21*\rad)$) {$\phi(S_2\cap B)$\\[0.3ex]$\psi(S_2\cap B)\!=\!1$};
\end{scope}

\begin{scope}[s3c]
\node[align=center] (l3) at ($(ps3)+(-0.23*\rad,-0.17*\rad)$) {$\phi(S_3\cap B)$\\[0.3ex]$\psi(S_3\cap B)\!=\!0$};
\end{scope}

\begin{scope}[s4c]
\node[align=center] (l4) at ($(ps4)+(0.38*\rad,0.00*\rad)$) {$\phi(S_4\cap B)$\\[0.3ex]$\psi(S_4\cap B)\!=\!1$};
\end{scope}

\begin{scope}[-latex, shorten >=1pt]
\draw[s1c] (l1.north) to[out=55, in=180] (ps1);
\draw[s2c] (l2.north) to[out=115, in=0] (ps2);
\draw[s3c] (l3.north) to[out=45, in=180] (ps3);
\draw[s4c] (l4.north) to[out=135, in=45] (ps4);
\end{scope}

\end{scope}

\begin{scope}
\node at (1)[above=2pt] {$r$};
\end{scope}

\begin{scope}
\def\irad{0.06*\rad}

\begin{scope}
\path (c) pic[pics/interval/radius=\irad,
              pics/interval/color=Ccol] {intervalN={18}{44}};
\end{scope}

\begin{scope}[Ccol]
\node at (-45:\rad cm + \irad + 4ex) {$C$};
\end{scope}

\end{scope}

\end{tikzpicture}
 \end{center}
\caption{Example of a link set $S=S_1\cup S_2 \cup S_3 \cup S_4 \cup S_5$ realizing a pattern $Q=(C,B,\mathcal{T},\phi,\psi)$.
The link set $S$ is shown as solid lines and consists of five connected components in $H[S]$, which are denoted by $S_1$ (shown in dark green), $S_2$ (red), $S_3$ (light blue), $S_4$ (light green), and $S_5$ (orange).
Only the first four of these link sets contain at least one link of $B$, i.e., a link crossing the $2$-cut $C$ (shown in violet).
Hence, the partition $\mathcal{T}$ of $S\cap B$ consists of four parts, namely $\mathcal{T}=\{S_1\cap B, S_2\cap B, S_3\cap B, S_4\cap B\}$.
We recall that $B$ are all links of $S$ that cross $C$.
The vertices $\phi(S_i\cap B)$ for $i\in [4]$ are colored with the color of the set $S_i$, and the corresponding $\psi$-values are written next to them.
}\label{fig:patternRealization}
\end{figure}
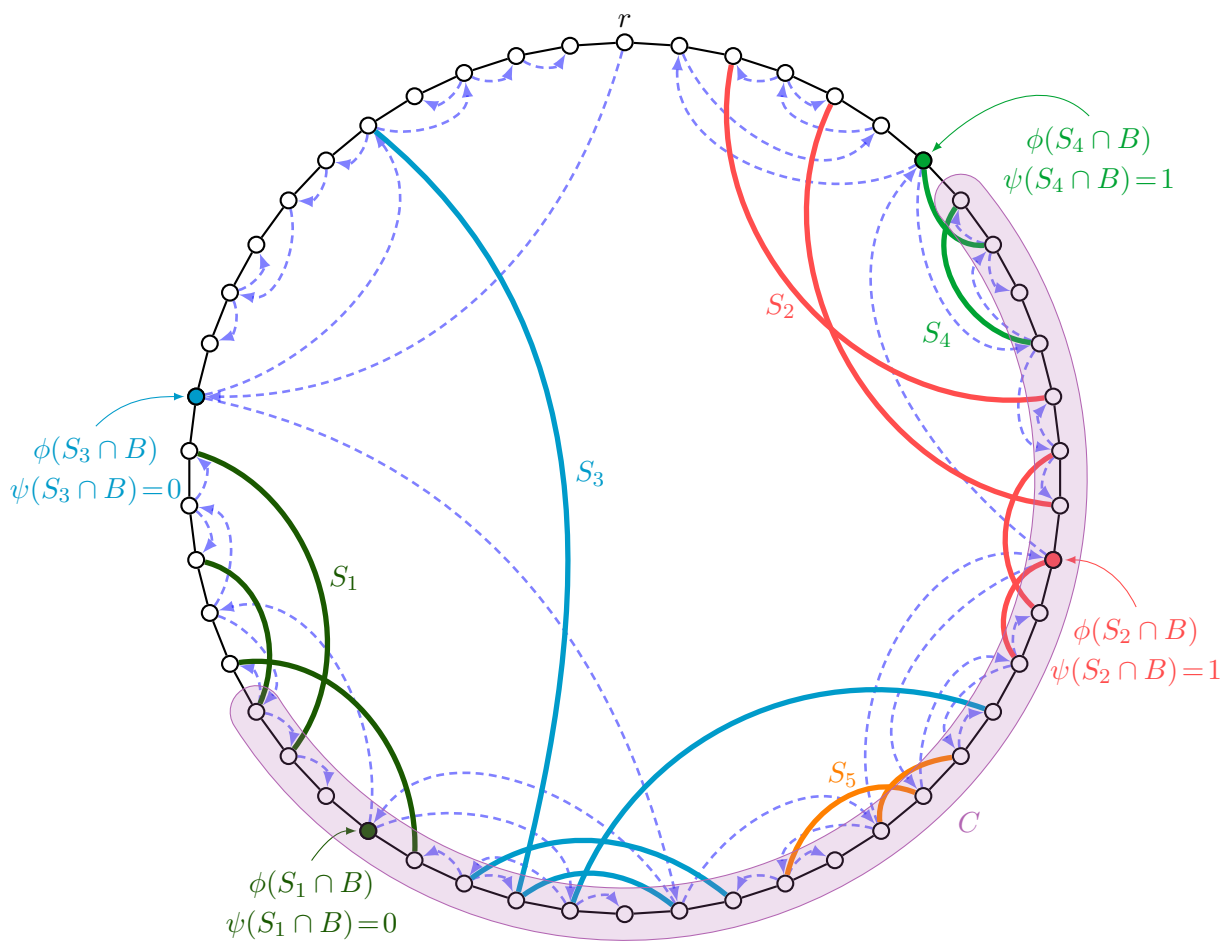

Our dynamic program constructs for each realizable pattern $Q=(C,B,\mathcal{T},\phi,\psi)$ an $(\alpha,C)$-thin link set $S_Q\subseteq L$ maximizing
\begin{equation}\label{eq:objDP}
\max\{\pi(S,C)\colon S\subseteq L \text{ is an $(\alpha,C)$-thin link set realizing $Q$}\}, \text{ where}
\end{equation}
\begin{equation*}
\pi(S,C) \coloneqq \tilde c\left(\Drop_{\vec{F}_0}(S)\cap \bigcup_{v\in C}\delta^-_{\vec{F}_0}(v)\right) - c(S).
\end{equation*}
We call the maximum value of $\pi(S,C)$ among all $(\alpha,C)$-thin link sets the \emph{optimal value of $Q$} and denote it by $\pi_Q^*$.
Hence, a table entry of our DP is a pattern $Q\in \mathcal{Q}$ and its value is $\pi_Q^*$, which is saved together with a maximizer $S_Q$ of~\eqref{eq:objDP}.
(For non-realizable patterns $Q$, for which we recall that $\pi_Q^*$ is not defined, we save the value $-\infty$ instead of $\pi_Q^*$.)
Moreover, we call a realizer $S_Q$ of $Q\in \mathcal{Q}$ that maximizes~\eqref{eq:objDP} a \emph{maximizing realizer of $Q$}, or simply a \emph{$Q$-maximizing} realizer.
Note that the function $\pi(S,C)$ captures the objective value obtained from $S$ when only counting droppable links with head in $C$.

Observe that if we can compute maximizers $S_Q$ of~\eqref{eq:objDP} for all patterns $Q\in \mathcal{Q}$, then we are done, i.e., we obtain an $\alpha$-thin component $K\in \mathfrak{K}$ attaining the maximum in \eqref{eq:thm-dp-objective}, as desired.
Indeed, it suffices to consider, among all components $Q=(C,B,\mathcal{T},\phi,\psi)$ with $C=V\setminus \{r\}$, the one with highest value $\pi^*_Q$, and return $S_Q$, which leads to an optimal component because, for any $S\subseteq L$, we have $\pi(S,V\setminus \{r\})=\tilde c(\Drop_{\vec{F_0}}(S)) - c(S)$.

We compute $\pi^*_Q$ together with a maximizer $S_Q$ of~\eqref{eq:objDP} first for all patterns $(C,B,\mathcal{T},\phi,\psi)$ with $|C|=1$ and then continue with respect to increasing cardinality of $C$.
Note that finding a maximizer of~\eqref{eq:objDP} for a pattern $Q=(C,B,\mathcal{T},\phi,\psi)\in \mathcal{Q}$ with $|C|=1$ can easily be done through enumeration, because, in this case, an $(\alpha,C)$-thin link set $S\subseteq L$ realizing $Q$ fulfills $|S|\leq \alpha = O(1)$.

For patterns $Q=(C,B,\mathcal{T},\phi,\psi)\in \mathcal{Q}$ with $|C|\geq 2$, we compute $\pi^*_Q$ and $S_Q$ by combining solutions of previously computed patterns.
To this end, we use that $S_Q$ must be $(\alpha,C)$-thin, which implies that $C$ can be partitioned into two $2$-cuts $C_1,C_2\in \mathcal{C}_G$ such that, for $i\in [2]$, all links of $S_Q$ with at least one endpoint in $C_i$ are $(\alpha,C_i)$-thin.
Hence, $S_Q$ is the union of an $(\alpha,C_1)$-thin link set realizing a $C_1$-pattern $Q_1$ and an $(\alpha,C_2)$-thin link set realizing a $C_2$-pattern $Q_2$.
Moreover, one can use maximizers $S_{Q_1}$ and $S_{Q_2}$ of~\eqref{eq:objDP} for $Q_1$ and $Q_2$, respectively.
We now formally define when two patterns $Q_1$ and $Q_2$ are compatible, in the sense that separate solutions for these patterns can be interpreted as parts of a larger solution for a larger pattern $Q$.
\begin{definition}[Compatibility of pattern pairs]\label{def:comp_patterns}
Two patterns $(C_1,B_1,\mathcal{T}_1,\phi_1,\psi_1), (C_2,B_2,\mathcal{T}_2,\allowbreak\phi_2,\psi_2) \in \mathcal{Q}$ are \emph{compatible} if
\begin{enumerate}
\item $C_1,C_2$ are neighboring $2$-cuts, 
\item $\delta_{B_1}(C_2) = \delta_{B_2}(C_1)$, and
\item $|\delta_{B_1\cup B_2}(C_1\cup C_2)|\leq \alpha$.
\end{enumerate}
\end{definition}

When taking the union of two links sets $S_1, S_2 \subseteq L$ realizing compatible patterns $Q_1=(C_1,B_1,\mathcal{T}_1,\phi_1,\psi_1)$ and $Q_2=(C_2,B_2,\mathcal{T}_2,\phi_2,\psi_2)$, respectively, we obtain a link set realizing a new pattern $Q$, which we call the \emph{merger} of $Q_1$ and $Q_2$.
As we will see later, this pattern $Q$ solely depends on $Q_1$ and $Q_2$ and not the link sets $S_1$ and $S_2$.
Before we formally define the notion of merger of $Q_1$ and $Q_2$, and show its independence of specific link sets realizing $Q_1$ and $Q_2$, we discuss how to derive in which connected components the links $B_1\cup B_2$ lie and what the least common ancestors are of those components when combining a $Q_1$-realizer and a $Q_2$-realizer for a compatible pair $Q_1$ and $Q_2$.
This provides the basis for formally defining the merger of $Q_1$ and $Q_2$.
To this end, we introduce the notion of \emph{combined partition}.
\begin{definition}[Combined partition, combined least common ancestor (lca) function]
\label{def:combPartition}
Let $Q_1 = (C_1,B_1,\allowbreak\mathcal{T}_1,\phi_1,\psi_1)$ and $Q_2 = (C_2,B_2,\mathcal{T}_2,\phi_2,\psi_2)$ be a pair of compatible patterns.
The \emph{combined partition} $\overline{\mathcal{T}}\subseteq 2^{B_1\cup B_2}$ of $B_1\cup B_2$ corresponds to the vertex sets of the connected components of the graph with vertices $B_1\cup B_2$ where, for distinct $\link{l},\link{f}\in B_1\cup B_2$, there is an edge between $\link{l}$ and $\link{f}$ if either
\begin{enumerate}[label=(\alph*)]
\item\label{item:lfIntersecting} $\link{l}$ intersects $\link{f}$,
\item\label{item:lfInB1} $\link{l},\link{f}\in B_1$ and $\link{l},\link{f}$ are in the same part of the partition $\mathcal{T}_1$, or
\item\label{item:lfInB2} $\link{l},\link{f}\in B_2$ and $\link{l},\link{f}$ are in the same part of the partition $\mathcal{T}_2$.
\end{enumerate}
Moreover, the \emph{combined least common ancestor (lca) function} $\overline{\phi}\colon \overline{\mathcal{T}} \to C_1\cup C_2$ is defined as follows.
For any $\overline{T}\in \overline{\mathcal{T}}$, which, by definition of $\overline{\mathcal{T}}$, can be written as 
\begin{equation}\label{eq:decomp_T_merger}
\overline{T} = \left(T^1_1 \cup \ldots \cup T^1_{q_1}\right) \cup \left(T^2_1\cup \ldots \cup T^2_{q_2}\right)
\end{equation}
with $T^1_1,\ldots, T^1_{q_1}\in \mathcal{T}_1$ and $T^2_1,\ldots, T^2_{q_2}\in \mathcal{T}_2$, we set
\begin{equation*}
\overline{\phi}(\overline{T}) \coloneqq \lca\left(\left\{\phi_1(T^1_1),\ldots, \phi_1(T^1_{q_1}), \phi_2(T^2_1),\ldots, \phi_2(T^2_{q_2})\right\}\right).
\end{equation*}
\end{definition}
Note that, because $\mathcal{T}_1$ and $\mathcal{T}_2$ partition $B_1$ and $B_2$, respectively, there is a unique way to write $\overline{T}$ as a union of sets $T^i_j$ as shown in~\eqref{eq:decomp_T_merger}.

The following statement shows that the combined partition indeed corresponds to the partition of $B_1\cup B_2$ with respect to the connected components of $H[S_1\cup S_2]$ for $S_1$ being a $Q_1$-realizer and $S_2$ being a $Q_2$-realizer set.
As the statement of \cref{lem:combPartitionOk} holds for any realizing sets $S_1$ and $S_2$, and the combined partition $\overline{\mathcal{T}}$ only depends on $Q_1$ and $Q_2$, this shows that the partition of $B_1\cup B_2$ with respect to the connected components of $H[S_1\cup S_2]$ is independent of what $Q_1$-realizer $S_1$ and $Q_2$-realizer $S_2$ is chosen.
Furthermore, the \lcnamecref{lem:combPartitionOk} below also shows that the combined lca function $\overline{\phi}$ correctly returns the least common ancestor of any connected component of $H[S_1\cup S_2]$ that contains at least one link of $B_1\cup B_2$. 
\begin{lemma}\label{lem:combPartitionOk}
Let $Q_1=(C_1,B_1,\mathcal{T}_1,\phi_1,\psi_1)$ and $Q_2=(C_2,B_2,\mathcal{T}_2,\phi_2,\psi_2)$ be a pair of compatible and realizable patterns with combined partition $\overline{\mathcal{T}}$ and combined lca function $\overline{\phi}$.
Moreover, for $i\in [2]$, let $S_i\subseteq L$ be an $(\alpha,C_i)$-thin link set realizing $Q_i$, and let $M_1,\dots, M_q\subseteq S_1\cup S_2$ be the links corresponding to the connected components of $H[S_1\cup S_2]$.

Then, $\overline{\mathcal{T}}$ consists of all nonempty sets among $M_1\cap (B_1\cup B_2), \dots, M_q \cap (B_1\cup B_2)$.
Furthermore, for any nonempty $\overline{T}=M_i\cap (B_1\cup B_2)$, we have $\overline{\phi}(\overline{T}) = \lca(V(M_i))$.
\end{lemma}
Before showing \cref{lem:combPartitionOk}, we prove an important auxiliary result needed in its proof.
More precisely, in the proof of~\cref{lem:combPartitionOk}, we partition link sets $M\subseteq S_1\cup S_2$ corresponding to a connected component of $H[S_1\cup S_2]$ with $M\cap (B_1\cup B_2) \neq \emptyset$ into connected components of the graphs $H[S_1]$ and $H[S_2]$.
The following statement shows that such a connected component $S^i_j$ of $H[S_i]$ has a non-empty intersection with $B_i$, which guarantees $S_i^j\cap B_i\in \mathcal{T}_i$, and will allow us to create the desired relation to the combined lca function $\overline{\phi}$.
\begin{lemma}\label{lem:MDecomposesInTParts}
Let $Q_1=(C_1,B_1,\mathcal{T}_1,\phi_1,\psi_1)$ and $Q_2=(C_2,B_2,\mathcal{T}_2,\phi_2,\psi_2)$ be a pair of compatible patterns and, for $i\in [2]$, let $S_i \subseteq L$ be an $(\alpha,C_i)$-thin link set realizing $Q_i$.
Moreover, let $M\subseteq S_1\cup S_2$ be a link set corresponding to a connected component of $H[S_1\cup S_2]$ satisfying $M\cap (B_1\cup B_2)\neq\emptyset$.
Hence, $M$ can be written as
\begin{equation*}
M = \left(S^1_1 \cup \ldots \cup S^1_{q_1}\right) \cup \left(S^2_1\cup \ldots \cup S^2_{q_2}\right),
\end{equation*}
where, for $i\in [2]$, the sets $S^i_1,\dots, S^i_{q_i}\subseteq S_i$ are link sets corresponding to connected components of $H[S_i]$.
Then
\begin{equation*}
S^i_j\cap B_i \neq \emptyset \qquad \forall i\in[2], j\in[q_i].
\end{equation*}
\end{lemma}
\begin{proof}
Because $M\cap (B_1\cup B_2) \neq \emptyset$, every link set $S^i_j$ for $i\in [2]$ and $j\in [q_i]$ either contains a link in $B_i$ or must be connected to another set $S^{i'}_{j'}$.
Such a set $S^{i'}_{j'}$ must satisfy $i'\neq i$ because $S^i_j$ is the link set of a connected component of $H[S_i]$.
The only way how $S^i_j$ can be connected to $S^{i'}_{j'}$ in $H[S_1\cup S_2]$ is through a link $\link{l}\in S_{i'}$ with one endpoint in $C_1$ and one in $C_2$.
Because the patterns $Q_1$ and $Q_2$ are compatible, and $S_1$ and $S_2$, respectively, are realizers for these patterns, we have that the links in $S_1$ with one endpoint in $C_1$ and one in $C_2$ are the same as the links in $S_2$ with one endpoint in $C_1$ and one in $C_2$.
Thus, $\link{l}\in B_1\cap B_2$, which implies that $\link{l}\in S^i_j$, and hence $S^i_j\cap B_i \neq \emptyset$.
\end{proof}
We are now ready to prove \cref{lem:combPartitionOk}.
\begin{proof}[Proof of \cref{lem:combPartitionOk}]
Let $\link{h}_s,\link{h}_t\in B_1\cup B_2$.
To prove the first part of the \lcnamecref{lem:combPartitionOk} we have to show that the following two statements are equivalent.
\begin{enumerate}
\item\label{item:lfInSamePartOfOvT} $\link{h}_s$ and $\link{h}_t$ are in the same part of the partition $\overline{\mathcal{T}}$.
\item\label{item:lfInSameConnCompH} $\link{h}_s$ and $\link{h}_t$ are in the same connected component of $H[S_1\cup S_2]$.
\end{enumerate}
First observe that \ref{item:lfInSamePartOfOvT} clearly implies \ref{item:lfInSameConnCompH}, because if two links $\link{l},\link{h}$ are connected by an edge in the graph defined in \cref{def:comp_patterns}, then they are in the same connected component of $H[S_1\cup S_2]$.

Hence, it remains to show that \ref{item:lfInSameConnCompH} implies \ref{item:lfInSamePartOfOvT}.
To this end, we show that if, for a pair of links $\link{h}_s,\link{h}_t\in B_1\cup B_2$, there is an $\link{h}_s$-$\link{h}_t$ path in $H[S_1\cup S_2]$ consisting only of links in $(S_1\cup S_2)\setminus (B_1\cup B_2)$, except for its endpoints, then $\link{h}_s$ and $\link{h}_t$ lie in the same part of the partition $\overline{\mathcal{T}}$.
Note that this shows the desired implication, because if two arbitrary links $\link{h}_s,\link{h}_t\in B_1\cup B_2$ are connected by a path in $H[S_1\cup S_2]$, then this path can be split into consecutive subpaths between links in $B_1\cup B_2$ with each subpath not containing any link of $B_1\cup B_2$ in its interior.
Thus, let $\link{h}_s,\link{h}_t\in B_1\cup B_2$ such that there is an $\link{h}_s$-$\link{h}_t$ path $P$ in $H[S_1\cup S_2]$---where we denote by $\link{h}_s \eqqcolon \link{l}_0, \link{l}_1,\dots,\link{l}_{p+1}\coloneqq \link{h}_t$ the traversed links---such that no link in the interior of $P$ is in $B_1\cup B_2$, i.e, $\link{l}_1,\dots, \link{l}_{p}\not\in B_1\cup B_2$.
As $P$ is a path in $H[S_1\cup S_2]$, we have that $\link{l}_i$ and $\link{l}_{i+1}$ are intersecting for $i\in \{0,\dots,p\}$.

Consider first the case of $P$ not containing any links in its interior, i.e., $p=0$.
Hence, $\link{h}_s$ and $\link{h}_t$ intersect, which, due to rule \ref{item:lfIntersecting} of \cref{def:combPartition}, implies that $\link{h}_s$ and $\link{h}_t$ are in the same part of $\overline{\mathcal{T}}$, as desired.

Assume now $p>0$.
Notice that none of the links $\link{l}_i$ for $i\in [p]$ can be in $S_1\cap S_2$, because $S_1\cap S_2 = B_1\cap B_2$ by definition of $B_1$ and $B_2$.
We now observe that either all links $\link{l}_1,\dots, \link{l}_p$ are in $S_1$ or all are in $S_2$.
Assume for the sake of deriving a contradiction that this is not the case.
Hence, there is a pair of consecutive links $\link{l}_i,\link{l}_{i+1}$ for some $i\in [p-1]$ with one link in $S_1$ and the other one in $S_2$, say $\link{l}_i\in S_1$ and $\link{l}_{i+1}\in S_2$.
(The opposite case is equivalent by symmetry of $S_1$ and $S_2$.)
Because this pair of links intersects, at least one of the links $\link{l}_i, \link{l}_{i+1}$ must have one endpoint in $C_1$ and the other endpoint in $C_2$.
However, all such links are part of $B_1\cup B_2$, which contradicts the fact that no interior link of $P$ is in $B_1\cup B_2$.
Thus, assume that all interior links of $P$ are in $S_1$.
(Again, the case of all being in $S_2$ is identical by symmetry.)
Because $\link{h}_s$ intersects $\link{l}_1$ and $\link{l}_1\in S_1\setminus S_2$, we have that $\link{h}_s$ has at least one endpoint in $C_1$, and thus $\link{h}_s\in B_1$.
Analogously, we obtain $\link{l}_p\in B_1$.
Hence, the whole path $P$ lies in $H[S_1]$, which shows that $\link{h}_s$ and $\link{h}_t$ are in the same part of the partition $\overline{\mathcal{T}}$ by rule~\ref{item:lfInB1}, as desired.
This finishes the proof of the first part of \cref{lem:combPartitionOk}.

To prove the second part of \cref{lem:combPartitionOk}, let $M\subseteq S_1 \cup S_2$ be the vertex set of a connected component of $H[S_1\cup S_2]$ for which we have $M\cap (B_1\cup B_2)\neq \emptyset$.
Hence, we can write
\begin{equation}\label{eq:MasUnionOfSij}
M = \left(S^1_1 \cup \dots \cup S^1_{q_1}\right) \cup \left(S^2_1 \cup \dots \cup S^2_{q_2}\right),
\end{equation}
where, for $i\in [2]$, we have that $S^i_1,\dots, S^i_{q_i}\subseteq S_i$ are link sets corresponding to connected components of $H[S_i]$.
By \cref{lem:MDecomposesInTParts}, we have 
\begin{equation*}
S^i_j\cap B_i \neq \emptyset \qquad \forall i\in[2], j\in[q_i].
\end{equation*}
Hence, the set
\begin{equation*}
\overline{T} \coloneqq M\cap (B_1\cup B_2)
\end{equation*}
fulfills
\begin{equation*}
\overline{T} = \left(T^1_1 \cup \dots \cup T^1_{q_1}\right) \cup \left(T^2_1 \cup \dots \cup T^2_{q_2}\right),
\end{equation*}
where
\begin{equation*}
T^i_j \coloneqq S^i_j\cap B_i \qquad \forall i\in [2], j\in [q_i]
\end{equation*}
are all nonempty sets.
Thus, $T^i_j\in \mathcal{T}_i$ for $i\in [2]$ and $j\in [q_i]$ because $S_i$ realizes the pattern $Q_i$.
The desired relation now follows from
\begin{align*}
\lca(V(M)) &= \lca\left(V\left(\left(S^1_1\cup \dots \cup S^1_{q_1}\right) \cup \left(S^2_1\cup \dots \cup S^2_{q_2}\right)\right)\right)\\
&=\lca\left(\lca\left(V\left(S^1_1\right)\right),\dots, \lca\left(V\left(S^1_{q_1}\right)\right), \lca\left(V\left(S^2_1\right)\right),\dots, \lca\left(V\left(S^2_{q_2}\right)\right)\right)\\
&=\lca\left(\phi_1(T^1_1),\dots, \phi_1(T^1_{q_1}), \phi_2(T^2_1),\dots, \phi_2(T^2_{q_2})\right)\\
&=\overline{\phi}(\overline{T}),
\end{align*}
where the first equality is due to~\eqref{eq:MasUnionOfSij},
the second one holds because the least common ancestor of a vertex set in a tree is the least common ancestor of the least common ancestors of any covering of that vertex set (above, we use this by exploiting that $V((S^1_1 \cup \dots \cup S^1_{q_1})\cup(S^2_1 \cup \dots \cup S^2_{q_2})) = V(S^1_1)\cup\dots\cup V(S^1_{q_1})\cup V(S^2_1)\cup\dots\cup V(S^2_{q_2})$),
the third one holds because, for $i\in [2]$, the set $S^i$ is a $Q_i$-realizer,
and the last one holds by definition of the combined lca function $\overline{\phi}$.
\end{proof}

Using the notion of combined partition and combined lca function, we are now ready to define the merger of two compatible patterns.
\begin{definition}[Merger of compatible patterns]\label{def:merge_patterns}
Let $Q_1=(C_1,B_1,\mathcal{T}_1,\phi_1,\psi_1)\in \mathcal{Q}$ and $Q_2=(C_2,B_2,\mathcal{T}_2,\phi_2,\psi_2) \in \mathcal{Q}$ be a pair of compatible patterns with combined partition $\overline{\mathcal{T}}$ and combined lca function $\overline{\phi}$.
Then the \emph{merger} of $Q_1$ and $Q_2$ is the pattern $Q=(C,B,\mathcal{T},\phi,\psi)$ defined as follows.
\begin{itemize}
\item $C \coloneqq C_1\cup C_2$.

\item $B \coloneqq \delta_{B_1\cup B_2}(C_1\cup C_2)$.

\item $\mathcal{T}\subseteq 2^B$ consists of all nonempty sets in $\{\overline{T}\cap B\colon \overline{T}\in \overline{\mathcal{T}}\}$.

\item $\phi(\overline{T}\cap B) = \overline{\phi}(\overline{T}) \qquad \forall\; \overline{T}\in \overline{\mathcal{T}} \text{ with } \overline{T}\cap B \neq \emptyset$.

\item To define $\psi\colon \mathcal{T}\to \{0,1\}$, let $\overline{T}\in \overline{\mathcal{T}}$ and $T\coloneqq \overline{T}\cap B$, and consider a decomposition of $\overline{T}$ as in~\eqref{eq:decomp_T_merger}.
We set
\begin{equation*}
\psi(T) \coloneqq \begin{cases}
1 &\text{if $\exists i\in [2]$ and $j\in [q_i]$ with $\phi_i(T^i_j)=\phi(T)$ and $\psi_i(T_j)=1$},\\
0 &\text{otherwise}.
\end{cases}
\end{equation*}
\end{itemize}
\end{definition}

We now show that the merger $Q$ of two compatible patterns $Q_1$ and $Q_2$ indeed has the property that the union of any realizers of $Q_1$ and $Q_2$, respectively, will be a realizer of $Q$.
Moreover, we can also relate their $\pi$-values.
To this end we define the following vertex set $U_{Q_1,Q_2}$ for any two compatible patterns $Q_1=(C_1,B_1,\mathcal{T}_1,\phi_q,\psi_1)$ and $Q_2=(C_2,B_2,\mathcal{T}_2,\phi_2,\psi_2)$ with merger $Q=(C,B,\mathcal{T},\phi,\psi)$.
As stated in the lemma that follows, this set $U_{Q_1,Q_2}$ allows for describing the difference of the $\pi$-value of a set realizing $Q$ compared to the sum of the $\pi$-values of sets realizing $Q_1$ and $Q_2$, respectively.
As before, $\overline{\mathcal{T}}$ denotes the combined partition of $B_1\cup B_2$ and $\overline{\phi}$ the combined lca function.
\begin{equation*}
U_{Q_1,Q_2} \coloneqq
\left\{
\phi_i(T) \colon i\in [2], T\in \mathcal{T}_i \text{ with } \psi_i(T)=1 \text{ and } \phi_i(T)\in C_i
\right\}
\setminus \bigcup_{\overline{T}\in \overline{\mathcal{T}}} \left\{\overline{\phi}(\overline{T})\right\}
\end{equation*}
\begin{lemma}\label{lem:valueOfMergedSolution}
Let $Q_1 = (C_1,B_1,\mathcal{T}_1,\phi_1,\psi_1)\in \mathcal{Q}$ and $Q_2=(C_2,B_2,\mathcal{T}_2,\phi_2,\psi_2) \in \mathcal{Q}$ be two compatible and realizable patterns and let $Q\in \mathcal{Q}$ be their merger.
Moreover, let $S_1\subseteq L$ be an $(\alpha,C_1)$-thin link set realizing $Q_1$ and $S_2 \subseteq L$ be an $(\alpha,C_2)$-thin link set realizing $Q_2$.
Then $S_1\cup S_2$ is an $(\alpha,C_1\cup C_2)$-thin link set realizing $Q$ and
\begin{equation}\label{eq:piOfMerger}
\pi(S_1\cup S_2, C_1\cup C_2) = \pi(S_1,C_1) + \pi(S_2,C_2) + c(B_1\cap B_2) + \tilde c\left(\bigcup_{u\in U_{Q_1,Q_2}} \delta^-_{\vec{F}_0}(u)\right).
\end{equation}
\end{lemma}
\begin{proof}
We first observe that $S_1\cup S_2$ is $(\alpha,C_1\cup C_2)$-thin.
For $i\in [2]$, the fact that $S_i$ is $(\alpha,C_i)$-thin implies that there is a maximal laminar family $\mathcal{L}_i \subseteq \mathcal{C}_G$ over the ground set $C_i$ such that $|\delta_{S_i}(\overline{C})|\leq \alpha$ for all $\overline{C}\in \mathcal{L}_i$.
To show $(\alpha,C_1\cup C_2)$-thinness of $S_1\cup S_2$, we define the maximal laminar family
\begin{equation*}
\mathcal{L} \coloneqq \mathcal{L}_1 \cup \mathcal{L}_2 \cup \{C_1\cup C_2\}
\end{equation*}
over the ground set $C_1\cup C_2$ and show that $\delta_{S_1\cup S_2}(\overline{C})\leq \alpha$ for all $\overline{C} \in \mathcal{L}$.
We have
\begin{equation*}
|\delta_{S_1\cup S_2}(C_1\cup C_2)| = |\delta_{B_1\cup B_2}(C_1\cup C_2)| \leq \alpha
\end{equation*}
because $Q_1$ and $Q_2$ are compatible patterns.
Consider a set $\overline{C}\in \mathcal{L}_1$.
(The case $\overline{C}\in \mathcal{L}_2$ is identical by symmetry.)
We have
\begin{equation}\label{eq:noNewLinksFromS2IntoCB}
\delta_{S_2}(\overline{C}) \subseteq \delta_{S_2}(C_1) = \delta_{B_2}(C_1) = \delta_{B_1}(C_2) \subseteq B_1 \subseteq S_1,
\end{equation}
where the first inclusion follows because each link in $S_2$ has an endpoint in $C_2$ and $C_1 \supseteq \overline{C}$ is disjoint from $C_2$,
the first equality holds because $B_2\subseteq S_2$ are precisely the links in $S_2$ with one endpoint outside of $C_2$,
and the second equality is due to compatibility of $Q_1$ and $Q_2$.
We thus obtain as desired $(\alpha,C_1\cup C_2)$-thinness of $S_1\cup S_2$ because
\begin{equation*}
|\delta_{S_1\cup S_2}(\overline{C})| = |\delta_{S_1}(\overline{C}) \cup \delta_{S_2}(\overline{C})| = |\delta_{S_1}(\overline{C})|\leq \alpha,
\end{equation*}
where the second equality follows from~\eqref{eq:noNewLinksFromS2IntoCB}, and the inequality from $S_1$ being $(\alpha,C_1)$-thin.

\smallskip

We now show that $S_1\cup S_2$ realizes $Q=(C,B,\mathcal{T},\phi,\psi)$ by checking one-by-one the five properties~\labelcref{item:realizerOneEndpointInC,item:realizerBOk,item:realizerTOk,item:realizerPhiOk,item:realizerPsiOk} of \cref{def:realizingLinkSet}.

Property~\ref{item:realizerOneEndpointInC} holds because, for $i\in [2]$, each link in $S_i$ has at least one endpoint in $C_i$; hence, each link in $S_1\cup S_2$ has at least one endpoint in $C_1\cup C_2$.

Property~\ref{item:realizerBOk} of \cref{def:realizingLinkSet} follows from
\begin{equation*}
\delta_{B_1\cup B_2}(C_1\cup C_2) = \delta_{S_1\cup S_2}(C_1\cup C_2),
\end{equation*}
because for a link in $S_1\cup S_2$ to cross $C_1\cup C_2$ it must either cross $C_1$ or $C_2$.

Both property~\ref{item:realizerTOk} and~\ref{item:realizerPhiOk} are direct consequences of \cref{lem:combPartitionOk}.

It remains to check that $\psi$ fulfills property~\ref{item:realizerPsiOk} of \cref{def:realizingLinkSet}.
(See \cref{def:merge_patterns} to recall how the function $\psi$ is defined.)
To this end, let $M\subseteq S_1\cup S_2$ be a connected component of $H[S_1\cup S_2]$ with $M\cap B \neq \emptyset$.
Hence, $M\cap B \in \mathcal{T}$.
To show property~\ref{item:realizerPsiOk} of \cref{def:realizingLinkSet}, we have to prove the following equivalence:
\begin{equation}\label{eq:equivToShowPsiOk}
\psi(M\cap B) = 1 \iff \lca(V(M)) \in V(M).
\end{equation}
We can write $M$ as
\begin{equation*}
M = \left(S^1_1 \cup \ldots \cup S^1_{q_1}\right) \cup \left(S^2_1\cup \ldots \cup S^2_{q_2}\right),
\end{equation*}
where, for $i\in [2]$, we have that $S^i_1,\dots,S^i_{q_i}\subseteq S_i$ are link sets corresponding to connected components of $H[S_i]$.
By \cref{lem:MDecomposesInTParts}, we have
\begin{equation*}
T^i_j\coloneqq S^i_j\cap B_i \neq \emptyset \qquad \forall i\in[2], j\in[q_i],
\end{equation*}
and thus
\begin{equation*}
T^i_j\in \mathcal{T}_i \qquad \forall i\in[2], j\in[q_i].
\end{equation*}
Hence, we can express $\overline{T}\coloneqq M\cap (B_1\cup B_2)$ as in~\eqref{eq:decomp_T_merger} as follows:
\begin{equation*}
M\cap (B_1\cup B_2) = \left(T^1_1 \cup \ldots \cup T^1_{q_1}\right) \cup \left(T^2_1\cup \ldots \cup T^2_{q_2}\right),
\end{equation*}
and by definition of $\phi$ (see \cref{def:combPartition}) and the fact that $\phi(M\cap B)=\lca V(M)$, which follows from property~\ref{item:realizerPhiOk}, we have
\begin{equation}\label{eq:lcaVMIsLcaOfATij}
\lca(V(M)) = \phi(M\cap B) \coloneqq \lca\left(\left\{\phi_1(T^1_1),\ldots, \phi_1(T^1_{q_1}), \phi_2(T^2_1),\ldots, \phi_2(T^2_{q_2})\right\}\right).
\end{equation}

We start by proving the left-to-right implication of~\eqref{eq:equivToShowPsiOk}.
$\psi(M\cap B)=1$ implies by \cref{def:merge_patterns} that there is a $i\in [2]$ and $j\in [q_i]$ with $\phi_i(T^i_j)=\phi(\overline{T}\cap B)$ and $\psi_i(T^i_j)=1$.
This implies as desired
\begin{equation*}
\lca(V(M)) = \phi(\overline{T}\cap B) = \phi_i(T_j^i) \in V(S_j^i) \subseteq V(M),
\end{equation*}
where the first equality holds by property~\ref{item:realizerPhiOk}, the second one by assumption as mentioned above, the relation $\phi_i(T^i_j)\in V(S^i_j)$ holds because $\psi_i(T^i_j)=1$, and the final inclusion is due to $S^i_j\subseteq M$.

Conversely, to show the right-to-left direction of~\eqref{eq:equivToShowPsiOk}, assume $\lca(V(M))\in V(M)$.
Thus, there is $i\in [2]$ and $j\in [q_i]$ with $\lca(V(M))\in V(S^i_j$).
Because $S^i_j\subseteq M$, this implies
\begin{equation*}
\phi_i(S^i_j) = \lca(V(S^i_j)) = \lca(V(M))\in V(S^i_j),
\end{equation*}
where the first equality holds because $S_i$ is a realizer of $Q_i$, and the second follows from the fact that $\lca(V(M))$ is a common ancestor of all of $V(M)$ and therefore also of $V(S^i_j)$; moreover, because $\lca(V(M))\in S^i_j$, the vertex set $V(S^i_j)$ cannot have a lower ancestor than $\lca(V(M))$.
Hence, $\psi_i(T^i_j)=1$, because $S_i$ is a realizer of $Q_i$.
We therefore obtain as desired
\begin{equation*}
\psi(M\cap B) = \psi((M\cap(B_1\cup B_2))\cap B) = \psi(\overline{T}\cap B) = 1,
\end{equation*}
by definition of $\psi$ (see \cref{def:merge_patterns}).

\smallskip

It remains to show the relation~\eqref{eq:piOfMerger}.
We start by expanding the right-hand side of~\eqref{eq:piOfMerger}:
\begin{align*}
\pi(S_1,&C_1) + \pi(S_2,C_2) + c(B_1\cap B_2) + \tilde c\bigg(\bigcup_{u\in U_{Q_1,Q_2}} \delta^-_{\vec{F}_0}(u)\bigg)\\
=\; &\tilde c\Bigg(\bigg(\Drop_{\vec{F}_0}(S_1)\cap \bigcup_{v\in C_1} \delta^-_{\vec{F}_0}(v)\bigg)\cup \bigg(\Drop_{\vec{F}_0}(S_2)\cap \bigcup_{v\in C_2} \delta^-_{\vec{F}_0}(v)\bigg)\Bigg)\\
&- c(S_1) - c(S_2) + c(B_1\cap B_2) + \tilde c\bigg(\bigcup_{u\in U_{Q_1,Q_2}} \delta^-_{\vec{F}_0}(u)\bigg)\\
=\; &\tilde c \Bigg(
 \bigg(\Drop_{\vec{F}_0}(S_1)\cap \bigcup_{v\in C_1} \delta^-_{\vec{F}_0}(v)\bigg)\cup \bigg(\Drop_{\vec{F}_0}(S_2)\cap \bigcup_{v\in C_2} \delta^-_{\vec{F}_0}(v)\bigg)\Bigg)\\
& + \tilde c\bigg(\bigcup_{u\in U_{Q_1,Q_2}} \delta^-_{\vec{F}_0}(u)\bigg) + c(B_1\cap B_2)
- c(S_1\cup S_2),
\end{align*}
where the first equality uses that $C_1$ and $C_2$ are disjoint.
Analogously, the left-hand side of~\eqref{eq:piOfMerger} equals
\begin{equation*}
\pi(S_1\cup S_2,C_1\cup C_2) = \tilde c\left(\Drop_{\vec{F}_0}(S_1\cup S_2)\cap \bigcup_{v\in C_1\cup C_2} \delta^-_{\vec{F}_0}(v)\right)  + c(B_1\cap B_2) - c(S_1\cup S_2).
\end{equation*}
Hence, it remains to show
\begin{multline*}
\left(\Drop_{\vec{F}_0}(S_1\cup S_2)\cap \bigcup_{v\in C_1\cup C_2} \delta^-_{\vec{F}_0}(v)\right)
\setminus \left(
\left(\Drop_{\vec{F}_0}(S_1)\cap \bigcup_{v\in C_1} \delta^-_{\vec{F}_0}(v)\right)\right. \cup \\
\left.\left(\Drop_{\vec{F}_0}(S_2)\cap \bigcup_{v\in C_2} \delta^-_{\vec{F}_0}(v)\right)
\right)
=
\bigcup_{u\in U_{Q_1,Q_2}} \delta^-_{\vec{F}_0}(u).
\end{multline*}
To further expand the $\Drop$-expressions above, we apply \cref{lem:ov_dropOfConnectedS}.
To this end, for $i\in [2]$, let $\mathcal{S}^i\subseteq 2^{S_i}$ be the partition of $S_i$ into the links sets corresponding to the connected components of $H[S_i]$.
Analogously, let $\mathcal{M}\subseteq 2^{S_1\cup S_2}$ be the partition of $S_1\cup S_2$ into the link sets corresponding to the connected components of $H[S_1\cup S_2]$.
\cref{lem:ov_dropOfConnectedS} now implies, for $i\in [2]$,
\begin{align*}
\Drop_{\vec{F}_0}(S_i) \cap \bigcup_{v\in C_i}\delta_{\vec{F}_0}(v) &=
\left(\bigcup_{v\in V(S_i)\cap C_i} \delta_{\vec{F}_0}^-(v)\right)\setminus
\bigcup_{\substack{S^i\in \mathcal{S}^i:\\ \lca(V(S^i))\in V(S^i)}} \delta_{\vec{F}_0}^-(\lca(V(S^i)))
\intertext{and, moreover,}
\Drop_{\vec{F}_0}(S_1\cup S_2) \cap \bigcup_{v\in C_1\cup C_2}\delta_{\vec{F}_0}(v) &=
\left(\bigcup_{v\in V(S_1\cup S_2)\cap (C_1\cup C_2)} \delta_{\vec{F}_0}^-(v)\right)\setminus
\bigcup_{\substack{M\in \mathcal{M}:\\ \lca(V(M))\in V(M)}} \delta_{\vec{F}_0}^-(\lca(V(M))).
\end{align*}

Because $C_1 \cap V(S_2) \subseteq C_1 \cap V(S_1)$ and $C_2 \cap V(S_1) \subseteq C_2 \cap V(S_2)$, we have $V(M) \cap (C_1 \cup C_2) = (V(S_1) \cap C_1) \cup (V(S_2) \cap C_2))$.
Thus, the above relations imply the following, where we denote by $M_{S^i}\in \mathcal{M}$, for $i\in [2]$ and $S^i \in \mathcal{S}^i$, the set in $\mathcal{M}$ containing $S^i$, i.e., $S^i\subseteq M_{S^i}$:
\begin{multline*}
\left(\Drop_{\vec{F}_0}(S_1\cup S_2)\cap \bigcup_{v\in C_1\cup C_2} \delta^-_{\vec{F}_0}(v)\right)
\setminus \left(
\left(\Drop_{\vec{F}_0}(S_1)\cap \bigcup_{v\in C_1} \delta^-_{\vec{F}_0}(v)\right)\right. \cup \\
\left.\left(\Drop_{\vec{F}_0}(S_2)\cap \bigcup_{v\in C_2} \delta^-_{\vec{F}_0}(v)\right)
\right)
=
\bigcup_{u\in U} \delta^-_{\vec{F}_0}(u),
\end{multline*}
where
\begin{equation*}
U \coloneqq \left\{ \lca(V(S^i)) \colon i\in [2], S^i \in \mathcal{S}^i \text{ with } \lca(V(S^i))\in V(S^i)\cap C_i \text{ and } \lca(V(S^i)) \neq \lca(V(M_{S_i}))
\right\}.
\end{equation*}
It thus suffices to show $U = U_{Q_1,Q_2}$.

The following claim shows that each of the sets $S^i$ considered in the definition of $U$ above has a nonempty intersection with $B_i$.
This implies that $S^i\cap B_i\in \mathcal{T}_i$ and will allow us in the following to write the set $U$ in terms of sets in $\mathcal{T}_1$ and $\mathcal{T}_2$.
\begin{claiminproof}\label{claim:relSiIntersectBi}
For any $i\in [2]$ and $S^i\in \mathcal{S}^i$ with $\lca(V(S^i))\neq \lca(V(M_{S^i}))$, we have $S^i\cap B_i\neq \emptyset$.
\end{claiminproof}
\begin{proof}[Proof of claim]
Because $\lca(V(S^i))\neq \lca(V(M_{S^i}))$, we have $M_{S^i} \supsetneq S^i$.
Moreover, as $S^i$ is the link set corresponding to a connected component in $H[S_i]$, and $M_{S^i}$ is the link set corresponding to the connected component of $H[S_1\cup S_2]$ that contains $S^i$, we have that $S^i$ must intersect some link set in $\mathcal{S}^{2-i}$.
Hence, it contains at least one link with one endpoint in $C_i$ and one in $C_{2-i}$.
Such a link must be in $B_i$, as desired.
\end{proof}

By \cref{claim:relSiIntersectBi}, we obtain
\begin{equation}\label{eq:UInTermsOfTSets}
U = \left\{\phi_i(T) \colon i\in [2] \text{ and } T\in \mathcal{T}_i \text{ with }
\psi_i(T) = 1\text{, }
\phi_i(T) \in C_i\text{, and }
\phi_i(T) \neq \lca(V(M_T))
\right\}.
\end{equation}

Finally, \cref{claim:onlyMTIsRelevant} below allows for observing that the description of $U$ in~\eqref{eq:UInTermsOfTSets} corresponds to the definition of $U_{Q_1,Q_2}$.
\begin{claiminproof}\label{claim:onlyMTIsRelevant}
Let $i\in [2]$, and $T\in \mathcal{T}_i$ with $\psi_i(T)=1$ and $\lca(V(M_T))\neq \phi_i(T)$.
Then, $\lca(V(M))\neq \phi_i(T)$ for all $M\in \mathcal{M}$.
\end{claiminproof}
\begin{proof}[Proof of claim]
Let $M\in \mathcal{M}$ with $M\neq M_T$, and, for the sake of deriving a contradiction, assume $\lca(V(M))=\phi_i(T)$.
Note that because $M$ and $M_T$ are link sets corresponding to different connected components of $H[S_1\cup S_2]$, we have that no link from $M$ intersects a link from $M_T$.
This implies in particular $V(M)\cap V(M_T)=\emptyset$.

Let $s_{M_T},t_{M_T}\in V(M_T)$ be the leftmost and rightmost vertex of $V(M_T)$, respectively.
Analogously, let $s_M, t_M\in V(M)$ be the leftmost and rightmost vertex of $V(M)$, respectively.
By \cref{lem:lca_from_leftmost_rightmost} we have that $s_M$ and $t_M$ lie on different sides of of $\lca(V(M))=\phi_i(T)$.
The same holds for $s_{M_T}$ and $t_{M_T}$.
Moreover, because $\psi_i(T)=1$, the vertex $\phi_i(T)$ is contained in $V(M)$.

First observe that both $s_M$ and $t_M$ must lie outside the interval from $s_{M_T}$ to $t_{M_T}$.
Indeed, if this were not the case, say $s_M$ lies within the interval from $s_{M_T}$ to $t_{M_t}$ (the case of $t_M$ lying in this interval is identical by symmetry), then $s_M$ lies strictly between $s_{M_T}$ and $\phi_i(T)$, whereas $t_M$ does not lie strictly between these points as it is to the right of $\phi_i(T)$.
However, as both $H[M]$ and $H[M_T]$ are connected, this implies that there must exist a link in $M$ intersecting a link in $M_T$, which contradicts that $M$ and $M_T$ are link sets corresponding to different connected components in $H[S_1\cup S_2]$.

Hence, $s_M$ lies to the left of $s_{M_T}$, and $t_M$ to the right of $t_{M_T}$.
As the least common ancestor of a vertex set depends solely on the leftmost and rightmost vertex, as guaranteed by \cref{lem:lca_from_leftmost_rightmost}, we have that $\lca(V(M))$ is an ancestor of $\lca(V(M_T))$.
However, as $\phi_i(T)\in V(M_T)$ (because $\psi_i(T)=1$) and $\lca(V(M_T))\neq \phi_i(T)$, this implies that $\lca(V(M_T))$ is a strict ancestor of $\phi_i(T)$.
Because $\lca(V(M))$ is an ancestor of $\lca(V(M_T))$, it is therefore also a strict ancestor of $\phi_i(T)$, and we thus obtain $\lca(V(M_T))\neq \phi_i(T)$, which contradicts the initial assumption.
\end{proof}

Indeed, we now obtain
\begin{align*}
U &= \bigg\{\phi_i(T) \colon i\in [2] \text{ and } T\in \mathcal{T}_i \text{ with }
\psi_i(T) = 1,
\phi_i(T) \in C_i,
\phi_i(T) \neq \lca(V(M_T))
\bigg\}\\
&= \left\{\phi_i(T) \colon i\in [2] \text{ and } T\in \mathcal{T}_i \text{ with }
\psi_i(T) = 1,
\phi_i(T) \in C_i,
\phi_i(T) \not\in \bigcup_{\overline{T}\in \overline{\mathcal{T}}}\left\{\lca(V(\overline{T}))\right\}
\right\}\\
&= \bigg\{\phi_i(T) \colon i\in [2] \text{ and } T\in \mathcal{T}_i \text{ with }
\psi_i(T) = 1,
\phi_i(T) \in C_i
\bigg\} \setminus 
\bigcup_{\overline{T}\in \overline{\mathcal{T}}} \left\{\overline{\phi}_i(\overline{T})\right\}\\
&= U_{Q_1,Q_2},
\end{align*}
because of the following.
The first equality is due to~\eqref{eq:UInTermsOfTSets},
the second one uses \cref{claim:onlyMTIsRelevant} and the fact that, for $i\in [2]$ and $T\in \mathcal{T}_i$, we have $M_T\cap (B_1\cap B_2)\supseteq T\cap B_i \neq \emptyset$ and thus $M_T\cap (B_1\cap B_2)\in \overline{\mathcal{T}}$.
The third equation follows from $\overline{\phi}_i(\overline{T})=\lca(V(\overline{T}))$, and the last one by the definition of $U_{Q_1,Q_2}$.
\end{proof}

Leveraging \cref{lem:valueOfMergedSolution}, we can now express how, for a pattern $Q\in \mathcal{Q}$, the optimal value $\pi^*_Q$ can be expressed in terms of patterns with respect to smaller $2$-cuts.
This is the result that will be exploited in the propagation step of the dynamic program.
\begin{lemma}\label{lem:dpPropagationRel}
Let $Q=(C,B,\mathcal{T},\phi,\psi)\in \mathcal{Q}$ be a realizable pattern with $|C|\geq 2$.
Then
\begin{multline}\label{eq:dpPropagationRel}
\pi_Q^* = \max\Bigg\{
\pi_{Q_1}^* + \pi_{Q_2}^* + c(B_1\cap B_2) + \tilde c\left(\bigcup_{u\in U_{Q_1,Q_2}} \delta^-_{\vec{F}}(u)\right)\\
\colon
\text{$Q_1=(C_1,B_1,\mathcal{T}_1,\phi_1,\psi_1),Q_2=(C_2,B_2,\mathcal{T}_2,\phi_2,\psi_2)\in \mathcal{Q}$ compatible}\\
\text{and realizable patterns with merger $Q$}
\Bigg\}.
\end{multline}
Moreover, if $Q_1,Q_2\in \mathcal{Q}$ are patterns maximizing the right-hand side above, then, for any maximizing realizers $S_{Q_1}$ and $S_{Q_2}$ of $Q_1$ and $Q_2$, respectively, $S_{Q_1}\cup S_{Q_2}$ is a $Q$-maximizing realizer.
\end{lemma}
\begin{proof}
We start by observing that the left-hand side of \eqref{eq:dpPropagationRel} is no more than its right-hand side.
This is a consequence of \cref{lem:valueOfMergedSolution} due to the following.
Let $Q_1=(C_1,B_1,\mathcal{T}_1,\phi_1,\psi_1),Q_2=(C_2,B_2,\mathcal{T}_2,\phi_2,\psi_2)$ be a pair of compatible and realizable patterns with merger $Q$.
Let $S_{Q_1}^*$ and $S_{Q_2}^*$ be maximizing realizers for $Q_1$ and $Q_2$, respectively.
By \cref{lem:valueOfMergedSolution} we have that $S^*_{Q_1}\cup S^*_{Q_2}$ is a realizer of $Q$.
This leads to the desired relation because
\begin{align*}
\pi^*_{Q} &\geq \pi(S^*_{Q_1}\cup S^*_{Q_2},C)\\
 &= \pi(S^*_{Q_1}, C_1) + \pi(S^*_{Q_2}, C_2) + c(B_1\cap B_2) + \tilde c\left(\bigcup_{u\in U_{Q_1,Q_2}} \delta^-_{\vec{F}}(u)\right)\\
 &= \pi^*_{Q_1} + \pi^*_{Q_2} + c(B_1\cap B_2) + \tilde c\left(\bigcup_{u\in U_{Q_1,Q_2}} \delta^-_{\vec{F}}(u)\right),
\end{align*}
where the first inequality holds because $S^*_{Q_1}\cup S^*_{Q_2}$ is a $Q$-realizer and $\pi_Q^*$ is the value of a maximizing $Q$-realizer, the first equality is due to \cref{lem:valueOfMergedSolution}, and the last one holds because $S_{Q_1}^*$ and $S_{Q_2}^*$ are maximizing realizers for $Q_1$ and $Q_2$, respectively.

To show that the left-hand side of \eqref{eq:dpPropagationRel} is no more than its right-hand side, let $S_Q^*\subseteq L$ be a maximizing realizer of $Q$, and let $\mathcal{L}\subseteq 2^C$ be a maximal laminar subfamily of $C$ that certifies $(\alpha,C)$-thinness of $S^*_Q$, i.e., $|\delta_{S^*_Q}(\overline{C})|\leq \alpha$ for all $\overline{C} \in \mathcal{L}$.
By \cref{lem:maximal_laminar}, $C$ has exactly two children in the laminar family $\Lscr$ and we denote these children by $C_1,C_2\in \mathcal{L}$.
Let
\begin{equation*}
S_i \coloneqq \bigcup_{v\in C_i} \delta_{S^*_{Q}}(v) \qquad \forall i\in [2].
\end{equation*}
For $i\in [2]$, we denote by $Q_i=(C_i,B_i,\mathcal{T}_i,\phi_i,\psi_i)$ the $C_i$-pattern that is realized by $S_i$.
(Observe that $S_i$ is indeed $(\alpha,C_i)$-thin, which is certified by the laminar subfamily of $\mathcal{L}$ consisting of $C_i$ and all of its descendants in $\mathcal{L}$.)
Note that $Q_1$ and $Q_2$ are compatible.
Moreover, \cref{lem:valueOfMergedSolution} implies that the merger of $Q_1$ and $Q_2$ is $Q$, because $S^*_Q=S_1\cup S_2$ realizes $Q$.
Finally, we obtain as desired
\begin{align*}
\pi_Q^* &= \pi(S_Q^*,C)\\
  &= \pi(S_1,C_1) + \pi(S_2,C_2) + c(B_1\cap B_2) + \tilde c\left(\bigcup_{u\in U_{Q_1,Q_2}} \delta^-_{\vec{F}}(u)\right)\\
  &\leq \pi^*_{Q_1} + \pi^*_{Q_2} + c(B_1\cap B_2) + \tilde c\left(\bigcup_{u\in U_{Q_1,Q_2}} \delta^-_{\vec{F}}(u)\right),
\end{align*}
where the first equality holds because $S_Q^*$ is a $Q$-maximizing realizer, the second one follows from \cref{lem:valueOfMergedSolution}, and the inequality by the definition of $\pi^*_{Q_1}$ and $\pi^*_{Q_2}$.

It remains to show the last part of \cref{lem:dpPropagationRel}.
To this end, let $Q_1=(C_1,B_1,\mathcal{T}_1,\phi_1,\psi_1),Q_2=(C_2,B_2,\mathcal{T}_2,\phi_2,\psi_2)\in \mathcal{Q}$ be patterns maximizing the right-hand side of~\eqref{eq:dpPropagationRel}, and let $S_{Q_1}$ and $S_{Q_2}$ be maximizing realizers of $Q_1$ and $Q_2$, respectively.
By \cref{lem:valueOfMergedSolution}, we have that $S_{Q_1}\cup S_{Q_2}$ is a realizers of $Q$ and
\begin{align*}
\pi(S_{Q_1}\cup S_{Q_2}, C) &= \pi(S_{Q_1},C_1) + \pi(S_{Q_2},C_2) + c(B_1\cap B_2) + \tilde c\left(\bigcup_{u\in U_{Q_1,Q_2}}\delta^-_{\vec{F}}(u)\right)\\
 &= \pi^*_Q,
\end{align*}
where the second equality follows from~\eqref{eq:dpPropagationRel}, the fact that $Q_1,Q_2\in \mathcal{Q}$ maximize the right-hand side of~\eqref{eq:dpPropagationRel}, and that $S_{Q_1}$ and $S_{Q_2}$ are maximizing realizers of $Q_1$ and $Q_2$, respectively.
Hence, $S_{Q_1}\cup S_{Q_2}$ is a $Q$-maximizing realizer as claimed.
\end{proof}
\cref{lem:dpPropagationRel}, and the fact that we can compute maximizing realizers for patterns corresponding to singleton cuts through enumeration, leads to a natural way to compute the optimal values $\pi^*_Q$ for all realizable patterns $Q\in \mathcal{Q}$ together with maximizing realizers $S_Q$, by considering patterns $Q=(C,B,\mathcal{T},\phi,\psi)$ in order of increasing size of their $2$-cut $C$.
\cref{alg:dpMaxRealizers} provides a pseudocode of the dynamic programming procedure we employ.
For a realizable pattern $Q=(C,B,\mathcal{T},\phi,\psi)$, we denote in the algorithm by $\pi_Q$ the value of the computed set $S_Q$.
As we will show, the algorithm indeed computes maximizing realizers $S_Q$, which then implies $\pi_Q=\pi_Q^*$.

In the inner for-loop of the algorithm, we may consider patterns $Q\in \mathcal{Q}$ for which there is no compatible pair $Q_1,Q_2\in Q$ that realizes $Q$.
Such patterns are not realizable and do not need to be considered further.
One formal way to handle this case is by setting $\pi_Q=-\infty$ in \eqref{eq:dpPropagationRel}, whenever $Q$ does not admit a compatible pair $Q_1,Q_2$ with merger $Q$.
This will make sure that the algorithm assigns a value of $\pi_Q=-\infty$ to all patterns $Q\in \mathcal{Q}$ that are not realizable (and finite values $\pi_Q$ to all realizable patterns $Q$).

\begin{algorithm2e}[H]
\For{$Q=(C,B,\mathcal{T},\phi,\psi)\in \mathcal{Q}$ with $|C|=1$}{
Determine maximizing realizer $S_Q$ (and corresponding value $\pi_Q$) by checking all link sets of size at most $\alpha$.
If $Q$ is not realizable, set $\pi_Q=-\infty$ and $S_Q=\emptyset$.
}

\For{$i=2$ to $|V|-1$}{
  \For{$Q=(C,B,\mathcal{T},\phi,\psi)\in \mathcal{Q}$ with $|C|=i$}{
     Among all compatible $Q_1=(C_1,B_1,\mathcal{T}_1,\phi_1,\psi_1), Q_2=(C_2,B_2,\mathcal{T}_2,\phi_2,\psi_2)\in \mathcal{Q}$ with merger $Q$, determine a pair that maximizes
\begin{equation*}
\pi_{Q_1} + \pi_{Q_2} + c(B_1\cup B_2) + \tilde c\left(\bigcup_{u\in U_{Q_1,Q_2}} \delta^-_{\vec{F}_0}(u)\right).
\end{equation*}\

Set $\pi_{Q} = \pi_{Q_1} + \pi_{Q_2} + c(B_1\cup B_2) + \tilde c\left(\bigcup_{u\in U_{Q_1,Q_2}} \delta^-_{\vec{F}_0}(u)\right)$.

Set $S_{Q} = S_{Q_1} \cup S_{Q_2}$.
  }
} 

\caption{Dynamic program to compute maximizing realizers}\label{alg:dpMaxRealizers}
\end{algorithm2e}

Finally, the following statement shows that \cref{alg:dpMaxRealizers} correctly and efficiently computes the desired quantities.
\begin{theorem}\label{thm:dpIsEfficientAndCorrect}
\cref{alg:dpMaxRealizers} is a polynomial-time procedure that computes, for each realizable $Q\in \mathcal{Q}$, the optimal value $\pi_Q=\pi^*_Q$ and a $Q$-maximizing realizer $S_Q$.
Moreover, for each $Q\in \mathcal{Q}$ that is not realizable, we have $\pi_Q=-\infty$.
\end{theorem}
\begin{proof}
Efficiency of the algorithm follows from the fact that $|Q|$ is polynomially bounded because $\alpha=O(1)$.

We show by induction on the size $|C|$ of the patterns $Q=(C,B,\mathcal{T},\phi,\psi)\in \mathcal{Q}$ that \cref{alg:dpMaxRealizers} computes the correct values $\pi_Q$ and sets $S_Q$.
For a pattern $Q$ with $|C|=1$, this holds because the first for-loop of the algorithm exhaustively checks all possible realizers of $Q$.
Hence, assume that the values $\pi_Q$ and sets $S_Q$ have the claimed properties for any pattern $Q=(C,B,\mathcal{T},\phi,\psi)$ with $|C|\leq k$, for some integer $k\in \mathbb{Z}_{\geq 1}$.
In particular, for any realizable pattern $Q=(C,B,\mathcal{T},\phi,\psi)$ with $|C|\leq k$, we have $\pi_Q=\pi_Q^*$.

Let $Q=(C,B,\mathcal{T},\phi,\psi)$ be a pattern with $|C|=k+1$.
First consider the case where $Q$ is not realizable.
Hence, for any compatible pair of patterns $Q_1,Q_2\in \mathcal{Q}$ with merger $Q$, we must have that at least one of $Q_1$ or $Q_2$ is not realizable.
Indeed, if both $Q_1$ and $Q_2$ were realizable, then, by \cref{lem:valueOfMergedSolution}, also $Q$ would be realizable.
Thus, the inner for-loop of \cref{alg:dpMaxRealizers} sets $\pi_Q=-\infty$ for any such non-realizable $Q$, as desired.\footnote{Note that for any pattern $Q=(C,B,\mathcal{T},\phi,\psi)\in \mathcal{Q}$ with $|C|\geq 2$, there is a pair of compatible patterns $Q_1=(C_1,B_1,\mathcal{T}_1,\phi_1,\psi_1)\in \mathcal{Q}$ and $Q_2=(C_2,B_2,\mathcal{T}_2,\phi_2,\psi_2)\in \mathcal{Q}$ with merger $Q$, even if $Q$ is not realizable.
Such a compatible pair can be obtained by choosing an arbitrary non-trivial partition $C_1, C_2$ of $C$ into neighboring 2-cuts and by setting $B_i \coloneqq \bigcup_{v\in C_i}\delta_B(v)$ for $i\in \{1,2\}$ (and choosing $\mathcal{T}_i, \phi_i, \psi_i$ arbitrarily such that $Q_1,Q_2\in Q$). 
}

Hence, from now on, let $Q=(C,B,\mathcal{T},\phi,\psi)$ be a realizable pattern with $|C|=k+1$.
By the induction hypothesis, for every pair $Q_1,Q_2$ of compatible and realizable patterns with merger $Q$, we have $\pi_{Q_1} +\pi_{Q_2} = \pi^*_{Q_1} +\pi^*_{Q_2} $.
Therefore, \cref{lem:dpPropagationRel} implies that $\pi_Q = \pi^*_Q$ and $S_Q$ is a $Q$-maximizing realizer.
\end{proof}

We recall that being able to efficiently compute maximizing realizers $S_Q$ for all realizable $Q\in \mathcal{Q}$ readily allows for identifying an $\alpha$-thin component $K\in \mathfrak{K}$ maximizing $\tilde c(\Drop_{\vec{F}_0}(K)) - c(K)$, as desired.
This can be achieved by considering, among all patterns $Q=(C,B,\mathcal{T},\phi,\psi)$ with $C=V\setminus \{r\}$, the one with highest value $\pi_Q^*$ and returning $K=S_Q$.
Thus, this finishes the proof of \cref{thm:dp-main}.
 \section{The Local Search Algorithm}\label{sec:local-search}

In this section we prove our main result, \cref{thm:main}. 
To this end, we use the local search technique from \cite{traub_2022_local} to improve on the relative greedy algorithm (\cref{algo:relative_greedy}).
All arguments in this section are analogous to \cite{traub_2022_local}.

In contrast to the relative greedy algorithm, where we removed only links from the initial directed WRAP solution, we will now also remove links added in previous iterations. 
In order to be able to apply the decomposition theorem, which allows us to improve only on directed solutions, we reinterpret undirected links $\{u,v\}$ as the combination of the two directed links $(u,v)$ and $(v,u)$.
This leads to a local search procedure where we maintain a WRAP solution $F$ together with a non-shortenable directed WRAP solution $\vec{F}$.

More precisely, our local search algorithm works as follows.
We start by choosing an arbitrary WRAP solution $F$.
Next construct the directed WRAP solution $\vec{F}$ from $F$ by first replacing every link $\{u,v\} \in F$ by its two shadows $(u,v)$ and $(v,u)$ and then applying \cref{lem:shortening_efficiently}, i.e., shortening (or removing) directed links so that $\vec{F}$ is non-shortenable.
For a link $\link{f}= \{u,v\}\in F$, we call the (up to) two directed links in $\vec{F}$ that arose from shortening $(u,v)$ and $(v,u)$, the witness set $W_{\link{f}}$ of the link $\link{f}= \{u,v\}$.
See \cref{fig:witness_sets}.

\begin{figure}[!ht]
\begin{center}
\begin{tikzpicture}[scale=1,
ns/.style={thick,draw=black,fill=white,circle,minimum size=7,inner sep=2pt},
es/.style={thick},
lks/.style={line width=1.7pt, blue, densely dashed},
dlks/.style={lks, -latex},
ts/.style={every node/.append style={font=\scriptsize}}
]

\def\rad{2.5}
\def\num{12}

\colorlet{Gcol}{green!50!black}
\colorlet{dacol}{cyan}
\colorlet{Kcol}{orange}

\coordinate (c) at (0,0);

\begin{scope}[shift={(-3.5,0)}]

\begin{scope}[every node/.style={ns}]
\foreach \i in {1,...,\num} {
  \pgfmathsetmacro\r{90+(\i-1)*360/\num}
  \node (\i) at (\r:\rad) {};
}
\end{scope}

\path (\num) to node[above=-2pt] {$e_r$} (1);
\node[above=3pt] at (1) {$r$};

\begin{scope}[es]
\foreach \i in {1,...,\num} {
\pgfmathtruncatemacro\j{1+mod(\i,\num)}
\draw (\i) -- (\j);
}
\end{scope}

\begin{scope}[lks,solid,line width=2pt,relative,black,opacity=0.7]
\draw[red] (1) to[bend left=10] (7);
\draw[blue] (6) to[bend left=10] (12);
\draw[green!70!black] (5) to[bend right=10] (11);
\draw[orange!90!black] (3) to[out=20,in=160] (5);
\draw[green!25!black] (2) to[out=20,in=160] (4);
\draw[cyan!90!black] (8) to[out=20,in=160] (10);
\draw[magenta] (9) to[out=20,in=160] (11);
\end{scope}

\end{scope}

\begin{scope}[shift={(3.5,0)}]

\begin{scope}[every node/.style={ns}]
\foreach \i in {1,...,\num} {
  \pgfmathsetmacro\r{90+(\i-1)*360/\num}
  \node (\i) at (\r:\rad) {};
}
\end{scope}

\path (\num) to node[above=-2pt] {$e_r$} (1);
\node[above=3pt] at (1) {$r$};

\begin{scope}[es]
\foreach \i in {1,...,\num} {
\pgfmathtruncatemacro\j{1+mod(\i,\num)}
\draw (\i) -- (\j);
}
\end{scope}

\begin{scope}[dlks,opacity=0.7]
\draw[red] (1) to[out=-100, in=100] (7);
\draw[blue] (7) to[bend right=40] (6);
\draw[green!70!black] (7) to[out=70,in=-140] (11);
\draw[green!70!black] (6) to[bend right=30] (5);
\draw[orange!90!black] (5) to[out=50,in=-30] (3);
\draw[green!25!black] (3) to[bend right=30] (2);
\draw[green!25!black] (3) to[bend left=30] (4);
\draw[blue] (11) to[bend left=40] (12);
\draw[magenta] (11) to[out=-110, in=130] (9);
\draw[cyan!90!black] (9) to[bend left=30] (10);
\draw[cyan!90!black] (9) to[bend right=30] (8);
\end{scope}

\end{scope}

\end{tikzpicture}
 \end{center}
\caption{ 
An example of a WRAP solution $F$ (left) and a resulting non-shortenable directed WRAP solution $\vec{F}$, which is the disjoint union of the witness sets $W_{\link{f}}$ for $\link{f}\in F$.
For every link $\link{f}\in F$, the witness set $W_{\link{f}} \subseteq \vec{F}$ is shown in the same color as the link $\link{f}$.
\label{fig:witness_sets}}
\end{figure}
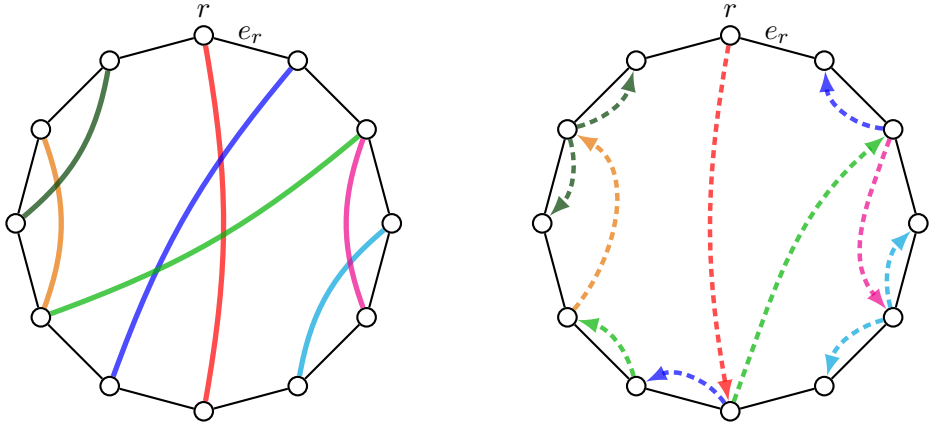

If the witness set of a link $\link{f}\in F$ is empty, we remove it from $F$.
We claim that $F$ remains a WRAP solution because $\vec{F}$ remains a directed WRAP solution.
Indeed, for every directed link $\vec{\link{l}}\in \vec{F}$, there is a link $\link{f}\in F$ with $\vec{\link{l}} \in W_{\link{f}}$.
Thus because the witness set $W_{\link{f}}$ contains only shadows of the link $\link{f}$, all 2-cuts covered by $\vec{\link{l}}\in \vec{F}$ are also covered by $\link{f}\in F$. 
In other words, one can think of a link $\link{f}\in F$ being responsible for the cuts in $\mathcal{C}_G$ for which the directed links in its witness set $W_{\link{f}}$ are responsible.

Then we iteratively apply local search steps as follows.
We first carefully choose a $4\lceil \sfrac{4}{\epsilon}\rceil$-thin component $K\subseteq L$ and add it to $F$.
Next, we remove $\Drop_{\vec{F}}(K)$ from $\vec{F}$.
We also remove the links in $\Drop_{\vec{F}}(K)$ from the witness sets of links in $F$; then we again have $\vec{F} = \bigcup_{\link{f}\in F} W_{\link{f}}$.
For every link $\link{f} = \{u,v\}$ in the component $K$, we let $W_{\link{f}} \coloneqq \{(u,v),(v,u)\}$ and add the links $(u,v)$ and $(v,u)$ to $\vec{F}$.
Then $\vec{F}$ is again a directed WRAP solution by the definition of the link set $\Drop_{\vec{F}}(K)$ we removed from $\vec{F}$ (\cref{def:Drop}).
Next, we use \cref{lem:shortening_efficiently} to delete and shorten links in $\vec{F}$ to make sure that $\vec{F}$ becomes a non-shortenable directed WRAP solution.
If the witness set of a link $\link{f}\in F$ became empty, we remove it from $F$.
See \cref{fig:local_search} for an example.

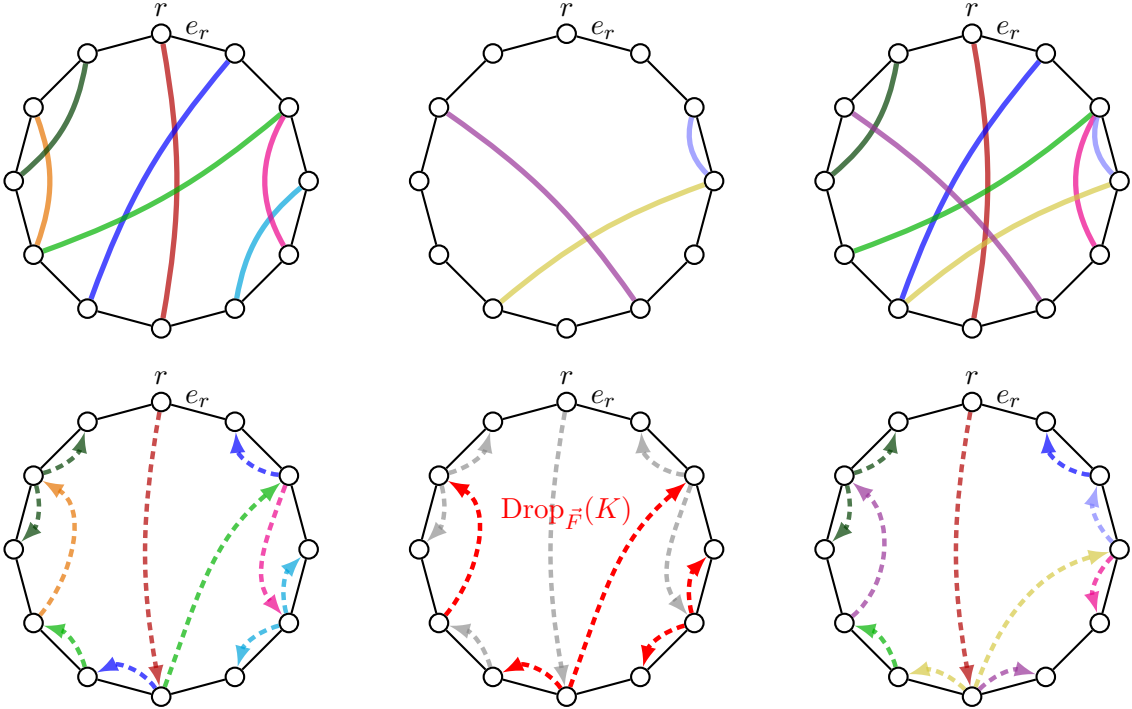
\begin{figure}[!ht]
\begin{center}
\begin{tikzpicture}[scale=0.98,
ns/.style={thick,draw=black,fill=white,circle,minimum size=7,inner sep=2pt},
es/.style={thick},
lks/.style={line width=1.7pt, blue, densely dashed},
dlks/.style={lks, -latex},
ts/.style={every node/.append style={font=\scriptsize}}
]

\def\rad{2}
\def\num{12}

\node () at (0,8.2) {$F$ and $\vec{F}$ before local search step};
\node () at (5.5,8.1) {$K$ and $\Drop_{\vec{F}}(K)$};
\node () at (11,8.2) {$F$ and $\vec{F}$ after local search step};

\begin{scope}[shift={(0,5)}]
\begin{scope}[every node/.style={ns}]
\foreach \i in {1,...,\num} {
  \pgfmathsetmacro\r{90+(\i-1)*360/\num}
  \node (\i) at (\r:\rad) {};
}
\end{scope}
\path (\num) to node[above=-2pt] {$e_r$} (1);
\node[above=3pt] at (1) {$r$};
\begin{scope}[es]
\foreach \i in {1,...,\num} {
\pgfmathtruncatemacro\j{1+mod(\i,\num)}
\draw (\i) -- (\j);
}
\end{scope}
\begin{scope}[lks,solid,line width=2pt,relative,black,opacity=0.7]
\draw[red!70!black] (1) to[bend left=10] (7);
\draw[blue] (6) to[bend left=10] (12);
\draw[green!70!black] (5) to[bend right=10] (11);
\draw[orange!90!black] (3) to[out=20,in=160] (5);
\draw[green!25!black] (2) to[out=20,in=160] (4);
\draw[cyan!90!black] (8) to[out=20,in=160] (10);
\draw[magenta] (9) to[bend left] (11);
\end{scope}
\end{scope}

\begin{scope}[shift={(0,0)}]
\begin{scope}[every node/.style={ns}]
\foreach \i in {1,...,\num} {
  \pgfmathsetmacro\r{90+(\i-1)*360/\num}
  \node (\i) at (\r:\rad) {};
}
\end{scope}
\path (\num) to node[above=-2pt] {$e_r$} (1);
\node[above=3pt] at (1) {$r$};
\begin{scope}[es]
\foreach \i in {1,...,\num} {
\pgfmathtruncatemacro\j{1+mod(\i,\num)}
\draw (\i) -- (\j);
}
\end{scope}
\begin{scope}[dlks,opacity=0.7]
\draw[red!70!black] (1) to[out=-100, in=100] (7);
\draw[blue] (7) to[bend right=40] (6);
\draw[green!70!black] (7) to[out=70,in=-140] (11);
\draw[green!70!black] (6) to[bend right=30] (5);
\draw[orange!90!black] (5) to[out=50,in=-30] (3);
\draw[green!25!black] (3) to[bend right=30] (2);
\draw[green!25!black] (3) to[bend left=30] (4);
\draw[blue] (11) to[bend left=40] (12);
\draw[magenta] (11) to[out=-110, in=130] (9);
\draw[cyan!90!black] (9) to[bend left=30] (10);
\draw[cyan!90!black] (9) to[bend right=30] (8);
\end{scope}
\end{scope}

\begin{scope}[shift={(5.5,5)}]
\begin{scope}[every node/.style={ns}]
\foreach \i in {1,...,\num} {
  \pgfmathsetmacro\r{90+(\i-1)*360/\num}
  \node (\i) at (\r:\rad) {};
}
\end{scope}
\path (\num) to node[above=-2pt] {$e_r$} (1);
\node[above=3pt] at (1) {$r$};
\begin{scope}[es]
\foreach \i in {1,...,\num} {
\pgfmathtruncatemacro\j{1+mod(\i,\num)}
\draw (\i) -- (\j);
}
\end{scope}
\begin{scope}[lks,solid,line width=2pt,relative,black,opacity=0.7]
\draw[yellow!80!black] (6) to[bend left=10] (10);
\draw[violet!80!white] (3) to[bend left=10] (8);
\draw[blue!50!white] (10) to[bend left=30] (11);
\end{scope}
\end{scope}

\begin{scope}[shift={(5.5,0)}]
\begin{scope}[every node/.style={ns}]
\foreach \i in {1,...,\num} {
  \pgfmathsetmacro\r{90+(\i-1)*360/\num}
  \node (\i) at (\r:\rad) {};
}
\end{scope}
\path (\num) to node[above=-2pt] {$e_r$} (1);
\node[above=3pt] at (1) {$r$};
\begin{scope}[es]
\foreach \i in {1,...,\num} {
\pgfmathtruncatemacro\j{1+mod(\i,\num)}
\draw (\i) -- (\j);
}
\end{scope}
\begin{scope}[dlks,gray,opacity=0.6]
\draw (1) to[out=-100, in=100] (7);
\draw[red, opacity=1] (7) to[bend right=40] (6);
\draw[red, opacity=1] (7) to[out=70,in=-140] (11);
\draw (6) to[bend right=30] (5);
\draw[red, opacity=1] (5) to[out=50,in=-30] (3);
\draw (3) to[bend right=30] (2);
\draw (3) to[bend left=30] (4);
\draw (11) to[bend left=40] (12);
\draw (11) to[out=-110, in=130] (9);
\draw[red, opacity=1] (9) to[bend left=30] (10);
\draw[red, opacity=1] (9) to[bend right=30] (8);
\end{scope}
\node[red] () at (0,0.5) {$\Drop_{\vec{F}}(K)$};
\end{scope}

\begin{scope}[shift={(11,5)}]
\begin{scope}[every node/.style={ns}]
\foreach \i in {1,...,\num} {
  \pgfmathsetmacro\r{90+(\i-1)*360/\num}
  \node (\i) at (\r:\rad) {};
}
\end{scope}
\path (\num) to node[above=-2pt] {$e_r$} (1);
\node[above=3pt] at (1) {$r$};
\begin{scope}[es]
\foreach \i in {1,...,\num} {
\pgfmathtruncatemacro\j{1+mod(\i,\num)}
\draw (\i) -- (\j);
}
\end{scope}
\begin{scope}[lks,solid,line width=2pt,relative,black,opacity=0.7]
\draw[red!70!black] (1) to[bend left=10] (7);
\draw[blue] (6) to[bend left=10] (12);
\draw[green!70!black] (5) to[bend right=10] (11);
\draw[green!25!black] (2) to[out=20,in=160] (4);
\draw[magenta] (9) to[bend left] (11);

\draw[yellow!80!black] (6) to[bend left=10] (10);
\draw[violet!80!white] (3) to[bend left=10] (8);
\draw[blue!50!white] (10) to[bend left=30] (11);
\end{scope}
\end{scope}

\begin{scope}[shift={(11,0)}]
\begin{scope}[every node/.style={ns}]
\foreach \i in {1,...,\num} {
  \pgfmathsetmacro\r{90+(\i-1)*360/\num}
  \node (\i) at (\r:\rad) {};
}
\end{scope}
\path (\num) to node[above=-2pt] {$e_r$} (1);
\node[above=3pt] at (1) {$r$};
\begin{scope}[es]
\foreach \i in {1,...,\num} {
\pgfmathtruncatemacro\j{1+mod(\i,\num)}
\draw (\i) -- (\j);
}
\end{scope}
\begin{scope}[dlks,opacity=0.7]
\draw[red!70!black] (1) to[out=-100, in=100] (7);
\draw[green!70!black] (6) to[bend right=30] (5);
\draw[green!25!black] (3) to[bend right=30] (2);
\draw[green!25!black] (3) to[bend left=30] (4);
\draw[blue] (11) to[bend left=40] (12);
\draw[magenta] (10) to[bend right=30] (9);

\draw[yellow!80!black] (7) to[out=70, in=-165] (10);
\draw[yellow!80!black] (7) to[bend right=30] (6);
\draw[violet!80!white] (5) to[bend right=50] (3);
\draw[violet!80!white] (7) to[bend left=30] (8);
\draw[blue!50!white] (10) to[bend left=30] (11);
\end{scope}
\end{scope}

\end{tikzpicture}
 \end{center}
\caption{An example of a local search step.
The left column shows a WRAP solution $F$ (top) and a corresponding non-shortenable directed WRAP solution $\vec{F}$ (bottom). 
For a link $\link{f}\in F$, the directed links in the witness set $W_{\link{f}}$ are shown in the same color as $\link{f}$.
The middle column shows a component $K$ (top) and again the directed WRAP solution $\vec{F}$, where $\Drop_{\vec{F}}(K)$ is highlighted in red (bottom).
The right column shows the WRAP solution $F$ (top) and the directed WRAP solution $\vec{F}$ (bottom) after the local search step.
This directed solution arises from the original directed WRAP solution $\vec{F}$ (bottom left) by removing the links in $\Drop_{\vec{F}}(K)$, adding the directed links $(u,v)$ and $(v,u)$ for all $\{u,v\}\in K$ and applying \cref{lem:shortening_efficiently} to obtain a non-shortenable directed WRAP solution.
Because the new directed WRAP solution $\vec{F}$ contains no orange and cyan links anymore, i.e., the witness sets of the orange and the cyan link became empty, these two links are removed from the solution $F$ while the links in the component $K$ are added to $F$.
\label{fig:local_search}}
\end{figure}

In order to measure progress of the algorithm, we use a carefully chosen potential function. 
We will not use the cost $c(F)$ of our current WRAP solution as a progress measure because it does not reflect the progress we make by dropping only one of  two links in the witness set $W_{\link{f}}$ of a link~$\link{f}$.
However, dropping one of the elements of $W_{\link{f}}$ will make it easier to remove the link $\link{f}$ from the solution $F$ in later iterations and should thus intuitively be considered as progress.
Instead of the cost $c(F)$ we work with the potential function
\[
\Phi(F) \coloneqq \sum_{\link{f}\in F: |W_{\link{f}}|=2} \frac{3}{2} \cdot c(\link{f}) + \sum_{\link{f}\in F: |W_{\link{f}}|=1} c(\link{f}).
\]
If we remove a single link $\vec{\link{l}}$ both from $\vec{F}$ and from the unique witness $W_{\link{f}}$ (with $\link{f}\in F$) that contains $\vec{\link{l}}$, then the potential $\Phi(F)$ decreases by 
\begin{equation}\label{eq:definition-auxiliary-cost}
\overline{c}\left(\vec{\link{l}}\right) \coloneqq 
\begin{cases}
\frac{1}{2}\cdot c(\link{f}) & \text{ if } |W_{\link{f}}|=2, \\
 c(\link{f}) & \text{ if } |W_{\link{f}}|=1.
\end{cases}
\end{equation}
Because the potential increases by at most $\frac{3}{2} \cdot c(\link{f})$ when we add a link $\link{f}$ to the solution $F$, in every local search step we choose a component $K$ that maximizes $\overline{c}(\Drop_{\vec{F}}(K)) - 1.5 \cdot c(K)$ among all $4\lceil \sfrac{4}{\epsilon} \rceil$-thin link sets.
In order to find such a component $K$ we will use the dynamic program from \cref{sec:dp} (\cref{thm:dp-main}).
To analyze the approximation ratio of the local search algorithm, we observe that the definition of the cost function $\overline{c}$ implies $\overline{c}(\vec{F})=c(F)$.

\begin{lemma}\label{lem:cost-distributed}
Let $\vec{F} = \bigcupp_{\link{f}} W_{\link{f}}$ be the disjoint union of the witness sets and let $\overline{c}$ be the cost function defined in \eqref{eq:definition-auxiliary-cost}.
Then $\overline{c}(\vec{F}) = c(F)$.
\end{lemma}
\begin{proof}
\[
c(F) = \sum_{\link{f}\in F} c(\link{f}) \ =\ \sum_{\link{f}\in F} \sum_{\vec{\link{l}}\in W_{\link{f}}} \overline{c}(\vec{\link{l}}) \ =\ \sum_{\vec{\link{l}}\in \vec{F}} \overline{c}(\vec{\link{l}}) =  \overline{c}(\vec{F}).
\]
\end{proof}

The decomposition theorem (\cref{thm:ov_decomposition-theorem}) together with \cref{lem:cost-distributed} imply that there exists a partition $\Kscr$ of $\OPT$ into $4\lceil \sfrac{4}{\epsilon} \rceil$-thin components such that
\[
 \sum_{K\in \Kscr} \overline{c}(\Drop_{\vec{F}}(K)) \ \ge\ (1-\sfrac{\epsilon}{4}) \cdot \overline{c}(\vec{F}) \ =\ (1-\sfrac{\epsilon}{4}) \cdot c(F).
\]
Then, if $(1-\sfrac{\epsilon}{4})\cdot c(F) > \frac{3}{2}\cdot c(\OPT)$, we have 
\[
\sum_{K\in \Kscr} \overline{c}(\Drop_{\vec{F}}(K)) > \frac{3}{2} \cdot c(\OPT) \ =\ \sum_{K\in \Kscr} \frac{3}{2} \cdot c(K).
\]
Thus, because all components in $\Kscr$ are $4\lceil \sfrac{4}{\epsilon} \rceil$-thin, there exists a $4\lceil \sfrac{4}{\epsilon} \rceil$-thin link set $K$ with $\overline{c}(\Drop_{\vec{F}}(K)) - 1.5 \cdot c(K) > 0$.
Choosing any such component in a local search step leads to a decrease of the potential $\Phi(F)$.
In order to achieve a polynomial runtime, we will stop the algorithm as soon as the potential function $\Phi(F)$ does not decrease by a sufficiently large factor.
\cref{algo:local_search} provides a full description of the local search algorithm.

\begin{algorithm2e}[H]
\KwIn{A rooted WRAP instance $(G=(V,E),L,c, r, e_r)$.}
\KwOut{A WRAP solution $F\subseteq L$ with $c(F) \le (1.5+\epsilon)\cdot c(\OPT)$.}
\vspace*{2mm}
\begin{enumerate}[label=\arabic*.,ref=\arabic*,rightmargin=7mm]\itemsep6pt
\item\label{item:initialize} Let $F\subseteq L$ be an arbitrary solution for the given WRAP instance. \\
  Set the witness sets to be $W_{\{u,v\}} \coloneqq \{ (u,v), (v,u) \}$ for all $\{u,v\}\in F$. \\
 Apply \cref{lem:shortening_efficiently} to $\vec{F}= \bigcupp_{\link{f} \in F} W_{\link{f}}$.\\
  If, for some link $\link{l}\in F$, the witness set $W_{\link{f}}$ became empty, remove $\link{f}$ from $F$.
\item\label{item:local_step} Iterate the following while $\Phi(F)$ decreases by at least a factor of $\left( 1- \frac{\epsilon}{6\cdot|V|}\right)$ per iteration.
\begin{itemize}\itemsep1pt
\item
Compute a $4\lceil \sfrac{4}{\epsilon} \rceil$-thin link set $K\subseteq L$ maximizing
$\overline{c}(\Drop_{\vec{F}}(K)) - 1.5 \cdot c(K)$, where $\vec{F}=\ \bigcupp_{\link{f} \in F} W_{\link{f}}$ (by applying \cref{thm:dp-main}).
\item Replace the witness set $W_{\link{f}}$ by $W_{\link{f}} \setminus \Drop_{\vec{F}}(K)$ for all $\link{f} \in F$.
\item Add $K$ to $F$ and set $W_{\{u,v\}} \coloneqq \{ (u,v), (v,u) \}$ for all $\{u,v\} \in K$.
\item Apply Lemma~\ref{lem:shortening_efficiently} to $\vec{F}= \bigcupp_{\link{f} \in F} W_{\link{f}}$.
\item If, for some link $\link{f} \in F$, the witness set $W_{\link{f}}$ became empty, remove $\link{f}$ from $F$.
\end{itemize}
\item Return $F$.
\end{enumerate}
\caption{Local search algorithm for WRAP}\label{algo:local_search}
\end{algorithm2e}

To analyze \cref{algo:local_search}, we first observe that it indeed returns a WRAP solution.

\begin{lemma}\label{lem:algReturnsWRAPSol}
Before and after each iteration of \cref{algo:local_search},
the link set $F$ is a WRAP solution and $\vec{F}=\bigcupp_{\link{f} \in F} W_{\link{f}}$ is a non-shortenable directed WRAP solution.
In particular, when the algorithm terminates, it returns a WRAP solution.
\end{lemma}
\begin{proof}
The directed link set $\vec{F}$ remains a directed WRAP solution by the definition of $\Drop_{\vec{F}}(K)$ and because the directed links added to $\vec{F}$ cover the same cuts as the links in the component $K$.
Moreover, the application of \cref{lem:shortening_efficiently} ensures that $\vec{F}$ is a non-shortenable directed WRAP solution at the end of each iteration.

For each link $\link{f}\in F$, we maintain the invariant that all directed links of $W_{\link{f}}$ are shortenings of $\link{f}$.
Because $\vec{F}=\bigcup_{\link{f}\in F} W_{\link{f}}$ is a directed WRAP solution, this implies that $F$ is a WRAP solution.
\end{proof}

Next, we bound the number of iterations of \cref{algo:local_search}.

\begin{lemma}\label{lem:bound_iterations_local_search}
\cref{algo:local_search} terminates after at most 
\[
\ln\!\left(\frac{1.5\cdot c(F_0)}{c(\OPT)}\right)\cdot \frac{6|V|}{\epsilon}
\]
iterations, where $F_0\subseteq L$ is the initial WRAP solution computed in Step~\ref{item:initialize} of \cref{algo:local_search}.
\end{lemma}
\begin{proof}
At the beginning of \cref{algo:local_search} we have $\Phi(F)=\Phi(F_0)\le 1.5 \cdot c(F_0)$.
Moreover, we have $\Phi(F)\ge c(F)\ge c(\OPT)$ throughout the algorithm.
By the stopping criterion of the while-loop, the potential $\Phi(F)$ decreases by a factor of at least $\left(1- \frac{\epsilon}{6\cdot |V|}\right)$ in every iteration.
This implies that the number of iterations is at most
\begin{equation*}
\log_{(1-\sfrac{\epsilon}{(6|V|)})^{-1}} \left(\frac{1.5\cdot c(F_0)}{c(\OPT)}\right) \ =\ \ln\left(\frac{1.5\cdot c(F_0)}{c(\OPT)}\right)\cdot \frac{1}{-\ln(1-\sfrac{\epsilon}{(6|V|)})}\ \le\ \ln\left(\frac{1.5\cdot c(F_0)}{c(\OPT)}\right)\cdot \frac{6|V|}{\epsilon}\enspace,
\end{equation*}
where we used $\ln(1+x) \le x$ for $x> -1$.
\end{proof} 

In order to prove the desired approximation ratio, we now provide a lower bound on the decrease of the potential function in a single iteration.
This lower bound will allow us to prove that the algorithm continues to make progress until $F$ is a $(1.5+\epsilon)$-approximation.

\begin{lemma}\label{lem:potential_decrease}
If we select a component $K\subseteq L$ in Step~\ref{item:local_step} of \cref{algo:local_search}, then $\Phi(F)$ decreases by at least $\overline{c}(\Drop_{\vec{F}}(K)) - 1.5 \cdot c(K)$ in this iteration.
\end{lemma}
\begin{proof}
First, we show that when removing $\Drop_{\vec{F}}(K)$ from $\vec{F}$ and from all witness sets, the potential $\Phi(F)$ decreases by at least $\overline{c}(\Drop_{\vec{F}}(K))$.
This follows from the observation that
\begin{equation*}
\Phi_{\link{f}} \coloneqq 
\begin{cases}
\frac{3}{2}\cdot c(\link{f}) & \text{ if } |W_{\link{f}}|=2, \\
 c(\link{f}) & \text{ if } |W_{\link{f}}|=1
\end{cases}
\end{equation*}
decreases for every link by at least 
\[
\overline{c}\left(W_{\link{f}} \cap \Drop_{\vec{F}}(K)\right) = \begin{cases}
c(\link{f}) & \text{ if }|W_{\link{f}}|=|W_{\link{f}} \cap \Drop_{\vec{F}}(K)|, \\
\frac{1}{2} \cdot c(\link{f}) & \text{ if } |W_{\link{f}}|=2 \text{ and }|W_{\link{f}} \cap \Drop_{\vec{F}}(K)| = 1, \\
0 & \text{ if } |W_{\link{f}} \cap \Drop_{\vec{F}}(K)| = 0,
\end{cases}
\]
 which can be seen from a case distinction over the possible sizes of $|W_{\link{f}}|$ and  $|W_{\link{f}} \cap \Drop_{\vec{F}}(K)|$.
Finally, we observe that adding the component $K$ to $F$ increases the potential $\Phi(F)$ by at most $1.5 \cdot c(K)$.
\end{proof}

\begin{lemma}\label{lem:good_improvement_exists}
In every iteration of \cref{algo:local_search}, there exists a $4\lceil\sfrac{4}{\epsilon}\rceil$-thin component $K\subseteq L$ such that
\begin{equation}\label{eq:local_search_bound_improvement}
\overline{c}(\Drop_{\vec{F}}(K)) - 1.5 \cdot c(K)\ \ge\ \frac{1}{|V|} \cdot \left(\left(1-\frac{\epsilon}{4}\right) \cdot c(F) - 1.5 \cdot c(\OPT)\right).
\end{equation}
\end{lemma}
\begin{proof}
We apply \cref{thm:ov_decomposition-theorem} to $\vec{F}$ with cost function $\overline{c}$ to obtain a partition $\Kscr$ of $\OPT$ into $4\lceil\sfrac{4}{\epsilon}\rceil$-thin components $K\subseteq L$ such that
 \[
\sum_{K\in \Kscr} \overline{c}(\Drop_{\vec{F}}(K))\ \ge\ \left(1-\frac{\epsilon}{4}\right) \cdot \overline{c}(\vec{F})\ =\ \left(1-\frac{\epsilon}{4}\right) \cdot c(F),
\]
where the equation follows from \cref{lem:cost-distributed}.
Because $\sum_{K\in \Kscr} c(K) = c(\OPT)$, we obtain 
\begin{align*}
\max_{K\in\Kscr} \Big(\overline{c}(\Drop_{\vec{F}}(K))-1.5 \cdot c(K)\Big)\ \ge&\ \frac{1}{|\Kscr|}\sum_{K\in \Kscr} \Big(\overline{c}(\Drop_{\vec{F}}(K))-1.5\cdot c(K)\Big)\\
 \ge&\ \frac{1}{|V|}\cdot\left(\left(1-\frac{\epsilon}{4}\right) \cdot c(F) - 1.5\cdot c(\OPT)\right),
\end{align*}
where we used $|\Kscr| \le |V|$.
\end{proof}

Finally, we prove that \cref{algo:local_search} is indeed a $(1.5+\epsilon)$-approximation algorithm for WRAP, which together with \cref{lem:reduction-to-WRAP} implies our main result (\cref{thm:main}).

\begin{theorem} Let $0 < \epsilon \le 0.5$.
Then \cref{algo:local_search} is a $(1.5+\epsilon)$-approximation algorithm for WRAP.
\end{theorem}
\begin{proof}
By \cref{lem:bound_iterations_local_search}, the number of iterations of \cref{algo:local_search} is polynomially bounded and, by \cref{thm:dp-main}, a single iteration takes only polynomial time.
Let $F$ be the link set returned by \cref{algo:local_search}, which is a WRAP solution by \cref{lem:algReturnsWRAPSol}.
Then the stopping criterion of the while-loop in \cref{algo:local_search} implies together with \cref{lem:potential_decrease} and \cref{lem:good_improvement_exists} that
\[
\frac{1}{|V|} \cdot \left(\left(1-\frac{\epsilon}{4}\right) \cdot c(F) - \frac{3}{2} \cdot c(\OPT)\right)\ <\ \frac{\epsilon}{6\cdot |V|} \cdot \Phi(F)\ \le\  \frac{\epsilon}{4\cdot |V|} \cdot c(F) .
\]
Hence, $(1- \sfrac{\epsilon}{2}) \cdot c(F) < 1.5 \cdot c(\OPT)$, which for $\epsilon \le 0.5$ implies $c(F) < (1.5 + \epsilon) \cdot c(\OPT)$.
\end{proof} 
 \section*{Acknowledgments}

We are thankful to Micha\l{} Pilipczuk for helpful discussions about Circle Graphs and questions related to them.
Moreover, we are grateful to the Hausdorff Institute for Mathematics (HIM) and the organizers of the 2021 HIM program ``Discrete Optimization'' during which fruitful discussions related to this work have taken place. 
\printbibliography

\end{document}